\documentclass[letterpaper,12pt,oneside,reqno]{amsart}
\usepackage[utf8]{inputenc}%
\usepackage[english]{babel}%
\usepackage{amsmath,amssymb,amsthm,amsfonts}%
\usepackage{MnSymbol}
\usepackage{skak}
\usepackage{hyperref}%
\usepackage{graphicx}
\usepackage{enumerate}
\usepackage[mathscr]{euscript}
\usepackage{color,tikz}
\usepackage[DIV14]{typearea}
\usepackage[width=.9\textwidth]{caption}
\allowdisplaybreaks%
\numberwithin{equation}{section}%

\usepackage{tikz}
\usetikzlibrary{decorations.markings, arrows}

\newcommand{\Z}{\mathbb{Z}}
\newcommand{\C}{\mathbb{C}}
\newcommand{\R}{\mathbb{R}}
\DeclareMathOperator{\E}{\mathbb{E}}
\renewcommand{\i}{\mathbf{i}}
\DeclareMathOperator{\Prob}{\mathsf{Prob}}

\newcommand{\al}{\alpha}

\newcommand{\la}{\lambda}
\newcommand{\La}{\Lambda}
\newcommand{\be}{\beta}

\newcommand{\ga}{\gamma}
\newcommand{\ka}{\kappa}

\DeclareMathOperator*{\Res}{\mathrm{Res}}

\newcommand{\de}{\delta}

\newcommand{\HT}{\mathfrak{h}}

\renewcommand{\a}{\widetilde a}
\renewcommand{\b}{\widetilde b}
\renewcommand{\c}{\widetilde c}
\renewcommand{\d}{\widetilde d}
\newcommand{\h}{\widetilde h}
\newcommand{\Dnorm}{D^{\mathrm{norm}}}
\newcommand{\Bnorm}{B^{\mathrm{norm}}}
\newcommand{\Bstoch}{B^{\mathrm{stoch}}}
\newcommand{\astoch}{a^{\mathrm{stoch}}}
\newcommand{\bstoch}{b^{\mathrm{stoch}}}
\newcommand{\cstoch}{c^{\mathrm{stoch}}}
\newcommand{\dstoch}{d^{\mathrm{stoch}}}
\newcommand{\Wstoch}{W^{\mathrm{stoch}}}
\newcommand{\tWstoch}{\widetilde{W}^{\mathrm{stoch}}}
\newcommand{\tW}{\widetilde W}
\newcommand{\w}{\ensuremath{\mathrm{weight}}}
\newcommand{\rw}{\ensuremath{\mathrm{r.weight}}}

\newcommand{\all}{\mathbf{+}}
\newcommand{\nul}{\zugzwang}
\newcommand{\rr}{\mathbf{--}}
\newcommand{\uu}{\mathbf{|}}
\newcommand{\ur}{\mathbf{ \lefthalfcap }}
\newcommand{\ru}{\mathbf{ \righthalfcup }}

\renewcommand{\o}{\mathcal{O}}
\newcommand{\inv}{\mathsf{inv}}
\newcommand{\erf}{\mathrm{erf}}
\newcommand{\erfc}{\mathrm{erfc}}

\newtheorem{proposition}{Proposition}[section]
\newtheorem{lemma}[proposition]{Lemma}
\newtheorem{corollary}[proposition]{Corollary}
\newtheorem{theorem}[proposition]{Theorem}
\theoremstyle{definition}
\newtheorem{definition}[proposition]{Definition}
\newtheorem{remark}[proposition]{Remark}
\newtheorem{example}[proposition]{Example}

\begin{document}

\title{Symmetric elliptic functions, IRF models, and dynamic exclusion processes}

\author[A. Borodin]{Alexei Borodin}
\address{Department of Mathematics,
Massachusetts Institute of Technology,
77 Massachusetts ave.,
Cambridge, MA 02139, USA\newline
Institute for Information Transmission Problems, Bolshoy Karetny per. 19, Moscow, 127994, Russia}
\email{borodin@math.mit.edu}

\date{\today}


\begin{abstract} We introduce stochastic Interaction-Round-a-Face (IRF) models that are related to representations of the elliptic quantum group $E_{\tau,\eta}(sl_2)$. For stochasic IRF models in a quadrant, we evaluate averages for a broad family of observables that can be viewed as higher analogs of $q$-moments of the height function for the stochastic (higher spin) six vertex models. 
	
In a certain limit, the stochastic IRF models degenerate to (1+1)d interacting particle systems that we call dynamic ASEP and SSEP; their jump rates depend on local values of the height function. For the step initial condition, we evaluate averages of observables for them as well, and use those to investigate one-point asymptotics of the dynamic SSEP. 

The construction and proofs are based on remarkable properties (branching and Pieri rules, Cauchy identities) of a (seemingly new) family of symmetric elliptic functions that arise as matrix elements in an infinite volume limit of the algebraic Bethe ansatz for $E_{\tau,\eta}(sl_2)$. 
\end{abstract}

\maketitle

\setcounter{tocdepth}{2}
\tableofcontents
\setcounter{tocdepth}{2}

\section{Introduction} 

\noindent\textbf{Preface.}
Yang-Baxter integrability has been a central theme in mathematical, statistical, and quantum physics for more than half a century, see, e.~g., the volume \cite{YB-volume} for a collection of foundational works. This article is concerned with two of its very recent and closely related applications to: 

(a) theory of symmetric functions, with new families of symmetric functions being introduced and new summation identities for older ones being proved, see \cite{Bor, WZ, M, MS, BP, BP-lect, GGW};

(b) deriving new exact formulas for averages of observables in two-dimensional integrable lattice models and (1+1)-dimensional random growth models, and using those to study large scale and time asymptotics \cite{CP, BP, BP-lect, AB, A, B-Schur, BO, BBW, D}. 

All the above cited papers deal with the $R$-matrix for the (higher spin) six 
vertex model and its degenerations (in other words, with representations of the 
affine quantum group $U_q(\widehat{sl}_2)$ and their limits). Certain progress 
has been achieved for the $U_q(\widehat{sl}_n)$ case as well, with new Markov 
chains introduced via the corresponding $R$-matrices \cite{KMMO, BM} and a 
duality functional for them provided in \cite{Kuan}.

The goal of this work is to climb higher in the hierarchy and to extend some of the recent progress from vertex models to the so-called \emph{Interaction-Round-a-Face (IRF) models}, also known as \emph{face} and \emph{solid-on-solid (SOS) models}. The corresponding $R$-matrices satisfy a face version of the star-triangle relation also known as the dynamical Yang-Baxter equation. 

The IRF models were originally introduced by Baxter \cite{Bax-IRF} as a tool to analyze the eight vertex model, but they quickly became a subject on their own, see, e.~g., \cite{Bax, YB-volume} and references therein. We will only be concerned with the $sl_2$ case, the basic instance of which is due to the original work \cite{Bax-IRF} and is often called the \emph{eight vertex SOS model}, and whose fused versions were introduced and extensively studied in \cite{DJMO, DJKMO1, DJKMO2}. 

The IRF models were framed in a representation theoretic language by Felder \cite{Felder}, who introduced the concept of an \emph{elliptic quantum group}. The simplest instance is the quantum elliptic group $E_{\tau,\eta}(sl_2)$; the corresponding representation theory and the algebraic Bethe ansatz were developed by Felder-Varchenko in \cite{FV-reps, FV-ABA}. The latter works are a convenient starting point for us; we extensively use their notations and results.

As the first step, we take the wavefunctions constructed in the algebraic Bethe ansatz framework of $E_{\tau,\eta}(sl_2)$ and consider their infinite volume limit. Such a limit makes Bethe equations on spectral parameters unnecessary, and the resulting objects are symmetric elliptic functions in their spectral parameters. In a suitable limit, they degenerate to the symmetric rational functions that were studied in \cite{Bor, BP}, which are, in their turn, generalizations of the classical (symmetric multivariate) Hall-Littlewood polynomials.  

Using the Yang-Baxter integrability (equivalently, the IRF star-triangle relation, or the dynamic Yang-Baxter equation, or the commutation relations of the elliptic quantum group), we show that these symmetric elliptic functions satisfy versions of the Pieri, Cauchy, and skew-Cauchy identities from the theory of symmetric functions. This requires an introduction of a family of dual symmetric elliptic functions, as well as skew variants of both families (the term `skew' is used by a direct analogy with the theory of symmetric functions, where skew Schur, Hall-Littlewood, etc. polynomials are broadly used). We also prove a kind of orthogonality relations for our functions. 
All the results are derived from an infinite volume limit of the $E_{\tau,\eta}(sl_2)$ algebraic Bethe ansatz; this is largely parallel to what was done in \cite{Bor, BP} for $U_q(\widehat{sl}_2)$. 

In order to proceed to a construction of stochastic models, we need to find a simplification of the Cauchy identities mentioned above. We achieve it by taking a trigonometric limit of our symmetric elliptic functions and finding a specialization of the dual family that simplifies it dramatically. Even though this brings us closer to the objects studied in \cite{BP}, we still have one extra parameter in the game. We call it the \emph{dynamic} parameter; it is responsible for the word `dynamic' in the dynamic Yang-Baxter equation that is behind our model. 

It is this simplification that tells us how to renormalize the weights of the associated IRF model to make it \emph{stochastic}, which basically means that the new $R$-matrix is stochastic, and the random states of the model can be constructed by a Markovian procedure involving only sampling of independent Bernoulli random variables. The existence of such a renormalization that also preserves the integrability of the model does not seem \emph{a priori} evident, and we view the construction of the stochastic IRF model as the first main result of the present paper.  

We show how our stochastic IRF model can be degenerated to certain interacting particle systems in (1+1)d that we call \emph{dynamic exclusion processes}; they are one-parameter generalizations of the usual exclusion processes with rates of jumps depending on the local value of the height function (which is closely related to the dynamic parameter). These dynamic exclusion processes appear to be new, and we expect them to enjoy the same degree of integrabiliy as the usual exclusion processes. 

Our second main result is an explicit evaluation of averages of certain observables for stochastic IRF models in a quadrant. The observables are quite simple yet rather nontrivial, and we discover them by analyzing the Cauchy identities and orthogonality relations for the appropriate trigonometric limits of the symmetric elliptic functions discussed above. 

Our method is similar to that of \cite{BP}, although things get more complicated because of the presence of the dynamic parameter. Also, the fact that this approach should bring a result is not \emph{a priori} obvious, and our first naive attempts to apply the philosophy of \cite{BP} were actually unsuccessful. Surprisingly, the averages we compute turn out to be independent of the dynamic parameter, which suggests that there may be a smarter way of evaluating them. 

In the non-dynamic case of $U_q(\widehat{sl}_2)$, the similar observables were always the $q$-moments of a suitably defined height function. The corresponding formulas have already been extensively used to obtain asymptotics in a large variety of probabilistic systems, see the references in part (b) of the first paragraph above and references therein. We hope that the new observables will allow to replicate at least some of that success, but at the level of the IRF models or dynamic interacting particle systems. 

As a step towards that goal, we take one of the simplest new models --- the dynamic Symmetric Simple Exclusion Process (SSEP) --- and use our observables to analyze its one-point asymptotics at large times for the step initial condition. We discover different growth exponents than for the usual SSEP and, surprisingly, the lack of the deterministic limiting height profile at large times. 

Let us now describe some of our results in more detail. 

\smallskip

\noindent \textbf{The stochastic IRF model.} The model has a number of (generally speaking, complex) parameters that we denote as: $\eta$  -- this is a Planck constant type parameter associated with the quantization parameter $q=e^{-4\pi\i\eta}$; $\la_0$ -- the overall shift of the dynamic parameters, $\{z_x,\Lambda_x\}$ -- the inhomogeneity and spin (or highest weight) parameters associated to columns marked by the $x$-coordinate; $\{w_y\}$ -- the spectral parameters associated to rows marked by the $y$-coordinate that can also be thought of as inhomogeneity parameters. We do not have the spin parameters that vary row-by-row (or rather we have the spin parameter 1 for all rows); adding those is possible but leads to substantial increase in complexity that we want avoid in this work, cf. \cite[Section 5]{BP} in the vertex model case. 

We will only consider stochastic IRF models in a quadrant with specific boundary conditions, although other domains and boundary conditions are certainly possible.

The states of our stochastic IRF model in the quadrant $\R_{\ge 0}^2$ are functions $\mathcal{F}:\Z_{\ge 0}^2\to \C$; we will call the value $\mathcal F(x,y)$  the \emph{filling} or the dynamic parameter of the unit square $[x,x+1]\times[y,y+1]$ in the quadrant. 

The fillings of two squares that share a horizontal edge must differ by an $\pm 2\eta$, while the fillings of two squares with shared vertical edge with coordinate $x$ must differ by an element of $2\eta(\La_x-2 \Z_{\ge 0})$, unless $\La_x=I+m/(2\eta)$ for integral $I\ge 0$ and $m$.\footnote{The `$m$' in this formula plays no role as the weights of our model do not depend on it, cf. \eqref{eq:intro1} below; in what follows we will simply take $m=0$.}  If $\La_x$ is of that form, the difference must be an element of the small set $2\eta\cdot\{-I,-I+2,\dots,I-2,I\}$, which shrinks to $\pm 2\eta$ in the spin $\frac 12$ case $I=1$.

The probability measure on the $\mathcal{F}$'s that we are interested in, is a limit of its projections -- finite probability distributions on the cylindric sets defined by prescribing all the fillings $\mathcal{F}(x,y)$ inside finite growing squares $[0,L]^2$. The weight of such a finite filling is defined as the product of Boltzmann factors over all $2\times 2$ squares (called \emph{plaquettes}) of the form $[x,x+2]\times[y,y+2]$ with $x,y\in \{0,1,\dots,M-2\}$. The above conditions on $\mathcal{F}$'s imply that each plaquette must be of one of the types pictured in Figure \ref{fig:intro1} with varying $\la\in\C$ and $k\in\Z_{\ge 0}$. The spin $\frac 12$ plaquettes are pictured in Figure \ref{fig:intro2}. 

\begin{figure}
	\includegraphics[scale=0.75]{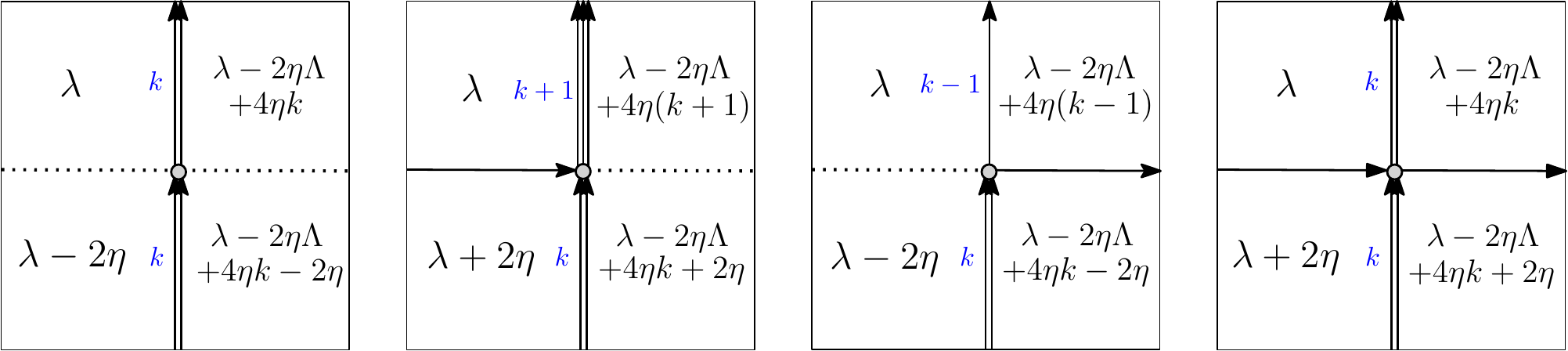}
	\caption{The four types of the IRF plaquettes.}
	\label{fig:intro1}
\end{figure}

\begin{figure}
	\includegraphics[scale=1]{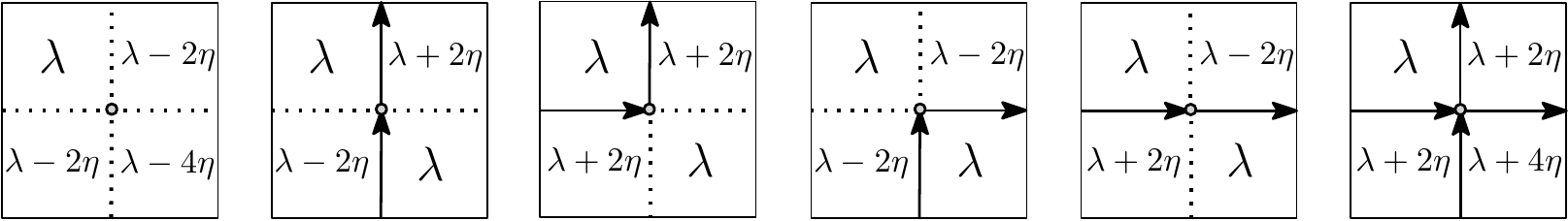}
	\caption{The spin $\frac 12$ IRF vertices.}
	\label{fig:intro2}
\end{figure}

If one wants to encode only the information about the differences between fillings rather than their actual values, there is another convenient graphical way of doing that, which is also included in Figures \ref{fig:intro1} and \ref{fig:intro2}. Namely, one draws a suitable number of upward and rightward arrows along the inner lattice edges of the plaquette as illustrated there. When such arrows are combined for all the plaquettes, they form directed lattice paths that move in the up-right direction. The procedure of passing from fillings to up-right paths projects the IRF model states to those for the higher spin six vertex model, and our IRF models actually converge to those vertex models as $\la_0\to -\i\infty$, see Section \ref{ss:hs6v} below. 

We will impose boundary conditions on $\mathcal{F}$ along the sides of the quadrant by requiring that 
$$
\mathcal{F}(x,0)=\la_0-2\eta(\La_1+\ldots+\La_x),\qquad \mathcal{F}(0,y)=\la_0-2\eta y,\qquad x,y\ge 0.
$$
This is equivalent to saying that $\mathcal{F}(0,0)=\la_0$, and the up-right paths from the previous paragraph enter the quadrant as rightward arrows joining $(0,y)$ and $(1,y)$ for all $y\ge 1$. 

Finally, to complete the definition of the stochastic IRF model, let us give the the Boltzmann weights of the four plaquette types of Figure \ref{fig:intro1}. They are, in the same order as in Figure \ref{fig:intro1},
\begin{equation}\label{eq:intro1}
\begin{aligned}
\astoch_k(\la;w)&=\frac{f(z-w+(\La+1-2k)\eta)}{f(z-w+(\La+1)\eta)}\frac{f(\la-2(\La+1-k)\eta)}{f(\la-2(\La+1-2k)\eta)}\,,\\
\bstoch_k(\la;w)&=\frac{f(-\la+z-w+(\La-1-2k)\eta)}{f(z-w+(\La+1)\eta)}\frac{f(2(k-\La)\eta)}{f(\la-2(\La-1-2k)\eta)}\,,\\
\cstoch_k(\la;w)&=\frac{f(\la+z-w-(\La+1-2k)\eta)}{f(z-w+(\La+1)\eta)}\frac{f(2k\eta)}{f(\la-2(\La+1-2k)\eta)}\,,\\
\dstoch_k(\la;w)&=\frac{f(z-w+(-\La+1+2k)\eta)}{f(z-w+(\La+1)\eta)}\frac{f(\la+2(k+1)\eta)}{f(\la-2(\La-1-2k)\eta)}\,,
\end{aligned}
\end{equation}
where $f(\zeta)\equiv \sin \pi \zeta$, $\la$ is the filling of the top left unit square of the considered plaquette, and $(z,\La,w)$ need to be specialized to $(z_x,\La_x,w_y)$ with $(x,y)\in\Z_{\ge 1}^2$ being the coordinates of the center of the plaquette. There are several ways to choose the parameters of the model so that these weights are actually nonnegative, and the probabilistic terminology we used above makes sense. However, our principal results are purely algebraic, they do not depend on this nonnegativity, thus we will not focus on it at the moment. 

The above plaquette weights satisfy the identities
\begin{equation}\label{eq:intro2}
\astoch_k(\la;w)+\cstoch_k(\la;w)\equiv 1,\qquad \bstoch_k(\la;w)+\dstoch_k(\la;w)\equiv 1,
\end{equation}
which can be interpreted as follows: If we fix the fillings of the three unit squares along the left and bottom sides of a plaquette, then there are two possibilities for the filling of its top right square, and the sum of weights for these two possibilities is always equal to 1. This makes it possible to construct random states inductively -- one starts from the left and bottom borders of the quadrant where the values of $\mathcal{F}$ are prescribed by the boundary conditions, and gradually moves inside the quadrant by filling unit squares whose left and bottom neighbors have already been filled, with each step requiring a Bernoulli random variable with biases given by one of the identities \eqref{eq:intro2}. We also see that the measures on fillings of $[0,L]$ defined by taking products of such plaquette weights are consistent for different $L$'s. 

\smallskip

\noindent \textbf{The observables.} For any $(x,y)\in \Z_{\ge 1}$ and an IRF configuration in the quadrant, define $\HT(x,y)$ as the number of up-right paths that pass through or below the vertex with coordinates $(x,y)$. Our boundary conditions imply $0\le\HT(x,y)\le y$, and it suffices to know the fillings of the finitely many unit squares in the rectangle $[0,x]\times[0,y+1]$ to know what $\HT(x,y)$ is. We call $\HT(x,N)$ the \emph{height function} for our IRF model.

Let us introduce another observable at $(x,y)\in\Z_{\ge 1}^2$ that is closely related to the height function and is defined by 
\begin{equation}\label{eq:intro3}
\o(x,y)=e^{2\pi\i\la_0}q^{\HT(x,y)}+q^{N-\HT(x,N)-\La_{[1,x)}},
\end{equation}
where we use the notations $q=e^{-4\pi\i\eta}$ and $\La_{[1,x)}=\La_1+\ldots+\La_{x-1}$. A known $\o(x,y)$ yields two possible values of $q^{\HT(x,y)}$ by solving a quadratic equation. In many cases this allows to reconstruct $\HT(x,y)$ uniquely as one of the roots would not satisfy natural inequalities imposed by the model. 
\begin{theorem}\label{th:intro1} 
	For any $n,y\ge 1$ and $x_1\ge x_2\ge\dots\ge x_n\ge 1$, we have
	\begin{multline}\label{eq:intro4}
	\frac{1}{(e^{2\pi\i\la_0};q)_n}\E \left[\prod_{k=1}^{n} 
	\left(q^{y-\La_{[1,x_{k})}}+e^{2\pi\i\la_0}q^{2(k-1)} -q^{k-1}\cdot\o(x_{k},y)\right)\right]=e^{-2\pi\i\eta\,\bigl(\frac{n(n-1)}2+ny-\sum_{i=1}^n\La_{[1,x_i)})\bigr)}\\ \times\oint\dots\oint \prod_{1\le i<j\le n} \frac{f(v_i-v_j)}{f(v_i-v_j+2\eta)}\prod_{i=1}^n\bigg(\prod_{j=1}^{x_i-1}\frac{f(v_i-p_j)}{f(v_i-q_j)}\prod_{k=1}^y \frac{f(v_i-w_k-2\eta)}{f(v_i-w_k)}\bigg)\,{dv_i},
	\end{multline}
	where the integration contours are positively oriented loops around $\{w_k\}_{1\le k\le y}$ (alternatively, the integral can be replaced by the  sum of possible residues at $v_i=w_k$ for all $i$ and $k$). 
\end{theorem} 

Observe that the right-hand side of \eqref{eq:intro4} is independent of $\la_0$, which is quite surprising. Taking the limit $\la_0\to -\i\infty$ turns the left-hand side into $\E_{\mathrm{h.s.6\,v.m.}} \left[\prod_{k=0}^{n-1} \left(q^{\HT(x_{k+1},y)}-q^{k}\right)\right]$, where the expectation is with respect to a stochastic higher spin six vertex in the same quadrant. In this limit the equality \eqref{eq:intro4} was actually proved as \cite[Lemma 9.11]{BP}, and since then two other proofs have appeared in \cite{OP, BBW}.

If all the plaquette weights are nonnegative and the observables \eqref{eq:intro3} are real, the averages \eqref{eq:intro4} uniquely determine the joint distribution of the observables $\o(x_1,y)$, $\dots$, $\o(x_n,y)$ for any $x_1,\dots,x_n,y\ge 1$. 

\smallskip

\noindent \textbf{Dynamic exclusion processes.} In the spin $\frac 12$ case (equivalently, $\La_j\equiv 1$) there is a way of specializing the parameters of the stochastic IRF model in such a way that the weights of the 2nd and 5th plaquettes in Figure \ref{fig:intro2} become infinitesimally small. Looking at a finite neighborhood of the diagonal of the quadrant infinitely far from the origin reveals a novel interacting system that we call the \emph{dynamic Asymmetric Simple Exclusion Process} (ASEP, for short).\footnote{A similar limit of the stochastic six vertex model leads to the usual ASEP, see \cite[Section 6.5]{BP}, \cite{Agg-ASEP}, and reference therein.}

The dynamic ASEP is a continuous time Markov chain on the state space of integer-valued sequences $\{s_x\}_{x\in \Z}$ subject to the condition $s_{x+1}-s_x\in\{-1,1\}$ for any $x\in\Z$. Alternatively, the increments $\{s_{x+1}-s_{x}\}_{x\in\Z}$ can be encoded by a particle configuration in $\Z+\frac 12$, where we say that there is a particle that resides at $x+\frac 12$ if and only if 
$s_{x+1}-s_{x}=-1$. We will only consider the initial condition $\{s_x=|x|\}_{x\in\Z}$ at time $t=0$, which corresponds to particles filling up the negative semi-axis $\{-\frac 12,-\frac 32,-\frac 52, \dots\}$ -- the so-called step initial condition. 

The dynamical ASEP depends on two parameters $q,\al\in\R$ (here $q=e^{-4\pi\i\eta}$ is the same as before, and $\al=-e^{-2\pi\i\la_0}$ in terms of the earlier notation). Its elementary jumps are independent and have exponential waiting times with variable rates. These jumps can be of two kinds, and their form and the corresponding rates are

\smallskip

$s_x\mapsto (s_x-2)$\quad with rate\quad $\dfrac{q(1+\alpha q^{-s_x})}{1+\alpha q^{-s_x+1}}$
\qquad and \qquad
$s_x\mapsto (s_x+2)$ \quad with rate\quad $\dfrac{1+\al q^{-s_x}}{1+\al  q^{-s_x-1}}\,$,
\smallskip

\noindent where $x\in\Z$ is arbitrary. The parameters $q$ and $\al$ are assumed to be such that the above rates are always nonnegative, which is the case, for example, for $q,\al > 0$. Note that $\al=0$ yields the usual ASEP with constant jump rates $q$ and 1. 

Taking a further limit of $q\to 1$ leads to the \emph{dynamic Symmetric Simple Exclusion Process} (SSEP) defined as follows. It is again a continuous time Markov chain on the state space of integer-valued sequences $\{s_x\}_{x\in \Z}$ subject to the condition $s_{x+1}-s_x\in\{-1,1\}$ for any $x\in\Z$, and it depends on a single parameter $\la\in\R$. Its elementary jumps are independent and have exponential waiting times with variable rates. These jumps can be of two kinds, and their form and the corresponding rates are

\smallskip

$s_x\mapsto (s_x-2)$\quad with rate\quad $\dfrac{s_x-\la}{s_x-1-\la}$
\qquad and \qquad
$s_x\mapsto (s_x+2)$ \quad with rate\quad $\dfrac{s_x-\la}{s_x+1-\la}$
\smallskip

\noindent for arbitrary $x\in\Z$. These rates are positive, for example, if we take $\la<0$ and the initial condition $s_x(t=0)=|x|$. 

The observable \eqref{eq:intro3} for the dynamical ASEP and SSEP takes the following form:
$$
\o_{ASEP}(x,t)=-\alpha^{-1} q^{\frac{s_x(t)-x}2}+q^{\frac{-s_x(t)-x}2},\qquad
\o_{SSEP}(x,t)={\frac{s_x(t)-x}2}\left(\frac{s_x(t)+x}2-\la\right),
$$
and Theorem \ref{th:intro1} leads to the following

\begin{corollary}
	Consider the dynamic ASEP and SSEP as defined above, with the initial condition $s_x(0)\equiv |x|$. Then, for any $t\ge 0$, $n\ge 1$, and $x_1\ge x_2\ge\dots\ge x_n$, we have
	\begin{multline*}
	\frac{1}{(-\al^{-1};q)_n}\E_{\mathrm{dynamic\, ASEP\,at\,time\,}t} \left[\prod_{k=1}^{n} 
	\left(q^{-x_{k}}-\al^{-1} q^{2(k-1)} -q^{k-1}\cdot\o_{ASEP}(x_{k},t)\right)\right]=	\frac{q^{\frac{n(n-1)}2}}{(2\pi\i)^n}\\ \times
	\oint
	\ldots
	\oint
	\prod_{1\le i<j\le n}\frac{y_i-y_j}{y_i-qy_j}
	\prod_{i=1}^{n}\bigg(
	\bigg(\frac{1-y_i}{1-q y_i}\bigg)^{x_i}
	\exp\bigg\{\frac{(1-q)^{2}y_i}{(1-y_i)(1-q y_i)}\,t\bigg\}
	\bigg)\frac{d y_i}{y_i}\,,
	\end{multline*}
\end{corollary}
where the integration contours are small positively oriented loops around $1$, and
\begin{multline*}
\frac{1}{(-\la)_n}\E_{\mathrm{dynamic\, SSEP\,at\,time\,}t} \left[\prod_{k=1}^{n} 
\left((k-1)(k-1-\la+x_{k})-\o_{SSEP}(x_{k},t)\right)\right]\\=\oint
\ldots
\oint
\prod_{1\le i<j\le n}\frac{v_i-v_j}{v_i-v_j+1}
\prod_{i=1}^{n}\bigg(
\bigg(\frac{v_i}{v_i-1}\bigg)^{x_i}
\exp\bigg\{\frac{t}{v_i(v_i-1)}\bigg\}
\bigg)\frac{d v_i}{2\pi\i}\,,
\end{multline*}
where the integration contours are small positively oriented loops around $0$.
\smallskip

These formulas can be used for large time asymptotic analysis, and we show that the last formula implies, for example, the following claim: 
Consider the dynamic SSEP with initial condition $s_x(0)\equiv |x|$ and parameter $\la<0$, and fix $\tau\ge 0$ and $\chi\in\R$. Then 
$$
\lim_{L\to \infty} L^{-\frac 14}\cdot s_{L^{-\frac 14}\chi}(L\tau)=\sqrt{Z+\chi^2},
$$ 
where $Z$ is a $\Gamma(-\la,4\sqrt{\tau/\pi})$-distributed random variable.\footnote{The {gamma distributions} form a two parameter family $\Gamma(a,b)$ of absolutely continuous probability measures on $\R_{> 0}$ with densities
	$p_{a,b}(x)={b^a x^{a-1} e^{-bx}}/{\Gamma(a)}.$} 

\smallskip

\noindent \textbf{Symmetric elliptic functions.} Let us return to the setting of the stochastic IRF model in the quadrant. One fact that shows its integrability (which is much simpler than Theorem \ref{th:intro1}) is the following: One can explicitly compute the probability that the $M\ge 1$ up-right paths started from the left boundary of the quadrant in rows $1,\dots,M$ intersect the line $y=M+\frac 12$ at prescribed locations $(\mu_1,M+\frac 12),(\mu_2,M+\frac 12),\dots (\mu_M,M+\frac 12)$ for some fixed $\mu_1\ge \dots\ge \mu_M\ge 1$. Up to simple elementary factors (that can be obtained from \eqref{eq:prob=B-stoch}, \eqref{eq:B-stoch}, \eqref{eq:D-rho-expl} below), this probability is equal to 
\begin{multline*} 
B_\mu(\la;w_1,\dots,w_M):=\frac{(-1)^M(f(2\eta))^{M}}{
	\prod_{i=0}^{M-1} f(\la+2\eta i)}\cdot \prod_{i\ge 
	0}\prod_{j=1}^{m_i}\frac{f(2\eta)}{f(2\eta j)}
\\
\times \sum_{\sigma\in S_M}\sigma\left[\prod_{1\le i<j\le M} 
\frac{f(w_i-w_j-2\eta)}{f(w_i-w_j)}\cdot \prod_{i=1}^M 
\phi_{\mu_i}(w_i)f\bigl(\la+w_i-q_{\mu_i}+2\eta(2(M-i)+1-\La_{[0,\mu_i-1)})\bigr)
\right],
\end{multline*}
where $\la=\la_0-2\eta(M-\La_0)$, $\La_{[0,x)}=\La_0+\dots+\La_{x-1}$, 
$$
p_j=z_j+(1-\La_j)\eta,\quad q_j=z_j+(1+\La_j)\eta,\qquad \phi_k(w)=\frac 1{f(w-q_k)} \prod_{i=0}^{k-1}\frac{f(w-p_i)}{f(w-q_j)}
$$
for all suitable index values, and $S_M$ is the symmetric group on $M$ symbols with its elements $\sigma$ permuting the variables $w_1,\dots,w_M$ inside the brackets. 

These $B_\mu$'s form a remarkable family of symmetric functions. Let us lift them to an elliptic setting by replacing $f(\zeta)=\sin \pi\zeta$ by $f(\zeta)=\theta(\zeta,\tau)$ with some $\tau\in \C$, $\Im\tau>0$, and
$$
\theta(\zeta,\tau)=-\sum_{j\in\Z} e^{\pi \i(j+\frac 12)^2\tau +2\pi \i (j+\frac 12)(\zeta+\frac 12)},
$$
see the beginning of Section \ref{sc:prelim} for basic properties of this theta-function. The limit $\tau\to +\i\infty$ gives back the trigonometric case $f(\zeta)=\sin\pi\zeta$.

The resulting symmetric elliptic functions first appeared in \cite{FV-ABA} as an expression for the wavefunctions of the transfer matrices in the algebraic Bethe ansatz (ABA, for short) for the elliptic quantum group $E_{\tau,\eta}(sl_2)$. In \cite{FV-ABA}, the spectral parameter $w_1,\dots,w_M$ had to satisfy certain Bethe equations because the corresponding IRF model had periodic boundary conditions; for our purposes, $w_1,\dots,w_M$ should be considered as free indeterminates. 

The functions $B_\mu(\la;w_1,\dots,w_M)$ are also matrix elements of the traditional ABA $B$-operators acting on the highest weight vector (thus the notation). They are naturally included into a larger family of skew functions $B_{\mu/\nu}$ defined as matrix elements of the same operators acting on generic vectors. Similarly, one defines a `dual' family $D_{\mu/\nu}$ of matrix elements of the ABA $D$-operators, properly renormalized in an infinite volume limit. 

These families of symmetric elliptic functions satisfy a host of identities (of Pieri and Cauchy types) that are direct consequences of the Yang-Baxter integrability (or the commutation relations of $E_{\tau,\eta}(sl_2)$), and also certain orthogonality relations. This approach was explained in great amount of detail in \cite{BP, BP-lect} in the case of the higher spin six vertex model (which is a degeneration of what is being done here), and in this work we basically follow the same path. An interested reader is referred to \cite{BP, BP-lect} for a more detailed discussion. 
In the end, it is the wonderful properties of these symmetric functions that 
make the proofs of the results stated earlier possible.  

\smallskip

\noindent\textbf{Perspectives.}\ A few follow-ups of the present paper are 
currectly in progress. The fusion for the IRF models discussed above, as well 
as fused stochastic models and a dynamic $q$-TASEP will appear in 
\cite{Agg-in-progress}. A duality approach to the dynamic ASEP/SSEP and an 
alternative proof of Corollary 1.2 will appear in \cite{BC-in-progress}. 
Yet another upcoming work will focus on equilibrium measures for the dynamic
ASEP and their properties. 

\smallskip

\noindent\textbf{Acknowledgments.}\  Over the course of this work I have 
greatly benefited from multiple conversations with Amol Aggarwal, Ivan Corwin, 
Vadim Gorin, Leonid Petrov, Eric Rains, and Nicolai Reshetikhin; I am very 
grateful to them. This work was partially supported by the NSF grant DMS-1607901 
and by Fellowships of the Radcliffe Institute for Advanced Study and the Simons 
Foundation. 
\section{Preliminaries}\label{sc:prelim} We will substantially rely on two works by Felder and Varchenko \cite{FV-reps, FV-ABA}, and will also try to keep our notation consistent with theirs. The goal of this section is to introduce notations and to recall some results from \cite{FV-reps, FV-ABA} that we will need later on. 

There are two basic parameters in the game, $\eta,\tau\in\C\setminus \{0\}$, $\Im \tau>0$. The following theta-function plays a key role: 
\begin{equation}\label{eq:theta}
\theta(z,\tau)=-\sum_{j\in\Z} e^{\pi \i(j+\frac 12)^2\tau +2\pi \i (j+\frac 12)(z+\frac 12)}. 
\end{equation}

If $\tau=\i T\to+\i\infty$ then the main contribution to the series comes from $j=-1$ and $j=0$, and $\theta(z,\i T)\sim - e^{-\pi T/4}(e^{-\i\pi (z+\frac 12)}+e^{\i\pi (z+\frac 12)})=2e^{-\pi T/4}\sin(\pi z)$. All the key expressions below are invariant with respect to multiplying $\theta(z,\tau)$ by a constant, and as a consequence we can replace $\theta(z,\tau)$ in the above limit by $\sin(\pi z)$. We will use the notation $f(z)$ to denote both $\theta(z,\tau)$ and $\sin(\pi z)$, and will refer to the latter case as `trigonometric'. 

We have the periodicity relations \cite[p.~496 on top]{FV-ABA}
\begin{equation}\label{eq:period}
\theta(z+1,\tau)=-\theta(z,\tau)=\theta(-z,\tau), \qquad \theta(z+\tau,\tau)=-e^{-\i \pi (\tau+2z)} \theta(z,\tau). 
\end{equation}
Also, $C(\tau)\theta(z,\tau)=H(2Kz)\Theta(2Kz)$ with a constant $C(\tau)$, in terms of the classical Jacobi notation used in Baxter's book \cite{Bax}, cf. \cite[p.~494-495]{FV-ABA}.
The infinite product representations of the theta-functions, cf., e.g., \cite[Chapter 15]{Bax}, implies that $\theta(z,\tau)$ has a single simple zero in each fundamental parallelogram of $\C/(\Z+\tau \Z)$ (note that $\theta(0,\tau)=0$). 

Fix $\Lambda,z\in\C$. Following \cite[Section 4]{FV-reps}, consider operators $a(\lambda,w),b(\la,w),c(\la,w),d(\la,w)$ acting in the \emph{evaluation Verma module} $V_\La(z)=\mathrm{Span}(\{e_k\}_{k\ge 0})$ by 
\begin{equation}\label{eq:abcd}
\begin{gathered}
a(\la,w) e_k= \frac{f(z-w+(\La+1-2k)\eta)}{f(z-w+(\La+1)\eta)}\frac{f(\la+2k\eta)}{f(\la)}\,e_k,\\
b(\la,w) e_k=-\frac{f(-\la+z-w+(\La-1-2k)\eta)}{f(z-w+(\La+1)\eta)}\frac{f(2\eta)}{f(\la)}\,e_{k+1},\\
c(\la,w)e_k=-\frac{f(-\la-z+w+(\La+1-2k)\eta)}{f(z-w+(\La+1)\eta)}\frac{f(2(\La+1-k)\eta)}{f(\la)}\frac{f(2k\eta)}{f(2\eta)}\,e_{k-1},\\
d(\la,w)e_k=\frac{f(z-w+(-\La+1+2k)\eta)}{f(z-w+(\La+1)\eta)}\frac{f(\la-2(\La-k)\eta)}{f(\la)}\, e_k.
\end{gathered}
\end{equation}

This is a highest weight module of the elliptic quantum group $E_{\tau,\eta}(sl_2)$ with the highest weight vector $e_0$ that is killed by $c(\la,w)$  and is an eigenvector of $a(\la,w)$ (with eigenvalue 1) and $d(\la,w)$. If $\La=n+(m+l\tau)/2\eta$ for $n\ge 0,m,l$ all integers, then $c(\la,w)e_{n+1}=0$, and $\mathrm{Span}(\{e_k\}_{k\ge n+1})$ is a submodule; the corresponding quotient is a module of dimension $(n+1)$, see \cite[Theorem 3]{FV-reps}. 

The operators \eqref{eq:abcd} satisfy a host of (quadratic) commutation relations, see \cite[Section 2]{FV-reps}. We just need a few of them that are more convenient to state in terms of modified operators $\a(w),\b(w),\c(w),\d(w)$ that act in the space of functions $F:\C\to V_\La(z)$ that are 1-periodic: $F(\la+1)=F(\la)$. The definition of this action is, see \cite[Section 3]{FV-reps},
\begin{equation}\label{eq:tilde-abcd}
\begin{gathered}
(\a(w)(w)F)(\la)=a(\la,w)F(\la-2\eta), \qquad (\c(w)F)(\la)=c(\la,w)F(\la-2\eta),\\
(\b(w)(w)F)(\la)=b(\la,w)F(\la+2\eta), \qquad (\d(w)F)(\la)=d(\la,w)F(\la+2\eta).
\end{gathered}
\end{equation}

First, we will need the fact that operators $\a(w)$ with different values of $w$ commute between themselves: $\a(w_1)\a(w_2)=\a(w_2)\a(w_1)$, and similarly for $\b,\c,\d$.

Second, we will need the following less trivial commutation relations, see \cite[Section 3]{FV-reps}:
\begin{gather}
\b(w_1)\a(w_2)=\be(\la)\b(w_2)\a(w_1)+\de(\la)\a(w_2)\b(w_1),\label{eq:ba-comm}\\
\b(w_2)\d(w_1)=\de(\la-2\eta\h)\d(w_1)\b(w_2)+\ga(\la-2\eta\h)\b(w_1)\d(w_2),\label{eq:db-comm}\\
\a(w_2)\c(w_1)=\de(\la-2\eta\h)\c(w_1)\a(w_2)+\ga(\la-2\eta\h)\a(w_1)\c(w_2),\label{eq:ac-comm}
\end{gather}
where $\be,\ga,\de$ are functions of $w=w_1-w_2,\la,\eta,\tau$ (only the dependence on $\la$ is indicated) defined by
\begin{equation*} 
\be(\la)=\frac{f(-w-\la)f(2\eta)}{f(w-2\eta)f(\la)},\quad
\ga(\la)=\frac{f(w-\la)f(2\eta)}{f(w-2\eta)f(\la)},\quad 
\de(\la)=\frac{f(w)f(\la-2\eta)}{f(w-2\eta)f(\la)},
\end{equation*}
and $\h$ returns the \emph{weight} of the vector in $V_\La(z)$ that the expression with $\h$ is being applied to, according to $F(\h)e_k=F(\La-2k)e_k$ for any function $F:\C\to\C$ and $k\ge 0$ (that is, the weight of $e_k$ is defined to be $(\La-2k)$).  

There is a way of tensor multiplying modules (=spaces with an action of $a,b,c,d$) that preserves commutation relations, see \cite[bottom of p.~745]{FV-reps}. Namely, for $V_1\otimes V_2$ the action of $a,b,c,d$ can be read off as matrix elements of the following $2\times 2$ matrix product
\begin{equation}\label{eq:tensor}
\begin{bmatrix} a(\la,w)&b(\la,w)\\ c(\la,w)&d(\la,w)\end{bmatrix}=
\begin{bmatrix} a_2(\la-2\eta h^{(1)},w)&b_2(\la-2\la h^{(1)},w)\\ c_2(\la-2\eta h^{(1)},w)&d_2(\la-2\eta h^{(1)},w)\end{bmatrix}\begin{bmatrix} a_1(\la,w)&b_1(\la,w)\\ c_1(\la,w)&d_1(\la,w)\end{bmatrix},
\end{equation}
e.~g., $a(\la,w)=a_2(\la-2\eta h^{(1)},w)a_1(\la,w)+b_2(\la-2\la h^{(1)},w)c_1(\la,w)$, where the subscript in $a,b,c,d$ refers to the tensor factor they are acting in, and $h^{(1)}$ is the weight of the $V_1$-component \emph{after} the action of the right-most factor in $V_1$. 

The weights for vectors in $V_1\otimes V_2$ are defined by the usual rule --- the weight of the tensor product of two weight vectors is the sum of the weights of the factors. 

\section{Infinite volume limit of $b$ and $d$ operators}\label{sc:inf-volume} We would like to define an action of the operators $b$ and $d$ on \emph{finitary} vectors in an infinite tensor product $V_0\otimes V_1\otimes V_2\otimes \cdots$ of evaluation Verma modules, where `finitary' means that the vectors are different from $E_\varnothing:=e_0\otimes e_0\otimes \cdots\otimes e_0\otimes \cdots$ in finitely many positions. For such vectors we can use the notation 
$$
E_\mu = e_{m_0}\otimes e_{m_1}\otimes\cdots\otimes e_{m_k}\otimes \cdots\otimes e_0\otimes e_0\otimes\cdots
$$
for signatures\footnote{By a \emph{signature} we mean a weakly decreasing finite sequence of integers $\nu=(\nu_1\ge \nu_2\ge\dots\ge \nu_k)$; the term originates from the representation theory of classical Lie groups. We will only need signatures with nonnegative parts so we will tacitly make this assumption. We will also use the multiplicative notation of the form $\nu=0^{n_0}1^{n_1}2^{n_2}\cdots$, which means that $\nu$ has $n_0$ zero parts, $n_1$ parts equal to 1, etc.} $\mu=0^{m_0}1^{m_1}2^{m_2}\cdots$ (we will also use this notation for finite tensor products when the number of tensor factors is larger than the largest part of $\mu$). The procedure is very similar to what is described in \cite[Sections 4.2-4.3]{BP, BP-lect} (see also \cite[Section 4]{Bor}), and thus we will be brief.

Defining the action of $b(\la,w)$ causes no difficulties, we set
\begin{equation*}
b(\la,w)E_\mu = \sum_{\nu} B_{\nu/\mu}(\la;w) E_\nu,
\end{equation*}
with each coefficient $B_{\nu/\mu}(\la;w)$ evaluated in a finite tensor product $V_0\otimes V_1\otimes \cdots \otimes V_N$, where $N$ is larger than $\max(\mu_1,\nu_1)$. This definition is independent of $N$ because the action of $a,b,c,d$ is normalized so that $a(\la;w)e_0\equiv e_0$. 
By standard reasoning with lattice paths, $B_{\nu/\mu}(\la;w)=0$ unless $\nu$ and $\mu$ interlace:
$$
\nu_1\ge\mu_1\ge\nu_2\ge\mu_2\dots, \qquad \mathrm{notation}\  \nu\succ\mu\ \mathrm{or}\ \mu\prec\nu,
$$
and $\ell(\nu)$, which is the \emph{length} or the number of parts of $\nu$, is by one greater than $\ell(\mu)$.

We can also introduce several variables $w$ by 
\begin{multline}\label{eq:B}
\b(w_1)\b(w_2)\cdots \b(w_n) E_\mu= b(\la,w_1)b(\la+2\eta,w_2)\cdots b(\la+2\eta(n-1),w_n)E_\mu\\
=
\sum_\nu B_{\nu/\mu}(\la;w_1,\dots,w_n)E_\nu.
\end{multline}

The fact that $\b(w_j)$'s with different $j$ commute implies that the coefficients $B_{\nu/\mu}(\la;w_1,\dots,w_n)$ are symmetric in the $w_j$'s. This definition also readily implies the following \emph{branching rule} (note the shift of $\la$ in the second factor on the right):
\begin{equation}\label{eq:B-branching}
B_{\nu/\mu}(\la;u_1,\dots,u_k,v_1,\dots,v_l)=\sum_{\ka} B_{\nu/\ka}(\la;u_1,\dots,u_k)B_{\ka/\mu}(\la+2\eta k; v_1,\dots,v_l). 
\end{equation}

Defining the action of $d(\la,w)$ in the tensor product $V_{\La_0}(z_0)\otimes V_{\La_1}(z_1)\otimes\cdots$ requires more efforts. For example, naively acting on $E_\varnothing$ we obtain (using the fourth line of \eqref{eq:abcd})
\begin{multline}\label{eq:d-norm-empty}
d(\la,w) E_\varnothing = d_0(\la,w) e_0\otimes d_1(\la-2\eta \La_0)e_0\otimes d_2(\la-2\eta(\La_0+\La_1))e_0\otimes \cdots \\=
\frac{f(z_0-w+(-\La_0+1)\eta)}{f(z_0-w+(\La_0+1)\eta)}\frac{f(\la-2\La_0\eta)}{f(\la)}\cdot \frac{f(z_1-w+(-\La_1+1)\eta)}{f(z_1-w+(\La_1+1)\eta)}\frac{f(\la-2(\La_0+\La_1)\eta)}{f(\la-2\La_0)}\cdot\ldots\cdot E_\varnothing. 
\end{multline}

One sees factors of two types: (1) ratios of the form $f(z_i-\dots)/f(z_i-\dots)$, and, up to finitely many factors, exactly the same will occur in a similar application of $d(\la,w)$ to any $E_\nu$; (2) ratios involving $f(\la-\dots)$ that cancel telescopically, and for the first $(m+1)$ tensor factors yield $f(\la-2\La_{[0,m]})/f(\la)$,  where $\La_{[0,m]}$ is the shorthand for $\La_0+\ldots+\La_{m}$. 

In a similar evaluation of $d(\la,w)E_\nu$ with a general signature $\nu$, far enough out in the tensor product so that only $e_0$'s are participating, the factors of the first type will be just the same, while the factors of the second type will look like (starting from the contribution of $V_{\La_i}(z_i)$)
$$
\prod_{k\ge i}\frac{f(\la-2\eta(\La_{[0,k-1]}-2\ell(\nu))-2\eta\La_k)}{f(\la-2\eta(\La_{[0,k-1]}-2\ell(\nu)))}, 
$$
where $(\Lambda_{[0,k-1]}-2\ell(\nu))$ is the weight of the output in the first $k$ factors $V_{\La_0}(z_0)\otimes \cdots\otimes V_{\La_{k-1}}(z_{k-1})$. 
This product also telescopes, and its value with $k$ ranging over $i,i+1,\dots,m$ equals 
$$
\frac{f(\la-2\eta(\La_{[0,m]}-2\ell(\nu)))}{f(\la-2\eta(\La_{[0,i-1]}-2\ell(\nu)))}\,.
$$ 
When sending $m$ to infinity we only need to worry about the $m$-dependent numerator.  

The above discussion shows that we can define the action of $d(\la,w)$ in the infinite tensor product with the help of the following normalization: For $\nu=0^{n_0}1^{n_1}2^{n_2}\cdots$ we set 
\begin{multline}\label{eq:d-norm}
\overline{d} (\la,w) E_\nu\\ :=\lim_{m\to\infty} \prod_{i=0}^m \frac{f(z_i-w+(\La_i+1)\eta)}{f(z_i-w+(-\La_i+1)\eta)}\frac{1}{f(\la-2\eta(\La_{[0,m]}-2\ell(\nu)))}\, d(\la,w)(e_{n_0}\otimes e_{n_1}\otimes\cdots\otimes e_{n_m})\\
=:\sum_\mu D_{\nu/\mu}(\la;w) E_\mu. 
\end{multline}
As for $B_{\nu/\mu}$'s, the standard reasoning shows that $D_{\nu/\mu}(\la;w)=0$ unless $\nu\succ\mu$ and $\ell(\nu)=\ell(\mu)$. 

Similarly to \eqref{eq:B-branching}, we can define multivariate $D$'s via the following branching relation
\begin{equation}\label{eq:D-branching}
D_{\nu/\mu}(\la;u_1,\dots,u_k,v_1\dots,v_l)=\sum_\ka D_{\ka/\mu}(\la;u_1,\dots,u_k) D_{\nu/\ka}(\la+2\eta k;v_1,\dots,v_l).
\end{equation}
This is equivalent to applying 
$$
\d(w_1)\cdots \d(w_n)=d(\la,w_1)d(\la+2\eta,w_2)\cdots d(\la+2\eta(n-1),w_n)
$$ 
to $E_\nu$, normalizing each $d(\,\cdot\,,\,\cdot\,)$ as above, and denoting the coefficients of $E_\mu$ in the result as $D_{\nu/\mu}(\la;w_1,\dots,w_n)$. The commutativity of $\d$'s implies that $D_{\nu/\mu}(\la;w_1,\dots,w_n)$ is symmetric in the $w_j$'s for any $\la\in \C$ and signatures $\mu,\nu$. 

\section{Cauchy identities}\label{sc:cauchy} The material of this section is similar to \cite[Section 4.3]{BP, BP-lect}, see also \cite[Section 4]{Bor}. 
 
We start by taking the infinite volume limit (in the sense of the previous section) of the commutation relation \eqref{eq:db-comm}. The result takes the following form. 

\begin{proposition}[an elementary skew-Cauchy identity]\label{prop:skew-cauchy} Assume that for  complex parameters $\eta,\tau,\la,u,v,\{z_i\}_{i\ge 0},\{\La_i\}_{i\ge 0}$, and $n\in \Z_{\ge 1}$ we have
\begin{equation}\label{eq:convergence}
\lim_{m\to\infty} \frac{f(u-v+\la-2\eta (\La_{[0,m]}-2n))}{f(\la-2\eta(\La_{[0,m]}-2n))}\\  \prod_{i=0}^m \frac{f(z_i-u+(-\La_i+1)\eta)}{f(z_i-u+(\La_i+1)\eta)}\frac{f(z_i-v+(\La_i+1)\eta)}{f(z_i-v+(-\La_i+1)\eta)}=0.
\end{equation}
Then for any signatures $\mu$ and $\nu$ with $\ell(\nu)=n-1$ we have
\begin{equation}\label{eq:skew-cauchy}
\sum_\ka D_{\ka/\mu}(\la;v)B_{\ka/\nu}(\la+2\eta;u)=\frac{f(v-u-2\eta)}{f(v-u)}\sum_\rho B_{\mu/\rho}(\la;u)D_{\nu/\rho}(\la+2\eta;v),
\end{equation}
which, in particular, means that the series in the left-hand side converges (the sum in the right-hand side always has finitely many nonzero terms).  
\end{proposition} 
\begin{proof} We first use \eqref{eq:db-comm} to write
\begin{equation}\label{eq:db-comm-again}
d(\la,v)b(\la+2\eta,u)E_\nu=\frac{b(\la,u)d(\la+2\eta,v)}{{\de(\la-2\eta h)}}\, E_\nu-\frac{\ga(\la-2\eta h)b(\la,v)d(\la+2\eta,u)}{\de(\la-2\eta h)}\,E_\nu
\end{equation}
in $V_{\La_0}(z_0)\otimes\cdots\otimes V_{\La_m}(z_m)$. The value of $h$ is $h_m:=\La_{[0,m]}-2(\ell(\nu)+1)$ for $m$ large enough (because $d$'s do not change the weight while $b$'s decrease it by 2). 

We now want to take the limit as $m\to\infty$, and to that end we need to compare the normalization of the $d$-factors in the above relation. We shall see that our proposition corresponds to the limit of the first two terms of \eqref{eq:db-comm-again}, while the assumption \eqref{eq:convergence} corresponds to the normalized third term of that relation becoming negligible in the limit. 

Let us look at the $\la$-dependent factors first. The $\la$-dependent part of $1/(\de(\la-2\eta h_m))$ is $f(\la-2\eta h_m)/f(\la-2\eta-2\eta h_m)$, and the $\la$-dependent part of $\ga(\la-2\eta h_m)/({\de(\la-2\eta h)})$ is $f(u-v+\la-2\eta h_m)/f(\la-2\eta-2\eta h_m)$. Further, the $\la$-dependent parts of the normalization \eqref{eq:d-norm} for the $d$'s in \eqref{eq:db-comm-again} are
$$
d(\la,v)\leadsto f(\la-2\eta h_m),\quad d(\la+2\eta,u),d(\la+2\eta,v)\leadsto f(\la+2\eta-2\eta (h_m+2)).
$$
We see that the required normalization (by $f(\la-2\eta h_m)$) for the first two terms of \eqref{eq:db-comm-again} is the same, and the remaining factor in the third term is 
${f(u-v+\la-2\eta h_m)}/(f(\la-2\eta h_m))$.

Turning to $\la$-independent normalization factors, we again observe that they are the same for the first two terms of \eqref{eq:db-comm-again}. Collecting all the normalization factors for the third term and requiring that their product converges to zero as $m\to\infty$, we arrive at the desired relation. 
\end{proof} 

One can now iterate \eqref{eq:skew-cauchy} to produce a more general skew-Cauchy identity. 
\begin{corollary}[a general skew-Cauchy identity]\label{cr:skew-cauchy} Assume that  complex parameters $\eta,\tau,\la$, $\{u_i\}_{i=1}^k$, $\{v_j\}_{j=1}^l$, $\{z_i\}_{i\ge 0},\{\La_i\}_{i\ge 0}$ are such that \eqref{eq:convergence} holds for $u=u_i$, $v=v_j$ with any $i,j$, and also for any $n\in \Z_{\ge 1}$. Then for any signatures $\mu$ and $\nu$ we have
\begin{multline}\label{eq:general-skew-cauchy}
\sum_\ka D_{\ka/\mu}(\la;v_1,\dots,v_l)B_{\ka/\nu}(\la+2\eta l;u_1,\dots,u_k)\\ =\prod_{i=1}^k\prod_{j=1}^l\frac{f(v_j-u_i-2\eta)}{f(v_j-u_i)}\sum_\rho B_{\mu/\rho}(\la;u_1,\dots,u_k)D_{\nu/\rho}(\la+2\eta k;v_1,\dots,v_l).
\end{multline}
\end{corollary}

\begin{proof} A straightforward application of the branching rules \eqref{eq:B-branching}, \eqref{eq:D-branching} and the elementary skew-Cauchy identity \eqref{eq:skew-cauchy}. Alternatively, one could directly commute a few $\d$ operators through a few $\b$ operators 
using \eqref{eq:db-comm} and keep the surviving terms after the infinite volume limit. 
\end{proof}

Let us proceed to the non-skew functions and the corresponding Cauchy identities. 

Set
$$
B_\mu(\la;u_1,\dots,u_k)=B_{\mu/\varnothing}(\la;u_1,\dots;u_k),\qquad D_{\nu}(\la;v_1,\dots,v_l)=D_{\nu/0^N}(\la;v_1,\dots,v_l),
$$
for any signatures $\mu,\nu$ with $N=\ell(\nu)$. It is also convenient to introduce slightly renormalized versions of our symmetric functions defined by
\begin{equation}\label{eq:BD-norm}
\begin{gathered}
\Bnorm_{\nu/\mu}(\la;u_1,\dots,u_k)=\prod_{i=0}^{k-1} f(\la+2\eta i)\cdot B_{\nu/\mu}(\la;u_1,\dots,u_k),\\
\Dnorm_{\nu/\mu}(\la;v_1,\dots,v_l)=\prod_{j=0}^{l-1} f(\la+2\eta j)\cdot D_{\nu/\mu}(\la;v_1,\dots,v_l),
\end{gathered}
\end{equation}
and similarly for the non-skew functions. 

\begin{corollary}[Pieri rules and a Cauchy identity] Under the assumptions of Corollary \ref{cr:skew-cauchy}, we have
\begin{multline}\label{eq:pieri2}
\sum_\ka D_{\ka}(\la;v_1,\dots,v_l)B_{\ka/\nu}(\la+2\eta l;u)\\ =\frac{f(\la+u+(\La_0-1-2\ell(\nu))\eta)}{f(z-u+(\La_0+1)\eta)}\frac{f(2\eta)}{f(\la)}\,\prod_{j=1}^l\frac{f(v_j-u-2\eta)}{f(v_j-u)} D_{\nu}(\la+2\eta ;v_1,\dots,v_l),
\end{multline}
\begin{equation}
\sum_\ka \Dnorm _{\ka/\mu}(\la;v)\Bnorm _\ka(\la+2\eta;u_1,\dots,u_k)=\prod_{i=1}^k\frac{f(v-u_i-2\eta)}{f(v-u_i)}\Bnorm_\mu(\la;u_1,\dots,u_k),\label{eq:pieri}
\end{equation}
\begin{multline}
\sum_\ka \Dnorm _{\ka}(\la;v_1,\dots,v_l)\Bnorm _\ka(\la+2\eta l;u_1,\dots,u_k)\\=\prod_{i=1}^k\frac{f(2\eta)f(\la-z_0+u_i+\eta(-\La_0+2k-1))}{f(z_0-u_i+\eta(\La_0+1))}\cdot\prod_{i=1}^k\prod_{j=1}^l\frac{f(v_j-u_i-2\eta)}{f(v_j-u_i)}.
\label{eq:cauchy}
\end{multline}
\end{corollary}
\begin{remark} Equation \eqref{eq:pieri} resembles an eigenrelation with the vector $[\Bnorm_\ka]$ and the matrix $[\Dnorm_{\ka/\mu}]$. Note, however, the shift of the $\lambda$-parameter in $\Bnorm$'s between the left and the right-hand sides. The relation \eqref{eq:pieri2} is also of a similar form, but the lengths of the signatures $\ka$ and $\nu$ there are actually different by 1 (if the sum has at least one nonzero term). 
\end{remark}
\begin{proof} Start with \eqref{eq:general-skew-cauchy} and set $\nu=\varnothing$. This forces $\rho=\varnothing$, and it remains to evaluate $D_{\varnothing/\varnothing}(\la+2\eta k;v_1,\dots,v_l)$. This quantity is the coefficient if $E_\varnothing$ in $\overline{d}(\la+2\eta k,v_1)\overline{d}(\la+2\eta (k+1),v_2)\cdots\overline{d}(\la+2\eta (k+l-1),v_l) E_\varnothing$, which is easily computed via the definitions \eqref{eq:abcd} and \eqref{eq:d-norm}:
\begin{equation}\label{eq:D-empty}
D_{\varnothing/\varnothing}(\la+2\eta k;v_1,\dots,v_l)=\left(\prod_{i=k}^{k+l-1} f(\la+2\eta i)\right)^{-1}.
\end{equation}
Changing the notations to $\Bnorm$ and $\Dnorm$ yields \eqref{eq:pieri}. 

Further, set $\mu=0^k$ in \eqref{eq:pieri}. This will lead to \eqref{eq:cauchy}, we just need to evaluate 
$\Bnorm_{0^k}(\la;u_1,\dots,u_k)$. Without the superscript `norm' this is the product of the coefficients in the second line of \eqref{eq:abcd} evaluated at $\la$ taking values $\la,\la+2\eta,\dots,\la+2\eta(k-1)$ and $w$ taking values $u_1,\dots,u_k$ (the order of the $u_j$'s is irrelevant because of the commutativity of $\b$'s). Collecting the factors and taking into account \eqref{eq:BD-norm} gives the result. 

Finally, for proving \eqref{eq:pieri2} one sets $k=1$ and $\mu=0^{N+1}$ in \eqref{eq:general-skew-cauchy}, where $N=\ell(\nu)$. This leaves only one nonzero term in the right-hand side that corresponds to $\rho=0^N$. Evaluating $B_{0^{N+1}/0^N}(\la;u)$ using the second line of \eqref{eq:abcd} yields \eqref{eq:pieri2}.
\end{proof}

\section{Symmetrization formulas}\label{sc:symm}
The goal of this section is to give explicit formulas for the symmetric functions $B_\mu$ 
and $D_\mu$ defined in Section \ref{sc:inf-volume}. Our formula for $B_\mu$ and its proof essentially coincide with 
\cite[Theorem 5]{FV-ABA} (although we give a more detailed computation); the 
formula for $D_\mu$ appears to be new. Similar formulas in simpler setups can 
be found in \cite[Section 5]{Bor}, \cite[Section 4.5]{BP, BP-lect}.

We will continue using the shorthand 
$$
\La_{[0,i]}=\La_0+\dots+\La_i, \qquad \La_{[0,i)}=\La_0+\dots+\La_{i-1}. 
$$
For a signature $\nu=0^{n_0}1^{n_1}2^{n_2}\cdots$, we will use $n_{<k}$ to denote $\sum_{i<k} n_i$, and similarly with the $\le,>,\ge$ signs. For a set of variables $\{u_i\}$, we will denote by $u_I$ the subset of them with indices ranging over $I$: $u_I=\{u_i\mid i\in I\}$. 

We also set 
\begin{equation}\label{eq:pq}
p_j=z_j+(1-\La_j)\eta,\quad q_j=z_j+(1+\La_j)\eta\qquad \Leftrightarrow\qquad q_j-p_j=2\eta\La_j,\quad p_j+q_j=2\eta+2z_j, 
\end{equation}
and introduce functions
\begin{equation}\label{eq:phi-psi}
\phi_k(u)=\frac 1{f(u-q_k)} \prod_{i=0}^{k-1}\frac{f(u-p_i)}{f(u-q_j)}\,,\qquad
\psi_l(v)=\frac 1{f(v-p_l)} \prod_{j=0}^{l-1}\frac{f(v-q_j)}{f(v-p_j)}\,.
\end{equation}
\begin{theorem}\label{th:symm} {(i)} For any signature $\mu=(\mu_1\ge\mu_2\ge 
\dots\ge \mu_M\ge 0)=0^{m_0}1^{m_1}2^{m_2}\cdots$, we have
\begin{equation}\label{eq:B-symm}
\begin{gathered} 
B_\mu(\la;u_1,\dots,u_M)=\frac{(-1)^M(f(2\eta))^{M}}{
\prod_{i=0}^{M-1} f(\la+2\eta i)}\cdot \prod_{i\ge 
0}\prod_{j=1}^{m_i}\frac{f(2\eta)}{f(2\eta j)}
\\
\times \sum_{\sigma\in S_M}\sigma\left(\prod_{1\le i<j\le M} 
\frac{f(u_i-u_j-2\eta)}{f(u_i-u_j)}\cdot \prod_{i=1}^M 
\phi_{\mu_i}(u_i)f(\la+u_i-q_{\mu_i}+2\eta+4\eta(M-i)-2\eta\La_{[0,\mu_i)})
\right),
\end{gathered}
\end{equation}
where $S_M$ is the symmetric group on $M$ symbols, and its elements $\sigma$ permute the variables $u_1,\dots,u_M$ in the expression inside the parentheses. 

(ii) For any signature $\nu=(\nu_1\ge \nu_2\ge \dots\ge \nu_N\ge 0)=0^{n_0}1^{n_1}2^{n_2}\cdots$ and $n\ge 1$, we have, with the notation $\widetilde \la=\la+2\eta n$,
\begin{equation}\label{eq:D-symm}
\begin{gathered}
D_\nu(\la;v_1,\dots,v_n)=\frac{(f(2\eta))^{N-n_0}}{\prod_{i=0}^{n-1}f(\la+2\eta i)}\cdot
\prod_{i=N}^{n+n_0-1}\frac{f(\la+2\eta(i-\La_0))}{f(\la+2\eta(i+n_0-\La_0))}\\
\times \prod_{i\ge 1} 
\prod_{j=0}^{n_i-1}\frac{f(2\eta(\La_i-j))}{f(\widetilde\la 
+2\eta(2n_{<i}+n_i+j-\La_{[0,i]}))f(\widetilde\la 
+2\eta(2n_{<i}+1+j-\La_{[0,i)}))}\\
\times \sum_{\substack{I\subset\{1,\dots,n\}\\ |I|=N-n_0}} \prod_{i\in I}
\frac{f(\la+v_i-q_0+2\eta N)}{f(v_i-q_0)}\prod_{j\notin 
I}\frac{f(v_j-p_0-2\eta n_0)}{f(v_j-p_0)}\prod_{\substack{i\in I\\j\notin 
I}}\frac{f(v_j-v_i-2\eta)}{f(v_j-v_i)}\\
\times \sum_{\substack{\sigma:\{1,\dots,N-n_0\}\to I\\ \sigma\ \text{is a 
bijection}}}\sigma\left(\prod_{i<j}\frac{f(v_i-v_j+2\eta)}{f(v_i-v_j)}
\prod_{i=1}^{N-n_0}\psi_{\nu_i}(v_i)f(\widetilde\la 
-v_i+p_{\nu_i}+2\eta+4\eta(N-i)-2\eta\La_{[0,\nu_i)})\right).
\end{gathered}
\end{equation}
\end{theorem}
\begin{proof} Our argument for (i) is essentially the proof of 
\cite[Theorem 5]{FV-ABA}, but we provide more computational details. Similar 
arguments for both (i) and (ii) were given in the proofs of (simpler) 
\cite[Theorem 4.14]{BP} and \cite[Theorem 4.12]{BP-lect}, thus we will 
skip over the parts of the proof that are very similar to what was described 
there. 

Let us start with (i). In the tensor product of two highest weight modules 
$V_0\otimes V_1$ with
highest weight vectors denoted by $e_0$, we obtain (cf., e.g., 
\cite[(4.27)]{BP-lect})
\begin{multline}\label{eq:bbb}
\b(u_1)\cdots \b(u_M) e_0\otimes e_0 \\=\sum_{K\subset\{1,\dots,M\}} C_K
\left(\prod_{k\in K}\b_0(u_k)\prod_{l\notin K}\d_0(u_l)\, e_0\right)\otimes
\left(\Gamma(-2\eta h^{(0)})\left[\prod_{l\notin K}\b_1(u_l)\prod_{k\in 
K}\a_1(u_l)\right]e_0\right),
\end{multline}
where the notation $\Gamma(-2\eta h^{(0)})[\cdots]$ means that the parameter 
$\la$ in the expression inside the brackets needs to be shifted by $(-2\eta)\cdot
(\text{weight of the vector in the 0th tensor 
component})$.\footnote{As before, the weight is taken \emph{after} the 
application of the operators in the 0th component, thus in this case it is equal to $\La_0-2|K|$.} The coefficients 
$C_K=C_K(u_1,\dots,u_M)$ must satisfy $C_{\sigma(K)}(u_1,\dots,u_M)=C_K(u_{\sigma(1)},\dots,u_{\sigma(M)})$ for any $\sigma\in S_M$ (essentially because
$\b(u_j)$'s commute, see \cite{FV-ABA, BP, BP-lect} for more details), thus it
suffices to consider $K=\{1,\dots,s\}$, $1\le s\le M$. To do that, one needs to look at the
expression
$$
\b_0(u_1)\cdots \b_0(u_s) \d_0(u_{s+1})\cdots \d_0(u_M)e_0\otimes
\Gamma(-2\eta h^{(0)})\left[\a_1(u_1)\cdots\a_1(u_s)\b_1(u_{s+1})\cdots \b_1(u_M)\right] e_0
$$
and commute all $\a_1$'s to the right of all $\b_1$'s while keeping their spectral parameters $u_j$ intact 
(it is easy to see that this is the only
way to obtain the needed term by composing $\b$'s in $V_0\otimes V_1$ defined via \eqref{eq:tensor}
and subsequently applying \eqref{eq:ba-comm}, \eqref{eq:db-comm} to reach the right-hand side of 
\eqref{eq:bbb}).

To that end, the commutation relation \eqref{eq:ba-comm} can be rewritten as
$$
\a(t_1)\b(t_2)=\frac{f(t_2-t_1-2\eta)}{f(t_2-t_1)}\frac{f(\la)}{f(\la-2\eta)}\cdot \b(t_2)\a(t_1) + (*)\cdot \b(t_1)\a(t_2),
$$
where the value of $(*)$ is irrelevant to us. Moving $\a_1(u_s)$ all the way to the right and removing the terms with $(*)$'s gives
\begin{multline*}
\a_1(u_1)\cdots\a_1(u_s)b_1(u_{s+1})\cdots \b_1(u_M)\sim\prod_{i=s+1}^M \frac{f(u_i-u_s-2\eta)}{f(u_i-u_s)}\, \a(u_1)\cdots \a(u_{s-1}) \\
\times \frac{f(\la)}{f(\la-2\eta)}\b_1(u_{s+1})\cdot\frac{f(\la)}{f(\la-2\eta)}\b_2(u_{s+2})\cdot\ldots\cdot \frac{f(\la)}{f(\la-2\eta)}\b_1(u_{M})\cdot \a(u_s).
\end{multline*}
Using the fact that $\b(w)F(\la)=F(\la+2\eta)\b(w)$, $\a(w)F(\la)=F(\la-2\eta)\a(w)$ for an arbitrary scalar function $F$, we obtain, after observing telescoping cancellations,  
$$
=\prod_{i=s+1}^M \frac{f(u_i-u_s-2\eta)}{f(u_i-u_s)}\cdot \frac{f(\la+2\eta(M-2s))}{f(\la-2\eta s)}\cdot 
\a_1(u_1)\cdots\a_1(u_{s-1})b_1(u_{s+1})\cdots \b_1(u_M)\a(s).
$$
Doing the same with the other $(s-1)$ factors $\a(u_k)$, $1\le k\le s-1$, we reach
$$
\sim\prod_{i=s+1}^M\prod_{j=1}^s \frac{f(u_i-u_j-2\eta)}{f(u_i-u_j)}\cdot \prod_{k=1}^s\frac{f(\la+2\eta(M-s-k))}{f(\la-2\eta k)}\cdot \b_1(u_{s+1})\cdots \b_1(u_M)\a_1(u_1)\cdots\a_1(u_s).
$$
If we now observe that in our normalization we always have $\a(w_1)\cdots\a(w_k)e_0=e_0$, introduce the notation
\begin{equation}\label{eq:d-eigenvalue}
\d(w_1)\cdots \d(w_k)=d(w_1,\dots,w_k)e_0,
\end{equation}
with a scalar function $d$ in the right-hand side (=the eigenvalue of the highest weight vector), 
and go back to \eqref{eq:bbb}, we can conclude that the term with $K=\{1,\dots,s\}$ has the form
\begin{gather*}
\prod_{i=s+1}^M\prod_{j=1}^s \frac{f(u_i-u_j-2\eta)}{f(u_i-u_j)}\cdot\Gamma(2\eta s)[d_0(u_{s+1},\dots,u_M)] \cdot\Gamma(-2\eta h^{(0)})\left[\prod_{k=1}^s\frac{f(\la+2\eta(M-s-k))}{f(\la-2\eta k)}\right]\\
\times \left(\b_0(u_1)\cdots \b_0(u_s)e_0\right)\otimes\left(\Gamma(-2\eta h^{(0)})\left[\b_1(u_{s+1})\cdots\b_1(u_M)\right]e_0\right).
\end{gather*}
The whole right-hand side of \eqref{eq:bbb} is obtained by replacing in the above expression $\{1,\dots,s\}$ by $K$, $\{s+1,\dots,M\}$ by $\overline{K}=\{1,\dots,M\}\setminus K$, $s$ by $|K|$, and summing over all $K\subset\{1,\dots,M\}$. 

This expression allows for an induction on the number of tensor factors. Replacing the second tensor factor in $V_0\otimes V_1$ by a product of two and repeating the computation, then replacing the third one by a product of two etc., we eventually reach the following result in the infinite tensor product $V_0\otimes V_1\otimes V_2\otimes\cdots$:
\begin{multline}\label{eq:bbb-intermediate}
\b(u_1)\cdots\b(u_M)E_0=\sum_{\substack{\{1,\dots,M\}=K_0\sqcup K_1\sqcup\dots \\ m_j:=|K_j|}}
\prod_{\substack{i>j\\ \alpha\in K_i,\be\in K_j}} \frac{f(u_\al-u_\be-2\eta)}{f(u_\al-u_\be)}\\
\times\prod_{i\ge 0} \left(\Gamma\left(-2\eta h^{([0,i))}+2\eta m_i\right)\left[d_i(u_{K_{>i}})\right]\cdot\prod_{k=1}^{m_i}\Gamma\left(-2\eta h^{([0,i])}\right)\left[\frac{f(\la+2\eta(m_{>i}-k))}{f(\la-2\eta k)}\right]\right)\\
\times \left(\b_0(u_{K_0})e_0\right)\otimes \left(\Gamma(-2\eta h^{(0)})\left[\b_1(u_{K_1})\right]e_0\right)\otimes\cdots\otimes \left(\Gamma(-2\eta h^{([0,i))})\left[\b_i(u_{K_i})\right]e_0\right)\otimes\cdots.
\end{multline}
Using \eqref{eq:abcd} one easily computes that in $V_\La(z)$
\begin{gather*}
d(w_1,\dots,w_k)=\prod_{i=1}^k \frac{f(z-w_i+(-\Lambda+1)\eta)}{f(z-w_i+(\Lambda+1)\eta)}\frac{f(\la+2\eta(i-1-\La))}{f(\la+2\eta(i-1))}\,,\\
\b(u_1)\cdots \b(u_s)e_0=\prod_{i=0}^{s-1} \frac{f(2\eta)}{f(\la+2\eta i)}\frac{f(\la-z+u_i+(-\La+1)\eta+2\eta (s-1))}{f(z-u_i+(\La+1)\eta)}\cdot e_s.
\end{gather*}

It remains to substitute these expressions into \eqref{eq:bbb-intermediate} and cancel out coinciding factors. 

Collecting the terms without the spectral parameters $u_j$ and with $\la$ not shifted by any $h^{(*)}$, we obtain $(f(2\eta))^M/\prod_{i=0}^{M-1}f(\la+2\eta i)$. Taking the terms without spectral parameters that involve $\la_i:=\la-2\eta h^{([0,i])}$ we observe
$$
\prod_{k=1}^{m_i}\frac{f(\la_i+2\eta(m_{>i}-k))}{f(\la_i-2\eta k)} \frac{\prod_{j=0}^{m_{>i}-1}f(\la_i+2\eta(j- m_i))}{\prod_{l=m_{i+1}}^{m_{>i}-1}f(\la_i+2\eta l)}
\prod_{j=0}^{m_{i+1}-1} \frac 1{f(\la_i+2\eta j)}=1,
$$
where the first ratio came from standalone $f$'s, the second ratio came from the numerator in $d_i$ and denominator in $d_{i+1}$, and the third ratio came from $\b(u_{K_{i+1}})$. Hence, we obtain
\begin{multline*}
\b(u_1)\cdots\b(u_M)E_0=\frac{(f(2\eta))^{M}}{
\prod_{i=0}^{M-1} f(\la+2\eta i)}\sum_{\substack{\{1,\dots,M\}=K_0\sqcup K_1\sqcup\dots \\ m_j:=|K_j|}}
\prod_{\substack{i>j\\ \alpha\in K_i,\be\in K_j}} \frac{f(u_\al-u_\be-2\eta)}{f(u_\al-u_\be)}\\
\times \prod_{i\ge 0}\prod_{k\in K_i} 
\phi_{i}(u_k)f(\la+u_k-q_{i}+2\eta m_i-2\eta h^{([0,i))}).
\end{multline*}
The last piece we need is the following symmetrization identity (this final step was not included in \cite[Theorem 5]{FV-ABA}). 

\begin{lemma}\label{lm:symm} For any $m\ge 1$, $v_1,\dots,v_m,\beta\in\C$, we have
\begin{equation}\label{eq:symm}
\sum_{\sigma\in S_m} \sigma\left(\prod_{1\le i<j\le m} \frac{f(v_i-v_j-\beta)}{f(v_i-v_j)}
\prod_{k=1}^m \frac{f(v_k+(m-2k+1)\beta)}{f(v_k)}\right)=\frac{f(\be)f(2\be)\cdots f(m\be)}{(f(\be))^m}\,.
\end{equation}
\end{lemma}
We will give its proof a little later; let us first use it to finish the proof of \eqref{eq:B-symm}. 

We apply Lemma \ref{lm:symm} in each cluster $K_i$, $i\ge 0$, with $\be=2\eta$, $m=|K_i|$, and 
$$
\{v_1,\dots,v_m\}=\{\la+u_{k}-q_i+2\eta m_i-2\eta h^{([0,i))} \}_{k\in K_i}. 
$$
This will add an additional factor with cross-differences of variables in the same clusters, and also turn the product of $f$'s in variables on the right into
$\prod_{j=0}^{m_i-1} f(\la-q_i+2\eta-2\eta h^{([0,i))} +4\eta j)$.
Noting that $h^{([0,i))}=\La_{[0,i)}-2m_{<i}$ we reach \eqref{eq:B-symm}. 
\begin{proof}[Proof of Lemma \ref{lm:symm}]\footnote{I am very grateful to E.~Rains for communicating this argument to me.} Using the fact that $f(-z)=-f(z)$ and the periodicity relations \eqref{eq:period}, one easily shows that the left-hand side of \eqref{eq:symm} viewed as a function of $v_1$ is elliptic, and it can only have one simple pole $v_1=0$ in the fundamental domain. Since non-constant elliptic functions must have at least two poles in their fundamental domain, this implies that the expression is independent of $v_1$. The same argument shows that it is independent of the other $v_j$'s as well, thus it is a constant. To compute the constant, one can set $v_j=j\be$, $1\le j\le m$, in which case only one term of the sum gives a nonzero contribution, which is exactly the required constant. 
\end{proof}

Let us proceed to the part (ii) of Theorem \ref{th:symm}. The proof is conceptually similar to that of (i), but it is more involved computationally. Its realization in simpler situations can be found in the proofs of \cite[Theorem 4.14]{BP}, \cite[Theorem 4.12]{BP-lect}, we will rely on those for the parts that are the same. 

We start with the action of $\d(v_1)\dots\d(v_n)$ on some vector $e^{(0)}\otimes e^{(1)}$ in the tensor product of two modules $V_0\otimes V_1$.\footnote{Note that for the proof of (i) we acted on the product of two highest weight vectors $e_0\otimes e_0$; the difference corresponds to the fact that in order to get $D_\nu=D_{\nu/0^{\ell(\nu)}}$ we need to act on $E_\nu$ (and collect the coefficient of $e_{\ell(\nu)}\otimes e_0\otimes e_0\otimes \cdots$), while for $B_\mu=B_{\mu/\varnothing}$ we had to act on $E_\varnothing$, see Section \ref{sc:inf-volume}.} This gives (cf., e.g., \cite[(4.39)]{BP-lect})
\begin{multline}\label{eq:ddd}
\d(v_1)\cdots \d(v_n) (e^{(0)}\otimes e^{(1)}) \\=\sum_{K\subset\{1,\dots,M\}} C_K
\left(\prod_{l\notin K}\b_0(v_l)\prod_{k\in K}\d_0(v_k)\, e^{(0)}\right)\otimes
\left(\Gamma(-2\eta h^{(0)})\left[\prod_{k\in K}\d_1(v_k)\prod_{l\notin 
K}\c_1(v_l)\right]e^{(1)}\right).
\end{multline}
As before, the commutativity of $\d(v_j)$'s implies that the coefficients 
$C_K=C_K(v_1,\dots,v_n)$ satisfy $C_{\sigma(K)}(v_1,\dots,v_n)=C_K(v_{\sigma(1)},\dots,v_{\sigma(n)})$, thus it suffices to compute $C_{\{1,\dots,s\}}$. One easily sees that the only way to obtain the corresponding term in \eqref{eq:ddd} is to look at
$$
\d_0(v_1)\cdots \d_0(v_s)\b_0(v_{s+1})\cdots \b_0(v_n) e^{(0)}\otimes \Gamma(-2\eta h^{(0)})\left[\d_1(v_1)\cdots \d_1(v_s)\c_1(v_{s+1})\cdots \c_1(v_n)\right]e^{(1)}
$$
and commute all the $\d_0$'s to the right of all the $\b_0$'s in the first factor. The needed commutation relation is \eqref{eq:db-comm} rewritten in the form
$$
\d(t_1)\b(t_2)=\frac{f(t_1-t_2-2\eta)}{f(t_1-t_2)}\frac{f(\la-2\eta\h)}{f(\la-2\eta \h-2\eta)}\cdot \b(t_2)\d(t_1)+(*)\cdot\b(t_1)\d(t_2),
$$
where we only need to use the first term in the right-hand side, as using the 
second term (with the $(*)$) does not lead to the needed form of the resulting 
expression. 

Moving $\d_0(v_s)$ all the way to the right gives (throwing out the terms with 
$(*)$'s)
\begin{multline*}
\d_0(v_1)\cdots \d_0(v_s)\b_0(v_{s+1})\cdots \b_0(v_n)\sim 
\d_0(v_1)\cdots 
\d_0(v_{s-1})\\
\times \frac{f(v_s-v_{s+1}-2\eta)}{f(v_s-v_{s+1})}\frac{f(\la-2\eta\h)}{
f(\la-2\eta \h-2\eta)}\,\b_0(v_{s+1})\cdots 
\frac{f(v_s-v_{n}-2\eta)}{f(v_s-v_{n})}\frac{f(\la-2\eta\h)}{
f(\la-2\eta \h-2\eta)}\,\b_0(v_{n}) \d_0(v_s).
\end{multline*}
We now use the relations (see \cite[p.~746]{FV-reps}) 
\begin{equation}\label{eq:shifts}
\d(w) F(\la,\h)=F(\la+2\eta,\h)\d(w), \qquad \b(w) 
F(\la,\h)=F(\la+2\eta,\h+2)\b(w)
\end{equation}
for an arbitrary scalar function $F(\la,\h)$, to get
$$
=\prod_{i=s+1}^n \frac{f(v_s-v_i-2\eta)}{f(v_s-v_i)} 
\frac{f(\la+2\eta(s-1)-2\eta \h)}{f(\la-2\eta(n-2s+1)-2\eta\h)}\,
\d_0(v_1)\cdots \d_0(v_{s-1})\b_0(v_{s+1})\cdots \b_0(v_n)\d_0(v_s). 
$$
Repeating the procedure for the other $\d_0(v_j)$, $1\le j\le s-1$, we obtain
$$
=\prod_{j=1}^s \frac{f(\la+2\eta(j-1-\h))}{f(\la+2\eta(j-1-n+s-\h))}\prod_{i=s+1}^n \frac{f(v_j-v_i-2\eta)}{f(v_j-v_i)} \cdot\b_0(v_{s+1})\cdots \b_0(v_n)\d_0(v_1)\cdots \d_0(v_{s}).
$$
The prefactor gives the value of $C_{\{1,\dots,s\}}(v_1,\dots,v_n)$. 

Let us now assume that the module $V^{(0)}$ is of the form $V_{\La_0}(z_0)$, the module $V^{(1)}$ is a highest weight module with the highest weight vector $e_0^{(1)}$, and require that the left-hand side of \eqref{eq:ddd} is proportional to $e_N\otimes e_0^{(1)}$ for some $N\ge 0$, as is required for the computation of $D_\nu$, cf. the footnote above. 

Introducing the scalar function $d^{(q)}$, cf. \eqref{eq:d-eigenvalue}, via the following eigenrelation in $V_\La(z)$: 
$$
\d(w_1)\cdots \d(w_k) e_q=d^{(q)}(w_1,\dots, w_k) e_q,
$$
we can rewrite \eqref{eq:ddd} as
\begin{multline}\label{eq:ddd-again}
\d(v_1)\cdots \d(v_n) (e_{n_0}\otimes e^{(1)})=
\sum_{\substack{K\subset\{1,\dots,n\}\\ n_{0}+n-|K|=N}}
\prod_{j=1}^{|K|}\frac{f(\la+2\eta(j-1-h^{(0)}))}{f(\la+2\eta(j-1-n+|K|-h^{(0)}))}
\\ \times \prod_{\substack{j\in K\\ i\notin K}}\frac{f(v_j-v_i-2\eta)}{f(v_j-v_i)} 
\cdot \Gamma(2\eta (n-|K|))\left[ d_0^{(n_0)}(v_K)\right]\cdot\Gamma(-2\eta h^{(0)})\left[ d_1(v_K)\right]\\ \times \left(
\prod_{i\notin K}\b_0(v_i)\, e_{n_0}\right)\otimes \left(\Gamma(2\eta(|K|-h^{(0)})\left[\prod_{i\notin K} \c_1(v_{i})\right]\,e^{(1)}\right).
\end{multline}

In order to proceed, we need to learn now how to apply $\prod \c(v_i)$ with the result being proportional to the highest weight vector. This is done in the following lemma, whose result is reminiscent of \eqref{eq:B-symm}, \eqref{eq:D-symm}. 

\begin{lemma}\label{lm:C-symm} Consider the tensor product of $m\ge 1$ evaluation Verma modules $V_{\La_1}(z_1)\otimes V_{\La_2}(z_2)\otimes\cdots\otimes V_{\La_m}(z_m)$, and let the pure tensors of the form $e_{k_1}\otimes \cdots \otimes e_{k_m}$ be its orthonormal basis. Then, with the notation $\widehat \la = \la-2\eta p$, 
\begin{equation}
\begin{gathered}\label{eq:C-symm}
\bigl\langle \c(w_1)\c(w_2)\cdots \c(w_p) \left(e_{k_1}\otimes e_{k_2} \cdots\otimes e_{k_m}\right),e_0\otimes e_0\otimes\cdots,\otimes e_0\bigr\rangle =
\prod_{i=0}^{p-1}{f(\la-2\eta\La_{[1,m]}+2\eta i)}\\ \times
\prod_{i=1}^m 
\prod_{j=0}^{k_i-1}\frac{f(2\eta(\La_i-j))}{f(\widehat\la 
+2\eta(2k_{<i}+k_i+j-\La_{[1,i]}))f(\widehat\la 
+2\eta(2k_{<i}+1+j-\La_{[1,i)}))}\\
\times (-1)^p\sum_{{\sigma\in S_p}}\sigma\left(\prod_{i<j}\frac{f(w_i-w_j+2\eta)}{f(w_i-w_j)}
\prod_{i=1}^{p}\widetilde\phi_{\ka_i}(w_i)f(\widehat\la 
-w_i+p_{\ka_i}+2\eta+4\eta(p-i)-2\eta\La_{[1,\ka_i)})\right),
\end{gathered}
\end{equation}
where $\ka=(\ka_1\ge\ka_2\ge \dots\ge \ka_p)=1^{k_1}2^{k_2}\cdots m^{k_m}$, we must have $k_1+\dots+k_m=p$ (otherwise the left-hand side vanishes), and
$$
\widetilde \phi_k(w)=\frac{1}{f(w-q_k)}\prod_{j={k+1}}^m\frac{f(w-p_j)}{f(w-q_j)}\,,\qquad k=1,\dots,m. 
$$ 
\end{lemma}
\begin{proof}
We follow the same basic route as in the computations for products of $\b$'s and $\d$'s above, and our explanations will be shorter here. First, we work in the tensor product of two highest weight modules $V^{(1)}\otimes V^{(2)}$, and
represent $\c(w_1)\cdots \c(w_p)$ in the form
\begin{equation*}
\sum_{K\subset\{1,\dots,p\}} C_K \left(\prod_{l\notin K}\a_1 (w_l)\prod_{k\in K}\c_1(w_k)\right)\otimes\left(\Gamma(-2\eta \h^{(1)})\left[\prod_{k\in K}\d_2(w_k)\prod_{l\notin K}\c_2(w_l)\right]\right)
\end{equation*}
with coefficients $C_K=C_K(w_1,\dots,w_p)$ satisfying the same symmetry as before. The coefficient $C_{\{1,\dots,s\}}$ is obtained from the expression
$$
\c_1(w_1)\cdots \c_1(w_s)\a_1(w_{s+1})\cdots \a_1(w_p)\otimes \Gamma(-2\eta h^{(1)}) \left[\d_2(w_1)\cdots \d_2(w_s)\c_2(w_{s+1})\cdots \c_2(w_p)\right]
$$
by commuting all the $\c$'s to the right of all the $\a$'s in the first factor without any exchanges of the spectral parameters in the process. For that we need the commutation relation
\eqref{eq:ac-comm} written in the form
$$
\c(t_1)\a(t_2)=\frac{f(t_1-t_2-2\eta)}{f(t_1-t_2)}\frac{f(\la-2\eta\h)}{f(\la-2\eta\h-2\eta)}+(*)\cdot \a(t_1)\c(t_2),
$$
where we will never use the second term with the $(*)$. With the help of the relations, see \cite[p.~746]{FV-reps} and cf. \eqref{eq:shifts},
$$
\a F(\la,\h)=F(\la-2h,\h)\,\a,\qquad \c F(\la,\h)=F(\la-2\eta,\h-2)\,\c,
$$
we eventually obtain, assuming that $V^{(1)}$ is of the form $V_{\La_1}(z_1)$ and the end result is proportional to the tensor product of the highest weight vectors,
\begin{multline*}
\c(w_1)\cdots \c(w_p)\left(e_{k_1}\otimes e^{(2)}\right)\\ \sim \sum_{\substack{K\subset\{1,\dots,p\}\\
|K|=k_1}} \prod_{j=1}^{|K|} \frac{f(\la+2\eta(j-1- h^{(1)}))}{f(\la+2\eta(j-1-p+|K|-h^{(1)})}
\prod_{\substack{j\in K\\ i\notin K}}\frac{f(w_j-w_i-2\eta)}{f(w_j-w_i)} \cdot \Gamma(-2\eta h^{(1)})\left[ d_2(w_K)\right]\\ \times \left(
\Gamma(-2\eta(p-|K|))\left[\prod_{j\in K}\c_1(w_j)\right]\, e_{n_1}\right)\otimes \left(\Gamma(2\eta(|K|-h^{(1)}))\left[\prod_{i\notin K} \c_2(w_{i})\right]\,e^{(2)}\right).
\end{multline*}
Replacing $V^{(2)}$ by $V_{\La_2}(z_2)\otimes V^{(3)}$, repeating the computation, and iterating yields (we keep assuming that the result is proportional to $e_0\otimes\cdots\otimes e_0$)
\begin{multline}\label{eq:ccc}
\c(w_1)\c(w_2)\cdots \c(w_p) \left(e_{k_1}\otimes e_{k_2} \cdots\otimes e_{k_m}\right)\sim \sum_{\substack{\{1,\dots,p\}=K_1\sqcup K_2\sqcup\dots\sqcup K_m\\
|K_1|=k_1,\,\dots,\,|K_m|=k_m}} \prod_{\substack{ i<j\\ \al\in K_i,\be\in K_j}} \frac{f(w_\al-w_\be-2\eta)}{f(w_\al-w_\be)}\\ \times
\prod_{i=1}^m\prod_{j=1}^{k_i} \frac{f(\la+2\eta(j-1+k_{<i}- h^{([1,i])}))}{f(\la+2\eta(j-1+k_{<i}-k_{>i}-h^{([1,i])})}\,\prod_{i=1}^m \Gamma(2\eta(k_{<i}-h^{([1,i])}))\left[d_{>i}(w_{K_i})\right]\\ 
\times\left(
\Gamma(-2\eta k_{>1})\left[\prod_{j\in K_1}\c_1(w_j)\right]\, e_{k_1}\right)\otimes \cdots \otimes \left(\Gamma(2\eta(-k_{>i}+k_{<i}-h^{([1,i))}))\left[\prod_{j\in K_2}\c_2(w_j)\right]\, e_{k_2}\right)\otimes\cdots,
\end{multline}
where $d_{>i}$ stands for the highest weight vector's eigenvalue in $V_{\La_{i+1}}(z_{i+1})\otimes\cdots\otimes V_{\La_m}(z_m)$ under the application of (products of) the  $\d$-operators. 

Observe that $h^{(k)}\equiv \La_k$ because our result is $e_0\otimes \cdots\otimes e_0$. 

In order to evaluate $\Gamma(2\eta(k_{<i}-h^{([1,i])}))\left[d_{>i}(w_{K_i})\right]$ we go back to \eqref{eq:d-norm-empty} and obtain
\begin{multline*}
\Gamma(2\eta(k_{<i}-h^{([1,i])}))\left[d_{>i}(w_{K_i})\right]=\prod_{\substack{j>i\\ k\in K_i}}\frac{f(z_j-w_k+(-\La_j+1)\eta)}{f(z_j-w_k+(\La_j+1)\eta)} \prod_{l=0}^{k_i-1}\frac{f(\la+2\eta(k_{<i}-\La_{[1,m]}+l))}{f(\la+2\eta(k_{<i}-\La_{[1,i]}+l))}\,.
\end{multline*}

To compute the $\c$-terms, we use \eqref{eq:abcd} to write (in $V_{\La}(z)$)
\begin{multline*}
\c(w_1)\cdots \c(w_k) e_k=c(\la,w_1)\cdots c(\la-2\eta(k-1),w_k)\,e_k\\=
\prod_{j=1}^k \frac{(-1)f(-\la-z+w_j+(\La-1)\eta)}{f(z-w_j+(\La+1)\eta)}\frac{f(2(\La+1-j)\eta)}{f(\la-2(j-1)\eta)}\frac{f(2j\eta)}{f(2\eta)}\, e_0. 
\end{multline*}
Hence, the last line of \eqref{eq:ccc} gives
\begin{multline*}
\prod_{i=1}^m \Biggl(\prod_{j=1}^{k_i}\frac{f(2j\eta)f(2(\La_i+1-j)\eta)}{f(\la+2\eta(k_{<i}-k_{>i}-\La_{[1,i)}-j+1))f(2\eta)}\\ \times
\prod_{k\in K_i} \frac{f(\la+2\eta(k_{<i}-k_{>i}-\La_{[1,i)})+z_i-w_k+(-\La_i+1)\eta)}{f(z-w_k+(\La_i+1)\eta)}\Biggr)\cdot e_0\otimes\cdots\otimes e_0 .
\end{multline*}
Collecting all the factors and recalling the notations $\widehat \la$ and $\widetilde \phi_k$ from the statement of the lemma, we obtain that the left-hand side of \eqref{eq:C-symm} equals
\begin{multline*} 
\prod_{i=1}^m 
\prod_{j=0}^{k_i-1}\frac{f(2(j+1)\eta)f(2(\La_i-j)\eta)}{f(\widehat\la 
+2\eta(2k_{<i}+k_i+j-\La_{[1,i]}))f(\widehat\la 
+2\eta(2k_{<i}+1+j-\La_{[1,i)}))f(2\eta)}\\ \times 
\prod_{i=0}^{p-1}{f(\la+2\eta(i-\La_{[0,m]}))}\cdot
\sum_{\substack{\{1,\dots,p\}=K_1\sqcup K_2\sqcup\dots\sqcup K_m\\
|K_1|=k_1,\,\dots,\,|K_m|=k_m}} \prod_{\substack{ i<j\\ \al\in K_i,\be\in K_j}} \frac{f(w_\al-w_\be-2\eta)}{f(w_\al-w_\be)}\\ \times
\prod_{i=1}^m \prod_{k\in K_i}(-1)\widetilde\phi_i(w_k)f(\widehat \la +2\eta(2k_{<i}+k_i-\La_{[1,i)})+z_i-w_k+(-\La_i+1)\eta),
\end{multline*}
where the appearance of $(-1)$ is due to $f(z-w_k+(\La_i+1)\eta)=-f(w_k-q_i)$. 

The claim of the lemma now follows by applying Lemma \ref{lm:symm} to additionally symmetrize the above expression over the variables $\{w_k\}$ within each cluster $k\in K_i$, $1\le i\le m$. 
\end{proof}

We return to the proof of \eqref{eq:D-symm}, and we pick up where we left off, i.e., at \eqref{eq:ddd-again}. We have proved Lemma \ref{lm:C-symm} to handle the very last factor in the right-hand side of \eqref{eq:ddd-again}, but we need to have explicit expressions for the other factors as well; we will give them now. In what follows we assume that $V^{(1)}$ participating in \eqref{eq:ddd-again} is $V_{\La_1}(z_1)\otimes\cdots\otimes V_{\La_m}(z_m)$ with some $m\ge 1$ that we will later send to infinity, and $e^{(1)}=e_{n_1}\otimes \cdots\otimes e_{n_m}$. 

Straightforward computations using \eqref{eq:abcd} give (it is also useful to remember that $|K|=n_0+n-N$ and $h^{(0)}=\La_0-2N$)
\begin{gather*}
\Gamma(2\eta(n-|K|))\left[d_0^{(n_0)}(v_K)\right]=\prod_{k\in K} \frac{f(z_0-v_k+(-\La_0+1+2n_0)\eta)}{f(z_0-v_k+(\La_0+1)\eta)}\cdot \prod_{i=N}^{n+n_0-1}\frac{f(\la+2\eta(i-\La_0))}{f(\la+2\eta(i-n_0))}\,,\\
\Gamma(-2\eta h^{(0)})\left[ d_1(v_K)\right]=\prod_{j=1}^m \prod_{k\in K} \frac{f(z_j-v_k+(-\La_j+1)\eta)}{f(z_j-v_k+(\La_j+1)\eta)}\cdot \prod_{i=N}^{n+n_0-1}\frac{f(\la+2\eta(i+N-\La_{[0,m]}))}{f(\la+2\eta(i+N-\La_0))}\,,\\
\prod_{l\notin K}\b_0(v_l)\, e_{n_0}=\prod_{i=0}^{N-n_0-1}\frac{(-1)f(2\eta)}{f(\la+2\eta i)}\cdot\prod_{l\notin K} \frac{f(-\la+z_0-v_l+(\La_0-1)\eta+2\eta-2\eta N)}{f(z_0-v_l+(\La_0+1)\eta)}\,e_N,\\
\prod_{j=1}^{|K|}\frac{f(\la+2\eta(j-1-h^{(0)}))}{f(\la+2\eta(j-1-n+|K|-h^{(0)}))}=\prod_{i=N}^{n+n_0-1} \frac{f(\la+2\eta(i+N-\La_0))}{f(\la+2\eta(i+n_0-\La_0))}\,.
\end{gather*}

Finally, we use Lemma \ref{lm:C-symm} to write $\Gamma(2\eta(|K|-h^{(0)})\left[\prod_{i\notin K} \c_1(v_{i})\right]\,e^{(1)}$. The shift of $\la$ we represent in the form $2\eta(|K|-h^{(0)})=2\eta(n+n_0+N)-\La_0$, and the match of notation with Lemma \ref{lm:C-symm} is as follows:
\begin{gather*}
k_j=n_j,\quad 1\le j\le m, \qquad p=n-|K|=N-n_0,\\
2\eta(|K|-h^{(0)})+\widehat \la=\la+2\eta(n+2n_0-\La_0). 
\end{gather*}
We obtain
\begin{multline*}
\left\langle\Gamma(2\eta(|K|-h^{(0)})\biggl[\prod_{i\notin K} \c_1(v_{i})\biggr]\,e^{(1)},e_0\otimes\cdots\otimes e_0\right\rangle=\prod_{i=0}^{N-n_0-1}{f(\la+2\eta(i+n+n_0+N-\La_{[0,m]}))}\\
\times \prod_{i=1}^m 
\prod_{j=0}^{n_i-1}\frac{f(2\eta(\La_i-j))}{f(\la+2\eta(n+2n_{<i}+n_i+j-\La_{[0,i]}))f(\la+2\eta(n+2n_{<i}+1+j-\La_{[0,i)}))}\cdot  (-1)^{N-n_0}\\
\times\sum_{\substack{\sigma:\{1,\dots,N-n_0\}\\ \sigma \text{ is a bijection}}}\sigma\left(\prod_{i<j}\frac{f(v_i-v_j+2\eta)}{f(v_i-v_j)}
\prod_{i=1}^{N-n_0}\widetilde\phi_{\nu_i}(v_i)f(\la+2\eta n 
-v_i+p_{\nu_i}+2\eta+4\eta(N-i)-2\eta\La_{[0,\nu_i)})\right).
\end{multline*}

We need to substitute all these expressions into \eqref{eq:ddd-again}, multiply by (see \eqref{eq:d-norm}) 
$$
\prod_{i=0}^m \prod_{j=1}^n\frac{f(z_i-v_j+(\La_i+1)\eta)}{f(z_i-v_j+(-\La_i+1)\eta)}
\cdot \prod_{l=0}^{n-1}\frac{1}{f(\la+2\eta(l+2N-\La_{[0,m]}))},
$$
and send $m\to\infty$ (this limit transition amounts to a stabilization at large enough $m$). 
It is now straightforward to see that we arrive at \eqref{eq:D-symm}. 
\end{proof}

\section{Orthogonality relations} \label{sc:ortho}

We would like to think of the functions $\{B_\mu(\la;u_1,\dots,u_M)\mid\mu\text{ is a signature of length }M\}$, explicitly given by \eqref{eq:B-symm}, as of a $\lambda$-dependent, Fourier-like basis in a suitable space of functions in $u_1,\dots,u_M$. This should entail two types of statements - completeness (any function from the space can be written as a linear combination of $B_\mu$'s) and orthogonality with respect to some inner product. Cauchy identities similar to our \eqref{eq:cauchy} are useful in proving both. Let us illustrate how this might work on a simple example.

Consider the simplest one-variable Cauchy identity of the form
	\begin{align*}
		\sum_{n=0}^{\infty}\frac{z^{n}}{w^{n+1}}=\frac{1}{w-z},\qquad \left|\frac{z}{w}\right|<1.
	\end{align*}
	By shifting the summation index towards $-\infty$, we can write
	\begin{align*}
		\sum_{n=-M}^{\infty}\frac{z^{n}}{w^{n+1}}=\frac{w^{M}}{z^{M}}\frac{1}{w-z},\qquad \left|\frac{z}{w}\right|<1.
	\end{align*}
	Now take contour integrals in $z$ and $w$ (over positively oriented
	circles with $|z|<|w|$) of both sides of this relation multiplied by 
	$P(z)Q(w)$, where $P$ and $Q$ are Laurent polynomials. Then in the left-hand side 
	we obtain the same result for any $M\gg 1$, and in the right-hand side the 
	$w$ contour can be shrunk to zero, thus picking the residue at $w=z$. Hence,	
	\begin{align*}
		\sum_{n=-\infty}^{\infty}\oint\oint \frac{z^{n}}{w^{n+1}}\,P(z)Q(w)\frac{dz}{2\pi\i}\frac{dw}{2\pi\i}
		=
		\oint P(z)Q(z)\frac{dz}{2\pi\i}.
	\end{align*}
	The convergence condition $|z|<|w|$ is irrelevant for the left-hand side because the 
	sum over $n$ now contains only finitely many terms.
	This implies
	$$
	\sum_{n=-\infty}^{\infty}z^{n}\oint\frac{Q(w)}{w^{n+1}}\,\frac{dw}{2\pi\i}
		=Q(z),
	$$
	which is a (in this case, obvious) completeness statement for the basis $\{z^n\}_{n\in\Z}$, and can also be rewritten as a biorthogonality relation
	\begin{align*}
		\sum_{n=-\infty}^{\infty}\frac{z^{n}}{w^{n+1}}=\delta(w-z).
	\end{align*} 
	To get another (bi)orthogonality relation, integrate both sides of the above identity against 
	$w^{m}\frac{dw}{2\pi\i}$,
	$m\in\Z$. Since $\{z^{n}\}_{n\in\Z}$ are linearly independent, 
	we obtain
	\begin{align}\label{eq:orth-single-var}
		\oint_{|w|=\textnormal{const}}w^{m}{(w^{-1})^{n}}\,\frac{w^{-1}dw}{2\pi\i}=\mathbf{1}_{m=n}.
	\end{align}
 This identity, of course, also readily follows from the Cauchy's integral formula.
 
 The above program was fully realized for simpler versions of the $B_{\mu}$'s arising from the (inhomogeneous higher spin) six vertex model in \cite[Section 7]{BP}. It might be possible to do the same for the $B_{\mu}$'s as well, but we will not pursue that here. Instead, we will focus on an analog of \eqref{eq:orth-single-var} that we originally derived from \eqref{eq:cauchy} using the arguments that are very similar to the ones above. 
 
To formulate a statement we will need integration contours that only exist under certain restrictions on the parameters; let us state them.\footnote{Our restrictions are sufficient but not necessary; we are not pursuing the most general statement of this form.} 

\begin{definition}\label{df:adm} For any $M\ge 1$, we say that the parameters $\tau,\eta, \{z_i\}_{i\ge 0}, \{\La_i\}_{i\ge 0}$ are admissible if there exists a fundamental parallelogram $\mathcal D$ of $\C/(\Z+\tau \Z)$ (or a fundamental strip of $\C/\Z$ in the trigonometric case) such that $\{p_i\}_{i\ge 0}$ and $\{q_i\}_{i\ge 0}$ of \eqref{eq:pq} form collections of points in $\mathcal{D}$ with points in each collection being close enough to each other, and $\eta$ is small enough, so that there exist (positively oriented) contours $\ga_1,\dots\ga_M\subset D$ such that $\ga_M$ includes all the $p_i$'s, $\ga_{M-1}$ includes $\ga_M+2\eta$ (which is the image of $\ga_M$ under $z\mapsto z+2\eta$), \dots, $\ga_1$ includes $\ga_2+2\eta$, and none of the $\ga_i$'s includes any of the $q_j$'s.  
See Figure \ref{fig:contours}. 
\end{definition} 

\begin{figure}
	\includegraphics[scale=0.7]{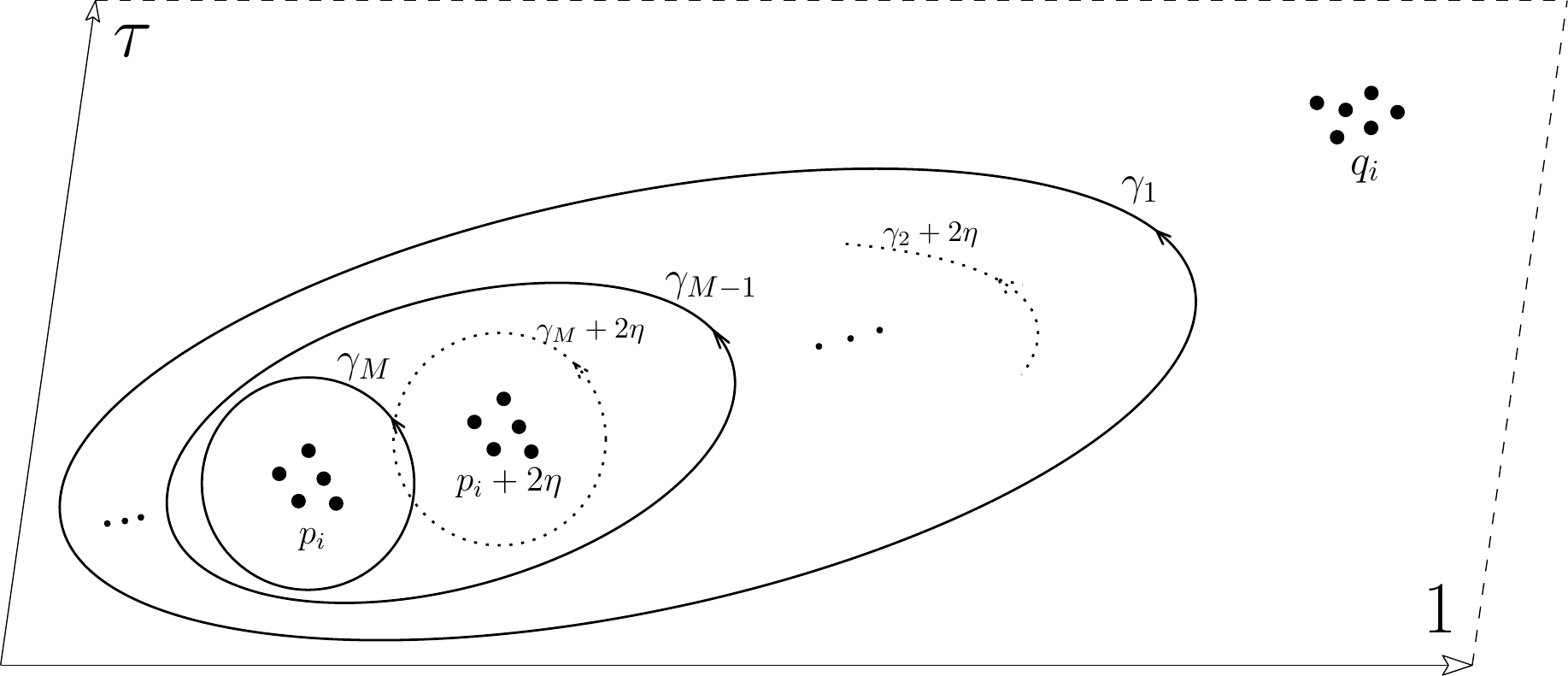}
	\caption{Contours $\ga_1,\dots,\ga_M$ of Definition \ref{df:adm}.}\label{fig:contours}
\end{figure}

\begin{theorem}\label{th:orth} Fix $M\ge 1$ and assume our parameters are admissible in the sense of Definition \ref{df:adm}. Then for any $\mu=(\mu_1\ge\dots\ge\mu_M\ge 0)=0^{m_0}1^{m_1}\cdots$, $\nu=(\nu_1\ge\dots\ge\nu_M\ge 0)$ we have
\begin{multline}\label{eq:orth} \oint_{\ga_1}\cdots\oint_{\ga_M}  B_{\mu}(\la;u_1,\dots,u_M) \\ \times\prod_{i<j} \frac{f(u_i-u_j)}{f(u_i-u_j-2\eta)}\prod_{i=1}^M \psi_{\nu_i}(u_i) f(\la-u_i+p_{\nu_i}+2\eta+4\eta(M-i)-2\eta\La_{[0,\nu_i)})\frac{du_i}{2\pi \i}=c_\mu\cdot  \mathbf{1}_{\la=\mu},
\end{multline}
where the integration contours $\{\ga_i\}_{i=1}^M$ are as in Definition \ref{df:adm}, $\psi_l(v)$ is defined in \eqref{eq:phi-psi}, and 
\begin{equation}
\begin{gathered}\label{eq:squared-norm}
c_{\mu}=c_\mu(\la)=\frac{(f(2\eta))^{M}}{(f'(0))^M\prod_{i=0}^{M-1}f(\la+2\eta i)}
\\
\times \prod_{i\ge 0} 
\prod_{j=0}^{m_i-1}\frac{f(\la 
+2\eta(2m_{<i}+m_i+j-\La_{[0,i]}))f(\la 
+2\eta(2m_{<i}+1+j-\La_{[0,i)}))}{f(2\eta(\La_i-j))}\,.
\end{gathered}
\end{equation}
\end{theorem}
\smallskip
\noindent\emph{Comments.} 1. The statement of the theorem can be extended to $\mu$ and $\nu$ being arbitrary signatures, not necessarily nonnegative. For that one needs to modify the formula \eqref{eq:B-symm} and the rest of the integrand in \eqref{eq:orth} using
$$
\prod_{i=0}^{k-1} \frac{f(u-p_i)}{f(u-q_i)}=\frac{\prod_{i=-\infty}^{k-1} \frac{f(u-p_i)}{f(u-q_i)}}{\prod_{i=-\infty}^{-1} \frac{f(u-p_i)}{f(u-q_i)}},\qquad \La_{[0,k)}=\sum_{i=-\infty}^{k-1}\La_i-\sum_{i=-\infty}^{-1}\La_i.
$$
While the infinite products and sums above may not make sense, canceling infinitely many terms leads to meaningful interpretation of the right-hand sides that one can also use for $k<0$. The proof for not necessarily nonnegative $\mu$ and $\nu$ remains the same.  

\smallskip
\noindent 2. While \eqref{eq:orth} does not immediately look like an orthogonality relation, it does allow to extract the coefficients of $B_\mu$'s in a series, as we will see after the proof of this theorem. To see the connection to orthogonality, assume for a second that the integration contours could be deformed to the same contour without crossing any poles of the integrand. (It is not clear how to do this in general, but in some cases, e.g., in the trigonometric limit, this is actually doable.) Then in the integrand one could replace 
$$
\prod_{i<j} \frac{f(u_i-u_j)}{f(u_i-u_j-2\eta)}\quad \text{by} \quad \prod_{i\ne j} \frac{f(u_i-u_j)}{f(u_i-u_j-2\eta)}\cdot \prod_{i>j} \frac{f(u_i-u_j-2\eta)}{f(u_i-u_j)}
$$
and then symmetrize the integrand in the $u$-variables. As the result, comparing to \eqref{eq:B-symm}, one would see  
$B_{\mu}(\la;u_1,\dots,u_M)$ integrated against another such $B_\nu$ subject to substitutions 
$$
p_i\mapsto 2\eta -q_i,\qquad q_j\mapsto 2\eta -p_j,\qquad u_k\mapsto 2\eta -u_k\quad \text{for all}\quad i,j,k, 
$$
times the `orthogonality weight' $\prod_{i\ne j} \frac{f(u_i-u_j)}{f(u_i-u_j-2\eta)}$; this could be viewed as an analog of \eqref{eq:orth-single-var}. For a similar re-interpretation in a simpler situation, see \cite[Theorem 7.4 and Corollary 7.5]{BP}.

\begin{proof}[Proof of Theorem \ref{th:orth}] First, let us check that the integrand of \eqref{eq:orth} is doubly periodic (periods 1 and $\tau$), viewed as a meromorphic function in each of the variables $u_i$, $1\le i\le M$. 

We employ \eqref{eq:B-symm} to write $B_\mu$ as a sum over permutations $\sigma\in S_M$ and then in each term rewrite all $u_i$-dependent factors in the form $f(u_i+*)$ with various $*$'s using $f(-x)=-f(x)$. Since there are always equally many factors of this type in the numerator and denominator, the first relation of \eqref{eq:period} implies the periodicity with period 1. The second relation of \eqref{eq:period} will imply the periodicity with period $\tau$ if we show that the total sum of $*$'s in $f(u_i+*)$'s in the numerators is equal to the similar sum in the denominator. 

First consider the term corresponding to $\sigma=\mathrm{id}$. Then all cross-terms involving two $u_j$'s cancel out, and the needed equality is immediately visible by a direct inspection (relations \eqref{eq:pq} need to be used). Further, as we modify the permutation by multiplying it by an elementary transposition, we swap two neighboring variables. In the cross-terms, this leads to a single change of the form
$$
\frac{f(u_k-u_l-2\eta)}{f(u_k-u_l)}\leadsto \frac{f(u_l-u_k-2\eta)}{f(u_l-u_k)}=\frac{f(u_k-u_l+2\eta)}{f(u_k-u_l)}\,.
$$
Hence, $u_k$ gains an extra $4\eta$ in the sum of the numerator $*$'s, and $u_l$ loses the same $4\eta$. The only other contribution that changes is $4\eta(M-i)$ in the right-hand side of \eqref{eq:B-symm}. It is readily seen to produce a compensatory effect, offsetting the sum of $*$'s by exactly the opposite amount. For example, when we pass from $\sigma=\mathrm{id}$ to $\sigma=(12)$ we observe the change
\begin{multline}
\frac{f(u_1-u_2-2\eta)}{f(u_1-u_2)}f(u_1+\dots+4\eta(M-1))f(u_2+\dots+4\eta(M-2))\\ \leadsto
\frac{f(u_1-u_2+2\eta)}{f(u_1-u_2)}f(u_1+\dots+4\eta(M-2))f(u_2+\dots+4\eta(M-1)),
\end{multline}
and the total sum of shifts of $u_1$ and $u_2$ did not change. Thus, the double periodicity of the integrand of \eqref{eq:orth} is verified. 

The next step is to prove that the left-hand side of \eqref{eq:orth} vanishes if $\mu\ne \nu$. This is a nontrivial combinatorial argument that is, however, completely analogous to the proof of \cite[Lemma 3.5]{BCPS14}, see \cite[Lemma 7.1]{BP} for a more generic statement. The periodicity of the integrand allows us to move $u$-contours from `surrounding $p_i$'s' to `surrounding $q_j$'s' position, because the total integral along the boundary of any fundamental parallelogram is 0. Keeping in mind that $f(z)$ has a single simple zero in each fundamental parallelogram of $\C/(\Z+\tau \Z)$, we literally repeat the arguments of the above cited lemmas to reach the desired conclusion.

The final step is the computation of the `squared norm' $c_\mu$. 
A part of the proof of \cite[Lemma 3.5]{BCPS14} shows that for $\mu=\nu$, the integral \eqref{eq:orth} splits into a product of similar integrals corresponding to the clusters (=groups of equal coordinates) of $\mu$. To see how the computation proceeds inside each cluster, let us first assume that $\mu_1=\dots=\mu_M=x$. The sum over $S_M$ in the expression \eqref{eq:B-symm} for $B_\mu$ is then computed via Lemma \ref{lm:symm}:
\begin{multline*}
\sum_{\sigma\in S_M}\sigma\left(\prod_{i<j} 
\frac{f(u_i-u_j-2\eta)}{f(u_i-u_j)}\cdot \prod_{i=1}^M 
f(\la+u_i-q_{x}+2\eta+4\eta(M-i)-2\eta\La_{[0,x)})\right)\\=
\prod_{i=1}^M\frac{f(2\eta i)f(\la+u_i-q_x+2\eta(M-\La_{[0,x)}))}{f(2\eta)};
\end{multline*}
this is very similar to the application of Lemma \ref{lm:symm} right after its statement. Hence, the left-hand side of \eqref{eq:orth} now takes the form (we take into account the prefactor of the sum in \eqref{eq:B-symm})
\begin{multline}\label{eq:orth-aux}
\frac{(-1)^M(f(2\eta))^{M}}{
\prod_{i=0}^{M-1} f(\la+2\eta i)}\oint_{\ga_1}\cdots \oint_{\ga_M}
\prod_{i<j} 
\frac{f(u_i-u_j-2\eta)}{f(u_i-u_j)}\\ \times\prod_{i=1}^M\frac{f(\la+u_i-q_x+2\eta(M-\La_{[0,x)}))f(\la-u_i+p_x+2\eta+4\eta(M-i)-2\eta\La_{[0,x)}))}{f(u_i-p_x)f(u_i-q_x)}\,\frac{du_i}{2\pi \i}.
\end{multline}
We can now sequentially evaluate the integrals by computing the residues, starting with the inner most one. Computing $\Res_{u_M=p_x}$ gives, writing only the terms that depended on $u_M$,
$$
\frac{1}{f'(0)}\prod_{i=1}^{M-1} \frac{f(u_i-p_x)}{f(u_i-p_x-2\eta)}\cdot\frac{f(\la+2\eta(M-\La_{[0,x]}))f(\la+2\eta-2\eta\La_{[0,x)})}{f(-2\eta\La_x)}\,.
$$
Since $\prod_{i=1}^{M-1}\frac{f(u_i-p_x)}{f(u_i-p_x-2\eta)}\prod_{i=1}^{M-1}{(f(u_i-p_x))^{-1}}=\prod_{i=1}^{M-1}{(f(u_i-p_x-2\eta))^{-1}}$, the poles of the integrand are now at $u_i=p_x+2\eta$, and the next residue that we take is $\Res_{u_{M-1}=p_x+2\eta}$. The $u_{M-1}$-dependent factors give
$$
\frac{1}{f'(0)}\prod_{i=1}^{M-2} \frac{f(u_i-p_x-2\eta)}{f(u_i-p_x-4\eta)}\cdot\frac{f(\la+2\eta(M+1-\La_{[0,x]}))f(\la+4\eta-2\eta\La_{[0,x)})}{f(2\eta-2\eta\La_x)}\,.
$$
Next, we compute residues at ${u_{M-2}=p_x+4\eta},\dots,{u_1=p_x+2(M-1)\eta}$ (in this order) to conclude that the integral in \eqref{eq:orth-aux} equals
$$
\frac {(-1)^M}{(f'(0))^M}\prod_{j=0}^{M-1} \frac{f(\la+2\eta(j+M-\La_{[0,x]}))f(\la+2\eta(j+1-\La_{[0,x)}))}{f(2\eta(\La_x-j))}\,.
$$
Going back to the case of the general $\mu=0^{m_0}1^{m_1}2^{m_2}\cdots$, we observe that the computation for the $j$th cluster is exactly the same with $M$ replaced by $m_j$ (the size of the $j$th cluster), and $\la$ shifted by $4\eta m_{<j}$, because $(M-i)$ in \eqref{eq:B-symm} takes values $m_{<j}+\{0,1,\dots,m_j-1\}$ when $\mu_i$ belongs to the $j$th cluster. This completes the proof of Theorem \ref{th:orth}. 
\end{proof}

As was mentioned in Comment 2 after Theorem \ref{th:orth}, \eqref{eq:orth} can be used to extract the coefficients of $B_\mu$ from their linear combinations. Let us check how this works for the Cauchy identity \eqref{eq:cauchy}. The coefficients of the $B_{\mu}$'s are essentially the $D_\mu$'s; hence, this should produce an integral representation for latter. Rather than justifying this type of an argument (which is not difficult but requires analytic control on the convergence of the Cauchy identity), we will directly verify that the resulting integral representation is indeed valid by showing that it is equivalent to the part (ii) of Theorem \ref{th:symm}. 

\begin{proposition} Fix $N\ge 1$, and assume our parameters are admissible in the sense of Definition \ref{df:adm} with $M=N$. Take $n\ge 1$, and let $v_1,\dots, v_n$ be complex numbers that lie inside the fundamental parallelogram used in Definition \ref{df:adm} but outside the contour $\ga_1$. Then for any $\nu=(\nu_1\ge \dots\ge\nu_N\ge 0)=0^{n_0}1^{n_1}2^{n_2}\cdots$ we have
\begin{multline}\label{eq:D-int}
D_\nu(\la-2\eta n;v_1,\dots,v_n)=\frac{(-1)^N(f(2\eta))^N}{\prod_{i=-n}^{N-1}f(\la+2\eta i)\cdot c_{\nu}(\la)}\,\oint_{\ga_1}\dots\oint_{\ga_N} \prod_{i<j}\frac{f(u_i-u_j)}{f(u_i-u_j-2\eta)}\\ \times \prod_{i=1}^N \Biggl(\psi_{\nu_i}(u_i) f(\la-u_i+p_{\nu_i}+2\eta+4\eta(N-i)-2\eta\La_{[0,\nu_i)})\\ \times \frac{f(\la+u_i-q_0+2\eta(N-n))}{f(u_i-q_0)}\prod_{j=1}^n \frac{f(u_i-v_j+2\eta)}{f(u_i-v_j)}
\Biggr)\frac{du_i}{2\pi \i}\,,
\end{multline}
where $c_\nu(\la)$ is as in \eqref{eq:squared-norm}. 
\end{proposition} 
\begin{proof} We will establish the equality of the right-hand side of \eqref{eq:D-int} and that of \eqref{eq:D-symm}. The contours $\ga_j$ with $j=N,N-1,\dots,N-n_0+1$ will be shrunk to $p_0$ (in that order); they correspond to the $0$th cluster of $\nu$. The other $\ga_j$'s will be expanded to surround the $q_j$'s. As the integrand will have no poles there, we would only need to take into account the poles of the form $u_i=v_j$ that we encounter along the way. In order to do manipulations of this sort, we first need to check the double periodicity of the integrand viewed as a function in each of the $u_i$'s. 

As in the periodicity check in the proof of Theorem \ref{th:orth}, for a given $u_i$ we write all the $u_i$-dependent factors in the form $f(u_i+*)$. We then need to verify that the sum of $*$'s in the numerator is exactly the same as the similar sum in the denominator, which is immediate (one needs to use \eqref{eq:pq} along the way). 

Let us now start taking the residues. Looking at the residue at $u_N=p_0$ and writing out the factors that depend on $u_N$ we obtain 
$$
\prod_{i=1}^{N-1} \frac{f(u_i-p_0)}{f(u_i-p_0-2\eta)}\cdot \frac{f(\la+2\eta)}{f'(0)}\cdot
\frac{f(\la-2\eta\La_0+2\eta(N-n))}{f(-2\eta\La_0)}\cdot\prod_{j=1}^n \frac{f(v_j-p_0-2\eta)}{f(v_j-p_0)}\,.
$$
The next relevant residue is at $u_{N-1}=p_0+2\eta$ (assuming that $n_0\ge 2$), and the $u_{N-1}$-dependent terms give
$$
\prod_{i=1}^{N-2} \frac{f(u_i-p_0-2\eta)}{f(u_i-p_0-4\eta)}\cdot \frac{f(\la+4\eta)}{f'(0)}\cdot
\frac{f(\la-2\eta\La_0+2\eta(N-n+1))}{f(2\eta-2\eta\La_0)}\cdot\prod_{j=1}^n \frac{f(v_j-p_0-4\eta)}{f(v_j-p_0-2\eta)}\,.
$$
After taking the total of $n_0$ residues at $u_{N-i}=p_0+2\eta i$, $1\le i\le n_0$, we see that the right-hand side of \eqref{eq:D-int} is equal to  
\begin{multline}\label{eq:D-int-aux}
\prod_{i=1}^{n_0-1}\frac{f(\la+2\eta(i+1))f(\la-2\eta\La_0+2\eta(N-n+i))}{f'(0)f(2\eta(\La_0-i))}
\frac{(-1)^{N-n_0}(f(2\eta))^N}{\prod_{i=-n}^{N-1}f(\la+2\eta i) c_{\nu}(\la)}\prod_{j=1}^n \frac{f(v_j-p_0-2\eta n_0)}{f(v_j-p_0)}\\ \times\oint_{\ga_1}\dots\oint_{\ga_{N-n_0}} \prod_{i<j}\frac{f(u_i-u_j)}{f(u_i-u_j-2\eta)} \prod_{i=1}^{N-n_0} \Biggl(\psi_{\nu_i}(u_i) f(\la-u_i+p_{\nu_i}+2\eta+4\eta(N-i)-2\eta\La_{[0,\nu_i)})\\ \times \frac{f(u_i-p_0)}{f(u_i-p_0-2\eta n_0)}\cdot \frac{f(\la+u_i-q_0+2\eta(N-n))}{f(u_i-q_0)}\prod_{j=1}^n \frac{f(u_i-v_j+2\eta)}{f(u_i-v_j)}
\Biggr)\frac{du_i}{2\pi \i}\,.
\end{multline}

The $\nu_i$'s that remain in the integral are all $\ge 1$, thus the integrand has no poles at $u_i=q_0$, $1\le i\le N-n_0$. Let us deform the outer-most contour $\ga_1$ to the boundary of the fundamental parallelogram of Definition \ref{df:adm}. The integral along that boundary vanishes (because of the double periodicity of the integrand), and hence the result is the negative sum of the residues at $u_1=v_t$, $1\le t\le n$. Looking at $-\Res_{u_1=v_t}$, let us record the $u_1$-dependent terms (including the $(-1)$ in front of the residue):
\begin{multline*}
(-1)\prod_{j=2}^{N-n_0}\frac{f(v_t-u_j)}{f(v_t-u_j-2\eta)}\cdot\psi_{\nu_1}(v_t)f(\la-v_t+p_{\nu_1}+2\eta+4\eta(N-1)-2\eta\La_{[0,\nu_1)})\\ \times \frac{f(v_t-p_0)}{f(v_t-p_0-2\eta n_0)}\cdot \frac{f(\la+v_t-q_0+2\eta(N-n))}{f(v_t-q_0)}\cdot
\frac{f(2\eta)}{f'(0)}\prod_{\substack{1\le j\le n\\ j\ne t}} \frac{f(v_t-v_j+2\eta)}{f(v_t-v_j)}\,.
\end{multline*}

Substituting this expression into \eqref{eq:D-int-aux} we see that all factors of the form $f(u_i-v_t)$ for $2\le i\le (N-n_0)$ cancel out, and the remaining integrand is of the same kind as the original one, but with $u_1$ and $v_t$ removed. Repeating the same computation the total of $(N-n_0)$ times, we obtain that the integral in \eqref{eq:D-int-aux} together with the last product in the first line equals (uniting all $t$'s that we meet when computing residues at $u_i=v_t$ into a set $I\subset\{1,\dots,n\}$)
\begin{multline*}
\left(-\frac{f(2\eta)}{f'(0)}\right)^{N-n_0}\sum_{\substack{I\subset\{1,\dots,n\}\\ |I|=N-n_0}}\prod_{j\notin 
I}\frac{f(v_j-p_0-2\eta n_0)}{f(v_j-p_0)} \prod_{i\in I}
\frac{f(\la+v_i-q_0+2\eta (N-n))}{f(v_i-q_0)}\prod_{\substack{i\in I\\j\notin 
I}}\frac{f(v_j-v_i-2\eta)}{f(v_j-v_i)}\\
\times \sum_{\substack{\sigma:\{1,\dots,N-n_0\}\to I\\ \sigma\ \text{is a 
bijection}}}\sigma\left(\prod_{i<j}\frac{f(v_i-v_j+2\eta)}{f(v_i-v_j)}
\prod_{i=1}^{N-n_0}\psi_{\nu_i}(v_i)f(\la 
-v_i+p_{\nu_i}+2\eta+4\eta(N-i)-2\eta\La_{[0,\nu_i)})\right).
\end{multline*}
Substituting this into \eqref{eq:D-int-aux} and using \eqref{eq:squared-norm} yields \eqref{eq:D-symm} after elementary cancellations (note the additional shift of $\la$ in the left-hand side of \eqref{eq:D-int}). 
\end{proof}

\section{Graphical representation: an IRF model}\label{sc:IRF}

In the context of the (higher spin) six vertex model, see, e.g., \cite[Sections 2-3]{BP} for a detailed description, the action of the traditional algebraic Bethe ansatz $A,B,C,D$ operators (the matrix elements of the monodromy matrix) on the basis vectors in an evaluation Verma module are commonly pictured by vertices of the square lattice together with a collection of arrows entering and exiting the vertex, see Figure \ref{fig:vertex}.  
\begin{figure}
\includegraphics[scale=1]{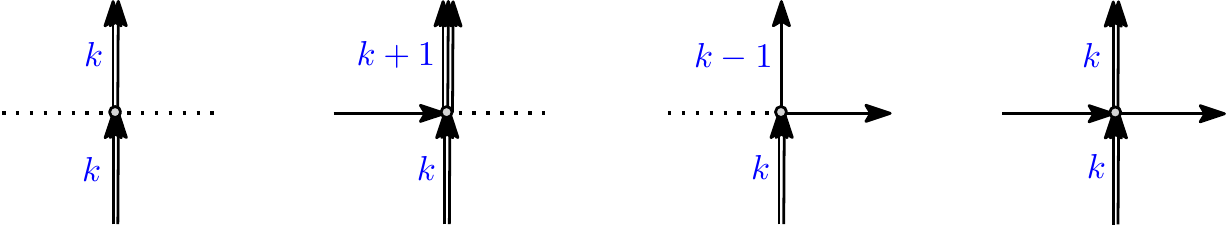}
\caption{Vertex representation of the higher spin six vertex $A,B,C,D$ operators.}\label{fig:vertex}
\end{figure}
This action is similar to (and is a degeneration of) our \eqref{eq:abcd}. For each vertex, the number of incoming arrows from the bottom corresponds to the index of the basis vector $e_k$ that an operator is being applied to, the number of outgoing arrows at the top corresponds to the index of the resulting basis vector $e_{l}$ with $l\in\{k-1,k,k+1\}$, and the weight associated to the vertex is equal to the coefficient of the $e_{l}$. The left edge may be occupied by a single incoming arrow, the right edge may be occupied by a single outgoing arrow, and their joint configuration determines which of the four operators $A,B,C,D$ is acting. Note that the total number of the incoming arrows is always equal to the total number of the outgoing ones. 

There is also a spectral parameter (similar to $w$ in \eqref{eq:abcd}) that is associated to the row in which the vertex is located, and also a spin (or a highest weight) parameter and an inhomogeneity parameter associated to its column (similar to $\Lambda$ and $z$ in \eqref{eq:abcd}). If the spin parameter is such that the corresponding Verma module has an irreducible finite-dimensional quotient, one speaks about a finite spin model and restricts the multiplicity of the vertical arrows to the indices of the basis vectors in this finite-dimensional quotient (the action of $A,B,C,D$ is also considered in this quotient). The case of dimension 2 quotient, often referred to as the \emph{spin $\frac 12$} situation, leads to six possible vertex types, i.e., to the six vertex model. 

Taking tensor products of highest weight modules corresponds to stacking vertices horizontally, while taking compositions of the $A,B,C,D$ operators corresponds to stacking vertices or strips of those vertically. For example, recalling the notation for basis vectors in tensor products of evaluation Verma modules from Section \ref{sc:inf-volume}, the coefficient of $E_\nu$ in a composition of four $B$-operators applied to $E_\mu$ can be represented as the sum of products of all vertex weights over the configurations pictured in Figure \ref{fig:paths_B}, where the boundary conditions, i.e., the arrows at the edges of the rectangle, are fixed, while the types of the inside vertices may vary. The spectral parameters of the $B$-operators correspond to the rows of the rectangle (their order is immaterial because these operators commute), while the spin and inhomogeneity parameters correspond to the columns and describe the modules in whose tensor product the action is happening. Note that the empty vertices can be 
ignored if they have weight 1, i.e. the $A$-operators are normalized so that $Ae_0=e_0$. 

\begin{figure}
\includegraphics[scale=1.2]{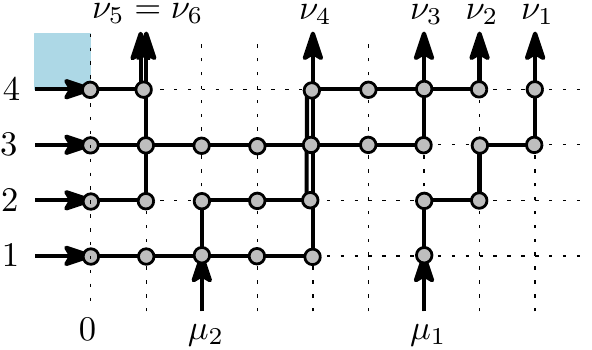}
\caption{Graphical representation of the action of four $B$-operators.}\label{fig:paths_B}
\end{figure}

We would like to extend this convenient graphical interpretation to the action of the operators $a,b,c,d$ of \eqref{eq:abcd}. The basic difficulty lies in the presence of the additional parameter $\la$. It can be taken into account with the help of an \emph{Interaction-Round-a-Face model} (IRF model, for short), see e.g. \cite[Chapter 13]{Bax} for a general discussion of such models. The name \emph{Solid-on-Solid} or \emph{SOS model}  is also used. 

\begin{figure}
\includegraphics[scale=0.75]{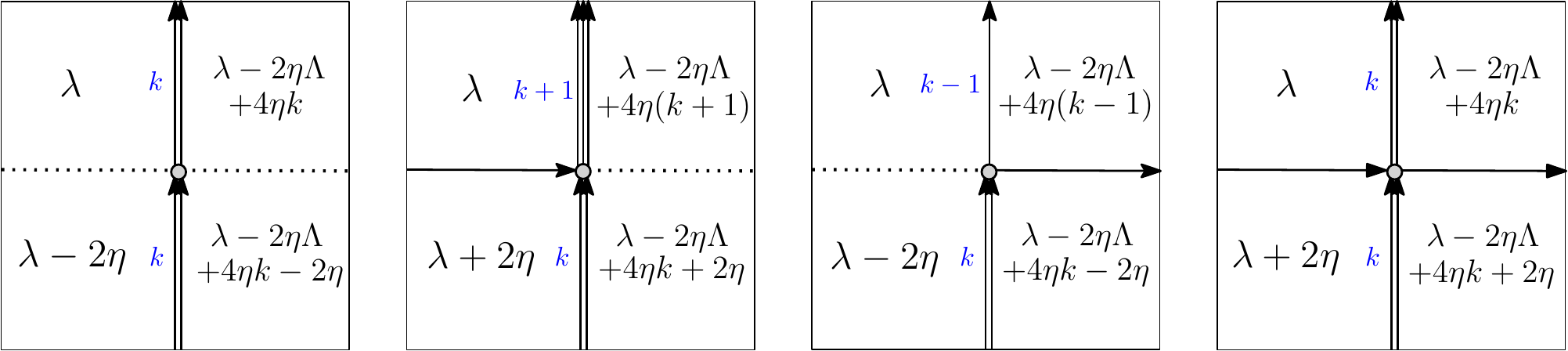}
\caption{The IRF plaquettes corresponding to the vertices in Figure \ref{fig:vertex}.}
\label{fig:vertices-IRF}
\end{figure}

As a first step, we replace the elementary vertices of Figure \ref{fig:vertex} by $2\times 2$ plaquettes that surround the vertex, see Figure \ref{fig:vertices-IRF}. We place a parameter in each of the four unit squares; we shall call it the \emph{dynamic parameter} or the \emph{filling} of a unit square. 

The top left filling in the plaquettes of Figure \ref{fig:vertices-IRF} is always chosen as $\lambda$; it corresponds to the $\la$ used in the definition of the action in  \eqref{eq:abcd}. The other three fillings are determined by the following rules: (a) crossing a $k$-fold vertical arrow left-to-right adds $4\eta k -2\eta \La$ to the dynamic parameter, where $\La$ is the parameter of the evaluation Verma module associated to the column in which the arrow lives; (b) crossing a $\delta$-fold arrow top-to-bottom ($\delta$ is either 0 or 1) adds $(2\delta-1)2\eta$ to the dynamic parameter. Figure \ref{fig:vertices-IRF} shows that these rules are consistent. As before, we also have a spectral parameter $w$ associated to the row of the vertex, as well as the highest weight parameter $\La$ and an inhomogeneity parameter $z$ associated to the column of the vertex. They are all used in evaluating the weight of the plaquette, which, by definition, is the coefficient of $e_*$ in the right-hand side of \eqref{eq:abcd}
. 

Observe that once we know the fillings of the unit squares, the arrows can be removed from the picture because they can be reconstructed uniquely. Further, one can assign these fillings to the vertices of the dual square lattice, and the weight of the plaquette (which becomes an elementary square of the dual lattice) is determined by the fillings of the four vertices surrounding it. This explains the `Interaction-Round-A-Face' term. However, it will be more convenient for us to work with the original lattice and the arrows on its edges. 

The case of $\Lambda=1$ deserves a special attention --- it corresponds to the dimension 2 irreducible quotient of $V_\La(z)$, and thus extends the six vertex model. The six possibilities for the plaquettes in this case are pictured in Figure \ref{fig:vertices-6v}. Note that the fillings of neighboring unit squares always differ by $\pm 2\eta$. 

\begin{figure}
\includegraphics[scale=1]{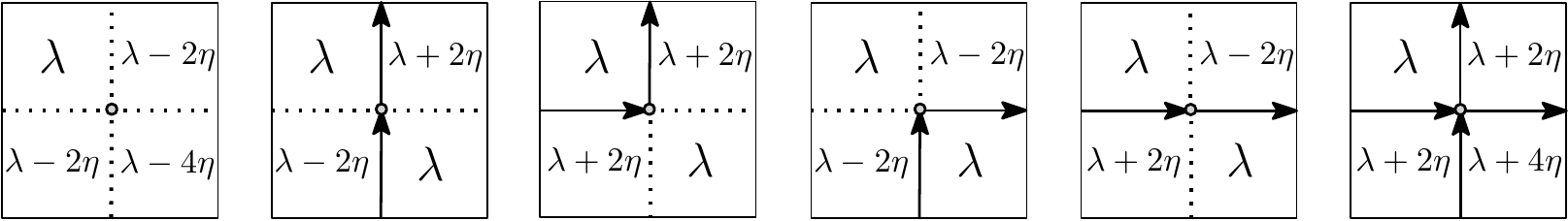}
\caption{The spin $\frac 12$ IRF vertices.}
\label{fig:vertices-6v}
\end{figure}

It is immediate to verify that tensor products \eqref{eq:tensor} are modeled by stacking the plaquettes horizontally, and compositions of operators \eqref{eq:tilde-abcd} correspond to stacking the plaquettes vertically, thus directly extending the graphical interpretation for the higher spin six vertex model discussed above. One only needs to make sure that the top-left unit square of the resulting rectangle (painted blue in Figure \ref{fig:paths_B}) is filled by $\la$. As an example, we have the following statement. 

\begin{proposition}\label{pr:B-picture} For any nonnegative signatures $\nu$ and $\mu$, the coefficient $B_{\nu/\mu}(\la;w_1,\dots,w_n)$ is equal to the sum of products of weights of all plaquettes over all possible vertex configurations of the type pictured in Figure \ref{fig:paths_B} with the prescribed boundary conditions, $n$ rows, row spectral parameters $w_1,\dots,w_n$, and the (painted) top left unit square filled by $\la$ (the dynamic parameters of the other boxes is then uniquely reconstructed from the vertex configuration). 
\end{proposition}  

Note that in the finite spin case, $\La_j\equiv I+(m+l\tau)/2\eta$ for $I\ge 0,m,l$ all integers, if $\nu$ in the above proposition has no parts of multiplicity $\ge (I+1)$, then no path configuration with a vertical edge of multiplicity $\ge(I+1)$ can give to a nonzero contribution to $B_{\nu/\mu}(\la;w_1,\dots,w_n)$. This is due to the vanishing of the right-hand side of the third relation in \eqref{eq:abcd} for $k=I+1$ because of the $f(2(\La+1-k)\eta)$ factor. 

\section{A pre-stochastic IRF model}\label{sc:stoch-IRF} In what follows we shall call a parameter-dependent matrix \emph{pre-stochastic} if it becomes stochastic when the parameters are such that all its matrix elements are nonnegative. In other words, a pre-stochastic matrix is a matrix all of whose row-sums are equal to one. Similarly, we shall use the term `pre-stochastic' for other objects that become stochastic or probabilistic under the condition of nonnegativity of some relevant quantities. 

The relation \eqref{eq:pieri2} states that the matrix with rows parameterized by signatures $\nu$ of length $\ell(\nu)=N\ge 1$, columns parameterized by signatures $\ka$ of length $\ell(\ka)=N+1$, and matrix elements 
\begin{equation}\label{eq:B-conj}
const(\lambda,u,v_1,\dots,v_l)\cdot \frac{D_{\ka}(\la-2\eta l;v_1,\dots,v_l)}{D_{\nu}(\la+2\eta-2\eta l ;v_1,\dots,v_l)}\,B_{\ka/\nu}(\la;u) 
\end{equation}
is pre-stochastic. Our current goal is to create a \emph{local} (pre-)Markov chain out of it, where the locality means that the matrix elements depend only on the interaction of neighboring parts of the interlacing signatures $\ka\succ\nu$. The skew-function $B_{\ka/\nu}(\la+2\eta l;u)$ is a local quantity --- it can be obtained by multiplying weights of the plaquettes stacked horizontally in a single row, see the previous section. On the other hand, $D_{\nu}(*;v_1,\dots,v_l)$ and $D_\ka(*;v_1,\dots,v_l)$ in the above expression are in general highly nonlocal, cf. Theorem \ref{th:symm}(ii). We will thus look for a way of simplifying these quantities. 

What we do next mimics, to a large extent, the remarkable `specialization $\rho$' introduced in \cite[Sections 6.3]{BP, BP-lect} in the context of the higher spin six vertex model. 

Observe that the way that the $v_j$'s enter the integrand of the integral representation \eqref{eq:D-int} is through the product of factors of the form $f(u_i-v_j+2\eta)/f(u_i-v_j)$ (which, in view of the orthogonality Theorem \ref{th:orth}, come from the right-hand side of the Cauchy identity \eqref{eq:cauchy}). We can simplify this product by the following substitution:
\begin{equation}
\label{eq:simplification}
\left(\prod_i\prod_{j=1}^n \frac{f(u_i-v_j+2\eta)}{f(u_i-v_j)}\right)\Biggl |_{(v_1,\dots,v_n)=(v,v-2\eta,\dots,v-2\eta(n-1))}=\prod_i \frac{f(u_i-v+2\eta n)}{f(u_i-v)}.
\end{equation}
Next, we choose $v=p_0+2\eta n$ so that $f(u_i-v+2\eta)=f(u_i-p_0)=\psi_0(v)$, which simplifies the integrand of \eqref{eq:D-int} further if the signature $\nu$ has zero parts. The last step is to get rid of the $n$-dependence in the integrand, which is contained in the expression (the denominator is from \eqref{eq:simplification})
\begin{equation}\label{eq:n-dep}
\prod_i\frac{f(\la+u_i-q_0+2\eta(N-n))}{f(u_i-p_0-2\eta n)}\,.
\end{equation}

One possibility is to assume that $2\eta$ is an element of $\C/(\Z+\tau \Z)$ of finite order, i.e., $f(z+2\eta m)\equiv f(z)$ for any $x\in \C$ and some $m\in\Z$. This is certainly very interesting, and we hope to return to this case in a future work, but this is not what we consider below. 

Instead, we will now restrict ourselves to the trigonometric case arising as $\tau\to+\i \infty$, i.e., from now on we assume that $f(x)=\sin \pi x$, cf. Section \ref{sc:prelim}. Then, assuming $\eta\notin\R$, we can actually take a limit of \eqref{eq:n-dep} as $n\to\infty$ and obtain a constant. Since the value of this constant is irrelevant to us, we will simply replace \eqref{eq:n-dep} by 1. There is also an $n$-dependent prefactor of the integral in the right-hand side of \eqref{eq:D-int}, which can be removed by passing from $D_\nu$'s to $\Dnorm_\nu$'s of \eqref{eq:BD-norm}. 
This leads the following definition of the `specialization $\rho$'. 

\begin{definition}\label{df:rho}
In the trigonometric case $f(z)= \sin\pi z$, and assuming the admissibility of the parameters in the sense of Definition \ref{df:adm}, we define, for any $\nu=(\nu_1\ge\dots\ge\nu_N\ge 0)$, $\Dnorm_\nu(\la;\rho)$ via the following integral, cf. \eqref{eq:BD-norm}, \eqref{eq:D-int}:
\begin{multline}
\label{eq:D-rho-int}
\Dnorm_\nu(\la;\rho)=
\frac{(-1)^N(f(2\eta))^N}{\prod_{i=0}^{N-1}f(\la+2\eta i)\cdot c_{\nu}(\la)}\,\oint_{\ga_1}\dots\oint_{\ga_N} \prod_{i<j}\frac{f(u_i-u_j)}{f(u_i-u_j-2\eta)}\\ \times \prod_{i=1}^N \Biggl(\psi_{\nu_i}(u_i) f(\la-u_i+p_{\nu_i}+2\eta+4\eta(N-i)-2\eta\La_{[0,\nu_i)})\cdot  \frac{f(u_i-p_0)}{f(u_i-q_0)}
\Biggr)\frac{du_i}{2\pi \i}\,,
\end{multline}
where $c_\nu$ is given by \eqref{eq:squared-norm}, and the integration contours are as in Definition \ref{df:adm}. We also agree that $\Dnorm_\varnothing(\la;\rho)=1$. 
\end{definition}
Note that following our logic, by comparison with \eqref{eq:D-int} $\Dnorm_\nu(\la;\rho)$ should have been $\Dnorm_\nu(\la-2\eta n;\rho)$, but since we took $n$ to infinity, we also removed it from the notation. 

Let us now check that the result of the above simplification indeed leads to something reasonably simple. 

\begin{proposition}\label{pr:D-rho-expl} In the notation and assumptions of Definition \ref{df:rho}, for any $\nu=(\nu_1\ge\dots\ge\nu_N\ge 0)$ we have
\begin{equation} \label{eq:D-rho-expl}
\Dnorm_\nu(\la;\rho)=\mathbf{1}_{\nu_N\ge 1}\cdot \frac{(-1)^N (f(2\eta))^N}{\pi^N\cdot c_{\nu}(\la)}\cdot\prod_{i=0}^{N-1} \frac{f(\la-2\eta\La_0+2\eta(i+1))}{f(\la+2\eta i)}\,,
\end{equation}
where $c_\nu$ is given by \eqref{eq:squared-norm}. 

\end{proposition}

\begin{proof} We need to evaluate the right-hand side of \eqref{eq:D-rho-int}. Observe that if $\nu_N=0$, then $\psi_{\nu_N}(u_N)=1/(f(u_N-p_0))$, and the $u_N$ integration contour $\ga_N$ has no singularities inside; hence, the integral vanishes. From now on assume that all $\nu_i$ are $\ge 1$. Then, using the definition \eqref{eq:phi-psi} of $\psi_{\nu_i}(u_i)$, we see that the integrand has no poles at $u_i=q_0$, and thus it has no poles outside $\gamma_1$  and inside the fundamental strip. Let us now deform the contours to the `boundary' of the fundamental strip one by one starting with $\gamma_1$, which is the outermost one. 

We deform $\gamma_1$ to the (positively oriented) boundary of the rectangle with the left and right sides going along the boundaries of the fundamental strip of Definition \ref{df:adm} (it has width 1), and the top and bottom sides joining the strip boundaries at heights $\Im u_1=-M$ and $\Im u_1=M$. The integrals along the boundaries of the strip cancel each other because of the 1-periodicity of the integrand, cf. the beginning of the proof of Theorem \ref{th:orth}. On the other hand, to evaluate the integrals along the top and bottom sides of the rectangle we send $M\to\infty$ and use (recall that we are in the trigonometric case with $f(z)= \sin \pi z$)
\begin{equation}\label{eq:lim-sin}
\frac{f(x+\i y-a)}{f(x+\i y-b)}\longrightarrow\begin{cases}e^{\i\pi(a-b)},&y\to+\infty, \\
e^{-\i\pi(a-b)}, &y\to -\infty.
\end{cases}
\end{equation}
Hence, the $u_1$-dependent part of the integrand tends, as $u_1\to\pm \i\infty$ along a vertical line $\Re u_1=const$, to
$$
-\exp(\pm\i\pi(\la+2\eta N-2\eta\La_0)). 
$$
Since this expression is independent of $\Re u_1$, its integral over the top and bottom sides of the rectangle is simply the negative difference of the corresponding limiting values, which after the normalization by $2\pi\i$ gives 
\begin{equation}\label{eq:sine}
\frac{\sin\pi(\la+2\eta N-2\eta\La_0)}{\pi}\,.
\end{equation} 
Repeating the argument for integration over $u_2,\dots,u_N$ (in that order) yields \eqref{eq:D-rho-expl}.
\end{proof} 

Specializing into $\rho$ also simplifies the Cauchy identity \eqref{eq:cauchy}.

\begin{proposition}\label{pr:cauchy-rho} In the trigonometric case $f(z)= \sin\pi z$, assuming the admissibility of the parameters in the sense of Definition \ref{df:adm}, and with $\Dnorm_\nu(\la;\rho)$ as in \eqref{eq:D-rho-int} or \eqref{eq:D-rho-expl}, for any $u_j$'s close enough to the $\{p_j\}_{j\ge 0}$ we have
\begin{equation}\label{eq:cauchy-rho}
\sum_{\ka} \Dnorm_\ka(\la;\rho)\Bnorm_\ka(\la;u_1,\dots,u_N)=(-f(2\eta))^{N}\prod_{i=1}^N \frac{f(u_i-p_0)}{f(u_i-q_0)}\,.
\end{equation}
\end{proposition} 
\begin{proof} We just need to take the limit of \eqref{eq:cauchy}. The assumptions make sure that the summation converges uniformly throughout the limit transition. Indeed, one easily sees from Theorem \ref{th:symm} that $B_\nu$ and $D_\nu$ can be both bounded by $const^{\sum \nu_i}$ as $\nu$ gets large, and by taking $u_i$'s close enough to the $p_j$'s (that are sufficiently close to each other by the assumptions of Definition \ref{df:adm}) one makes sure that $B_\nu$'s decay faster than $D_\nu$'s may grow. 

As we discussed just before Definition \ref{df:rho}, one uses the integral representation \eqref{eq:D-int} for the $D_\nu$'s, specializes $v_j$'s as in $\eqref{eq:simplification}$, and takes $n\to\infty$ in \eqref{eq:n-dep}. The same specialization and limit transition happens in the right-hand side of \eqref{eq:cauchy}. The integral representation \eqref{eq:D-int} readily implies that the series remains uniformly convergent throughout the limit. The condition $\Im\eta\ne 0$ required for the $n\to\infty$ limit of \eqref{eq:n-dep} to exist, may be dropped in the end as both sides of \eqref{eq:cauchy-rho} are analytic in $\eta$. 
\end{proof}

Let us now proceed towards the goal declared in the beginning of the section, i.e., let us investigate what happens to the pre-stochastic matrix \eqref{eq:B-conj} when we specialize the $v_j$'s into $\rho$. 

\begin{definition}\label{df:B-stoch} In the trigonometric case $f(z)\equiv \sin\pi z$, define 
the (pre-)stochastic skew $B$-functions as
\begin{equation}\label{eq:B-stoch}
\Bstoch_{\ka/\nu}(\la;u_1,\dots,u_k)=\frac{(-1)^k}{(f(2\eta))^k}\prod_{i=1}^k \frac{f(u_k-q_0)}{f(u_k-p_0)}\cdot \frac{\Dnorm_\ka(\la;\rho)}{\Dnorm_\nu(\la+2\eta k;\rho)}\,\Bnorm_{\ka/\nu}(\la;u_1,\dots,u_k)
\end{equation}
for any signatures $\ka$ and $\nu$ all of whose parts are $\ge 1$, and with $\Dnorm_*(\la;\rho)$ as in \eqref{eq:D-rho-expl}.  
\end{definition}

One readily sees that these functions satisfy the same branching relation \eqref{eq:B-branching} as the usual skew $B$-functions. Also, since $\Dnorm_\varnothing(\la,\rho)\equiv 1$ (which is natural given that $\Dnorm_{\varnothing}(\la,v_1\dots,v_n)\equiv 1$, cf. \eqref{eq:D-empty}, \eqref{eq:BD-norm}), we see that \eqref{eq:cauchy-rho} can be restated as $\sum_\ka \Bstoch_\ka(\la;\rho)=1$. This is a special case of the following more general statement that also justifies the superscript in the notation for the $\Bstoch_{\ka/\nu}$'s. 

\begin{proposition}\label{pr:B-stoch-sum}
In the trigonometric case $f(z)= \sin\pi z$, and assuming the admissibility of the parameters in the sense of Definition \ref{df:adm}, for any $k\ge 1$ and any signature $\nu$ all of whose parts are $\ge 1$, we have
\begin{equation}\label{eq:B-stoch-sum}
\sum_\ka \Bstoch_{\ka/\nu} (\lambda;u_1,\dots,u_k)=1. 
\end{equation}
\end{proposition}
\begin{proof} Very similarly to the proof of Proposition \ref{pr:cauchy-rho} above, the result is obtained by taking the same limit transition of the Pieri type identity \eqref{eq:pieri2}. This leads to the special case of \eqref{eq:B-stoch-sum} with $k=1$, and the general case follows from the branching rule \eqref{eq:B-branching}. 
\end{proof} 

The skew $B$-functions can be defined via products of (local) plaquette weights in an IRF model, as stated in Proposition \ref{pr:B-picture}. The main result of the present section is a similar representation for the $\Bstoch_{\ka/\nu}$'s given below.

\begin{theorem}\label{th:B-stoch} Define the (pre-)stochastic weights of the four types of plaquettes in Figure \ref{fig:vertices-IRF} as the following four expressions, respectively (cf. \eqref{eq:abcd}):
\begin{equation}\label{eq:stoch-weights}
\begin{aligned}
\astoch_k(\la;w)&=\frac{f(z-w+(\La+1-2k)\eta)}{f(z-w+(\La+1)\eta)}\frac{f(-\la+2(\La+1-k)\eta)}{f(-\la+2(\La+1-2k)\eta)}\,,\\
\bstoch_k(\la;w)&=\frac{f(-\la+z-w+(\La-1-2k)\eta)}{f(z-w+(\La+1)\eta)}\frac{f(2(k-\La)\eta)}{f(\la-2(\La-1-2k)\eta)}\,,\\
\cstoch_k(\la;w)&=\frac{f(-\la-z+w+(\La+1-2k)\eta)}{f(z-w+(\La+1)\eta)}\frac{f(2k\eta)}{f(-\la+2(\La+1-2k)\eta)}\,,\\
\dstoch_k(\la;w)&=\frac{f(z-w+(-\La+1+2k)\eta)}{f(z-w+(\La+1)\eta)}\frac{f(\la+2(k+1)\eta)}{f(\la-2(\La-1-2k)\eta)}\,,
\end{aligned}
\end{equation}
where the spectral parameter $w$ corresponds to the row in which the center of the plaquette is located, and the highest weight $\La$ and the inhomogeneity parameter $z$ correspond to the column in which that center is located.

Then $\Bstoch_{\ka/\nu}(\la;w_1,\dots,w_k)$ is equal to the sum of products of the stochastic plaquette weights as described by Proposition \ref{pr:B-picture}, with the only other difference being that the left-most column must have coordinate 1 rather than 0, and the top left unit square is filled with $\la-2\eta\La_0$, see Figure \ref{fig:paths_B-stoch} for an illustration. 
\end{theorem}

\begin{figure}
	\includegraphics[scale=1.2]{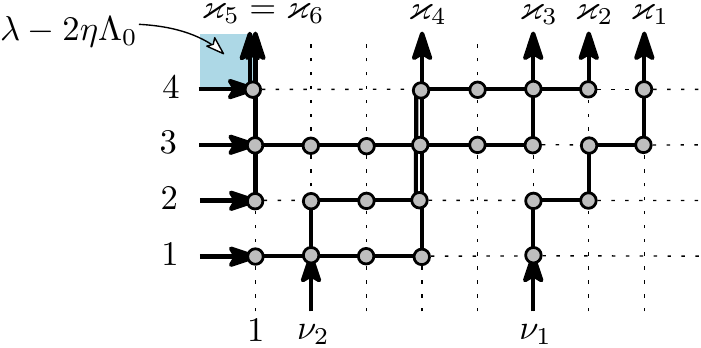}
	\caption{An IRF configuration for $\Bstoch_{\ka/\nu}(\la;w_1,\dots,w_k)$ with $k=4$.}\label{fig:paths_B-stoch}
\end{figure}

\begin{example}\label{ex:B-stoch-one-row} In the case $\ka=(K)$ for some $K\ge 1$, and  $\nu=\varnothing$, Theorem \ref{th:B-stoch} states that
\begin{multline}\label{eq:B-stoch-single-part}
\Bstoch_{(K)}(\la;u)=\Bstoch_{(K)/\varnothing}(\la;u)\\=\dstoch_0(\la-2\eta\La_0,u)\dstoch_0(\la-2\eta\La_{[0,1]},u)\cdots \dstoch_0(\la-2\eta\La_{[0,K-1)},u)\bstoch_0(\la-2\eta\La_{[0,K)},u),
\end{multline}
where the factors come from columns $1,2,\dots,K$, respectively (thus, one has to use the corresponding parameters $z_i,\La_i$, $1\le i\le K$, for the them when substituting \eqref{eq:stoch-weights}).
\end{example}

The reason we call the weights \eqref{eq:stoch-weights} `pre-stochastic' is the following: If one 
fixes the fillings of the three elementary squares in a $2\times 2$ plaquette that are located along the plaquette's bottom and left sides, then the fourth (top-right) filling has two possible values, cf. Figure \ref{fig:vertices-IRF}. The pre-stochasticity means that the two corresponding plaquette weights add up to one, which in terms of \eqref{eq:stoch-weights} means that
\begin{equation}\label{eq:weights-add-up}
\astoch_k(\la;w)+\cstoch_k(\la;w)\equiv 1,\qquad \bstoch_k(\la;w)+\dstoch_k(\la;w)\equiv 1,
\end{equation}
for all values of the parameters. These identities immediately follow from the following identity for $f(z)=\sin\pi z$:
\begin{equation}\label{eq:sine-identity}
f(B-C)f(w-A)=f(A-C)f(w-B)-f(A-B)f(w-C),\qquad A,B,C,w\in \C. 
\end{equation}

We call the IRF model with plaquette weights given by \eqref{eq:stoch-weights} \emph{pre-stochastic}. 

If these weights are actually nonnegative, and one considers the corresponding stochastic IRF model in a domain with prescribed boundary conditions along its left and bottom boundaries (that do not need to be straight), then the random state of the model can be constructed in a Markovian way by moving in the up-right direction from the boundary and sequentially deciding on the fillings of the new unit squares via independent Bernoulli trials corresponding to \eqref{eq:weights-add-up}. 

The same iterative procedure of moving away from the bottom-left boundary shows that for the pre-stochastic IRF and for any domain with fixed boundary conditions on its left and bottom boundaries, the partition function is always equal to 1. In particular, when the domain is a half-infinite strip of the form pictured in Figure \ref{fig:paths_B-stoch}, then, as Theorem \ref{th:B-stoch} shows, the equality \eqref{eq:B-stoch-sum} is exactly the fact that the partition function is equal to 1, and the assumptions on the parameters are necessary to make sure that the weights of the configurations with paths running infinitely far to the right are infinitesimally small. An important special case is when no arrows enter at the bottom boundary of the strip, and we dedicate a separate definition to it. 

\begin{definition}\label{df:quadrant-IRF} Assuming the weights \eqref{eq:stoch-weights} are nonnegative, we define the \emph{stochastic IRF model in the quadrant} as a probability measure on plaquette configurations in the quadrant $\R_{\ge 0}^2$ (the centers of the plaquettes, which are the inner vertices for the corresponding paths, range over $\Z_{\ge 1}^2$) defined in the Markovian way as described above, with no paths entering from the bottom into the vertices of the bottom row $\Z_{\ge 1}\times \{1\}$, and a single path entering from the left into each vertex of the leftmost column $\{1\}\times \Z_{\ge 1}$. 

The filling of the bottom-left unit box $[0,1]^2$ must also be specified, and denoting it by $\la^{(0)}$ we read from the boundary conditions on paths that the fillings of the unit boxes $[k,k+1]\times[0,1]$ in the bottom row are $(\la^{(0)}-2\eta\La_{[1,k]})$, while the fillings of the unit boxes $[0,1]\times [k,k+1]$ in the leftmost column are $(\la^{(0)}-2\eta k)$ for any $k\ge 0$. 

Assuming, in addition, that for some $k\ge 1$ the probability of configurations with paths that have segments of infinite length in rows $1,\dots,k$ is zero, and using the notation $(i_1,k+\frac 12),(i_2,k+\frac 12),\dots (i_k,k+\frac 12)$ with $i_1\ge i_2\ge\dots i_k\ge 1$ for the $k$ (random) intersection points of the paths with the horizontal line of coordinate $k+\frac 12$, we conclude from Theorem \ref{th:B-stoch} that
\begin{equation}\label{eq:prob=B-stoch}
\Prob\{\ka=(i_1,\dots,i_k)\}=\Bstoch_\ka(\la^{(0)}-2\eta(k-\La_0);w_1,\dots,w_k),
\end{equation}
where the first argument of the $\Bstoch_\ka$ in the right-hand side is chosen so that the filling of the unit box $[0,1]\times[k,k+1]$ is $(\la-2\eta\La_0)$, as required by Theorem \ref{th:B-stoch}.
\end{definition}

\begin{remark}\label{rm:finite-spin} In the finite spin case, $\La_j\equiv I+m/\eta$ for any $j\ge 1$ and integral $I\ge 0,m$, no plaquette configuration with a vertical edge occupied by $\ge (I+1)$ paths can give a nonzero contribution to $\Bstoch_{\ka/\nu}$ because of the vanishing of the right-hand side of \eqref{eq:stoch-weights} at $k=I$. 

Similarly, in that case the stochastic IRF model in the quadrant of Definition \ref{df:quadrant-IRF} has no vertical edges occupied my more than $I$ paths almost surely. In particular, for $I=1$ (the spin $\frac 12$ case), no more than one path can pass through any edge, vertical or horizontal, and the allowed plaquettes have the form pictured in Figure \ref{fig:vertices-6v}. 
\end{remark}

\begin{remark}\label{rm:stoch} The proof of Theorem \ref{th:B-stoch} that we give below does not clarify where the exact expressions for the weights \eqref{eq:stoch-weights} came from; let us try to offer an explanation. Denote by $a_k(\la,w),b_k(\la,w),c_k(\la,w),d_k(\la,w)$ the coefficients of $e_{*}$ in the right-hand sides of \eqref{eq:abcd} (these are the plaquette weights in the IRF model of Section \ref{sc:IRF}), and denote by $\widehat a_k(\la),\widehat b_k(\la),\widehat c_k(\la),\widehat d_k(\la)$ the respective ratios of \eqref{eq:stoch-weights} and the coefficients of \eqref{eq:abcd}, i.e., $\astoch_k(\la,w)=\widehat a_k(\la) a_k(\la,w)$, etc. Observe that we dropped the dependence on the spectral parameter $w$ in those ratios; indeed, a direct inspection shows that they are $w$-independent:

\begin{equation}\label{eq:hat-abcd}
\begin{aligned}
\widehat a_k(\la)&=\frac{f(\la)}{f(\la-2(\La+1-2k)\eta)}\frac{f(\la-2(\La+1-k)\eta)}{f(\la+2 k\eta)}\,,\\
\widehat b_k(\la)&=\frac{f(\la)}{f(\la-2 (\La-1-2k)\eta)}\frac{f(2(\La-k)\eta)}{f(2\eta)}\,,\\ 
\widehat c_k(\la)&=\frac{f(\la)}{f(\la-2 (\La+1-2k)\eta)}\frac{f(2\eta)}{f(2(\La+1-k)\eta)}\,,\\
\widehat d_k(\la)&=\frac{f(\la)}{f(\la-2 (\La-1-2k)\eta)}\frac{f(\la+2(k+1)\eta)}{f(\la-2(\La-k)\eta)}\,.
\end{aligned}
\end{equation}

If one imagines that a pre-stochastic IRF satisfying Theorem \ref{th:B-stoch} exists, then it is natural to assume the existence of such correction coefficients $\widehat a_k(\la),\widehat b_k(\la),\widehat c_k(\la),\widehat d_k(\la)$ from \eqref{eq:B-conj}; indeed, the $\Dnorm$ factors are independent of the spectral parameters, and the $p_0,q_0$-dependent factors there are simply responsible for removing the 0th column for passing from from Figure \ref{fig:paths_B} to Figure \ref{fig:paths_B-stoch}. Once we assume this existence, the desired stochasticity relations \eqref{eq:weights-add-up} read
$$
\widehat a_k(\la) a_k(\la;w)+\widehat c_k(\la) c_k(\la;w)\equiv 1,\qquad \widehat b_k(\la) b_k(\la;w)+\widehat d_k(\la) d_k(\la;w)\equiv 1,
$$
and the fact that they must hold for arbitrary $w\in\C$ uniquely determines possible candidates for $\widehat a_k(\la),\widehat b_k(\la),\widehat c_k(\la),\widehat d_k(\la)$. One still needs to verify that Theorem \ref{th:B-stoch} holds though. 
\end{remark}
\begin{proof}[Proof of Theorem \ref{th:B-stoch}] We first note that it suffices to consider the single variable case $k=1$; indeed, the case of several variables readily follows by applying the branching rule for the $\Bstoch$'s that is identical to the branching rule \eqref{eq:B-branching} for the usual skew $B$-functions, and observing that it agrees with the IRF prescription for computing the $\Bstoch$'s as described in Theorem \ref{th:B-stoch}. 

From now one use assume that we are in the single variable case $k=1$.  

We use the induction on $N:=\ell(\nu)$. Note that since $k=1$, we must have $\ell(\ka)=N+1$. 

The base of the induction is $N=0$, i.e., $\nu=\varnothing$ and $\ka=(K)$ for some $K\ge 1$. The value of $\Bstoch_{(K)}(\la;u)$ according to Theorem \ref{th:B-stoch} was computed in \eqref{eq:B-stoch-single-part}, and similarly we have 
$$
B_{(K)}(\la;u)=d_0(\la,u)d_0(\la-2\eta\La_{0},u)\cdots d_0(\la-2\eta\La_{[0,K-1)},u)b_0(\la-2\eta\La_{[0,K)},u),
$$ 
where we used the notation introduced in Remark \ref{rm:stoch}, and the factors in the right-hand side correspond to the IRF vertices in columns $0,1,\dots,K$, respectively. 
Hence, Theorem \ref{th:B-stoch} claims that
\begin{multline*}
\frac{\Bstoch_{(K)}(\la;u)}{B_{(K)}(\la;u)}=\frac{1}{d_0(\la,u)}\,\widehat d_0(\la-2\eta\La_{0},u)\cdots \widehat d_0(\la-2\eta\La_{[0,K-1)},u)\widehat b_0(\la-2\eta\La_{[0,K)},u)\\
=\frac{f(u-q_0)}{f(u-p_0)}\frac{f(\la)}{f(2\eta)}\frac{f(\la+2\eta-2\eta\La_0)f(2\eta\La_k)}{f(\la+2\eta-2\eta\La_{[0,k)})f(\la+2\eta-2\eta\La_{[0,k]})}\,,
\end{multline*}
where the factors in the first line correspond to columns $0,1,\dots,k$, and we used \eqref{eq:hat-abcd}. 

According to Definition \ref{df:B-stoch}, this needs to be compared to 
$$
\frac{-f(\la)}{f(2\eta)}\frac{f(u-q_0)}{f(u-p_0)}\frac{\Dnorm_{(K)}(\la;\rho)}{\Dnorm_{\varnothing}(\la+2\eta;\rho)}
$$
(note the appearance of $f(\la)=\Bnorm_{(K)}(\la;u)/B_{(K)}(\la;u)$, cf. \eqref{eq:BD-norm}), which is immediate as 
$$
\Dnorm_{(K)}(\la;\rho)=\left(\frac{f(\la+2\eta-2\eta\La_{[0,K)})f(\la+2\eta-2\eta\La_{[0,k]})}{-f(2\eta\La_k)}\right)^{-1} \cdot f(\la+2\eta-2\eta\La_0)
$$
by \eqref{eq:D-rho-expl} and \eqref{eq:squared-norm} (we use $f'(0)=\pi$), and $\Dnorm_\varnothing(\la+2\eta;\rho)=1$ by Definition \ref{df:rho}. 

Let us proceed to the induction step $\ell(\nu)\leadsto \ell(\nu)+1$. Our strategy is the following.
We aim to prove that for the IRF representations of $\Bstoch_{\ka/\nu}(\la;u)$ and $B_{\ka/\nu}(\la;u)$, the contributions coming from equivalent path configurations match. More exactly, we think of the path configurations pictured in Figures \ref{fig:paths_B} and \ref{fig:paths_B-stoch} as `equivalent' if they are only different by the 0th column than contains only vertices with one incoming arrow on the left and one outgoing arrow on the right (so that the configurations pictured there are indeed equivalent). The contribution to $\Bstoch_{\ka/\nu}(\la;u)$ is the product of the stochastic IRF weights \eqref{eq:stoch-weights} corresponding to the configuration in Figure \ref{fig:paths_B-stoch}. Similarly, the contribution to $B_{\la/\nu}(\la;u)$ is the product of the IRF weights of Section \ref{sc:IRF} corresponding to the configuration in Figure \ref{fig:paths_B}, but it needs to be re-normalized according to the right-hand side of \eqref{eq:B-stoch}. 

In the base case $\nu=\varnothing$ considered above, there was only one possible path configuration, and our computation above proved the desired match. 

In order to increase $\ell(\nu)$, we will examine the effect of removing the right-most path that joins $\nu_1$ and $\ka_1$ in two equivalent path configurations. The new path configurations correspond to a term in the IRF representations for $\Bstoch_{\widetilde\ka/\widetilde\nu}$ and 
$B_{\widetilde\ka/\widetilde\nu}$, where $\widetilde\ka$ and $\widetilde\nu$ are obtained from $\ka$ and $\nu$ by removing the largest coordinates $\ka_1$ and $\nu_1$, respectively. Note that $\ell(\widetilde\nu)=\ell(\nu)-1$. 

Our goal is to verify that as a result of such a removal, the contributions of the two path configurations get modified by the same correcting factor when we compare $\ka/\nu$ to $\widetilde\ka/\widetilde\nu$. Clearly, this will imply our induction step. 

It is convenient to introduce some notation that we will use in our computations. 

Denote the product of the stochastic plaquette weights that correspond to a path configuration for $\Bstoch_{\ka/\nu}$ by $\Wstoch$, the product of the non-stochastic plaquette weights that correspond to the equivalent path configuration for $B_{\ka/\nu}$ by $W$, and denote similar products that correspond to path configurations with removed right-most paths for $\Bstoch_{\widetilde\ka/\widetilde\nu}$ and $B_{\widetilde\ka/\widetilde \nu}(\la;u)$ by $\tWstoch$ and $\tW$, respectively. Furthermore, set
$$
C=\frac{-f(u-q_0)f(\la)}{f(u-p_0)f(2\eta)}\cdot \frac{\Dnorm_\ka(\la;\rho)}{\Dnorm_\nu(\la+2\eta;\rho)},\qquad \widetilde C=\frac{-f(u-q_0)f(\la)}{f(u-p_0)f(2\eta)}\cdot \frac{\Dnorm_{\widetilde\ka}(\la;\rho)}{\Dnorm_{\widetilde\nu}(\la+2\eta;\rho)};
$$   
these are the conjugating factors from the right-hand side of \eqref{eq:B-stoch}. We want to show that $\Wstoch/\tWstoch=(W/\tW)\cdot (C/\widetilde C)$ or, equivalently, 
\begin{equation}\label{eq:correction-factors}
\frac{\Wstoch/W}{\tWstoch/\tW}=\frac{C}{\widetilde C}=\frac{\Dnorm_\ka(\la;\rho)\Dnorm_{\widetilde\nu}(\la+2\eta;\rho)}{\Dnorm_\nu(\la+2\eta;\rho)\Dnorm_{\widetilde\ka}(\la;\rho)}=\frac{c_\nu(\la+2\eta) c_{\widetilde\ka}(\la)}{c_\ka(\la) c_{\widetilde\nu}(\la+2\eta)}\,,
\end{equation}
where we use \eqref{eq:D-rho-expl} and \eqref{eq:squared-norm}.

We will always denote by $\widetilde\la$ the filling of the top left unit square in the  plaquette centered at the left-most vertex of the path that is being removed, and we will also denote by $x$ the horizontal coordinate of the same vertex. For the computation of the right-hand side of \eqref{eq:correction-factors} it will be useful to observe that, with the notation $\ka=1^{k_1}2^{k_2}\cdots$, $\nu=1^{n_1}2^{n_2}\cdots$, $\widetilde\ka=1^{\tilde k_1}2^{\tilde k_2}\cdots$, $\widetilde\nu=1^{\tilde n_1}2^{\tilde n_2}\cdots$, if no arrow enters vertex $x$ from the left then 
\begin{multline*}
\la+2\eta(2k_{<x}-\La_{[0,x)})=\la+2\eta(2\tilde k_{<x}-\La_{[0,x)})\\
=\la+2\eta (2(n_{<x}+1)-\La_{[0,x)})=\la+2\eta(2(\tilde n_{<x}+1)-\La_{[0,x)})=\widetilde \la,
\end{multline*}
which readily follows from the rule (a) in the IRF model description in Section \ref{sc:IRF} and the fact that $n_{<x}=\tilde n_{<x}=k_{<x}-1=\tilde k_{<x}-1$. On the other hand, if there is an arrow that enters vertex $x$ from the left then similar arguments show that
\begin{multline*}
\la+2\eta(2k_{<x}-\La_{[0,x)})=\la+2\eta(2\tilde k_{<x}-\La_{[0,x)})\\
=\la+2\eta (2 n_{<x}-\La_{[0,x)})=\la+2\eta(2 \tilde n_{<x}-\La_{[0,x)})=\widetilde \la,
\end{multline*}
because $n_{<x}=\tilde n_{<x}=k_{<x}=\tilde k_{<x}$.

There are four distinct combinatorial possibilities for the right-most path pictured in Figure \ref{fig:right-path} (we do not draw the complete path configuration there, just the plaquettes that surround the vertices occupied by the right-most paths). These four cases require different computations, and we will consider them one by one. 

\begin{figure}
	\includegraphics[scale=0.5]{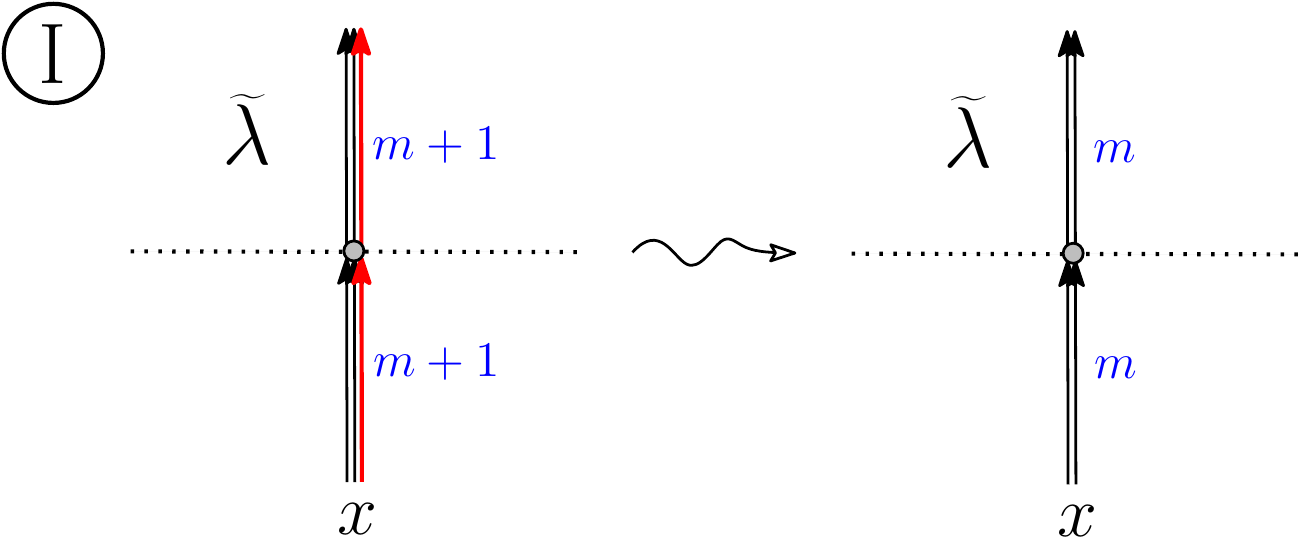}\qquad \includegraphics[scale=0.5]{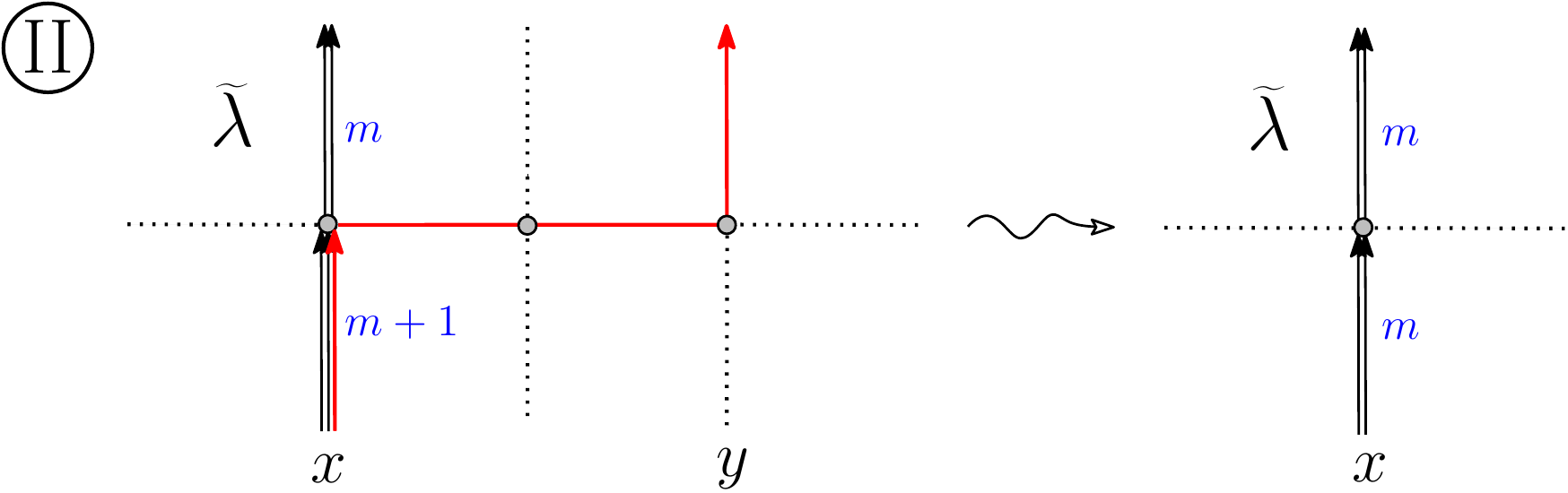}
	\vskip 0.2cm
	\includegraphics[scale=0.5]{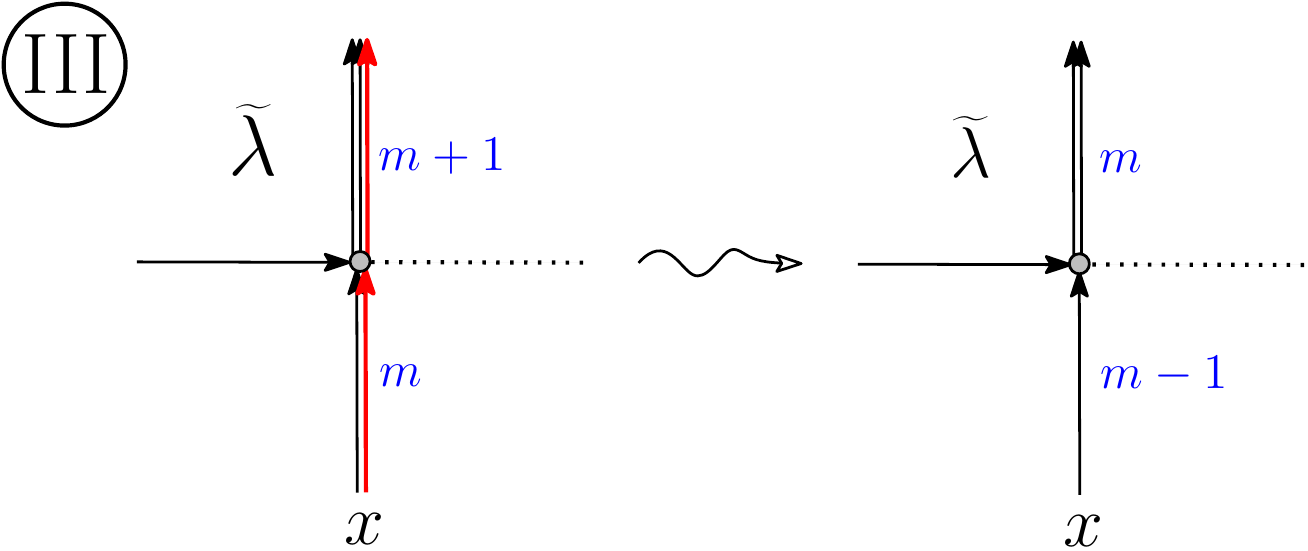}\qquad \includegraphics[scale=0.5]{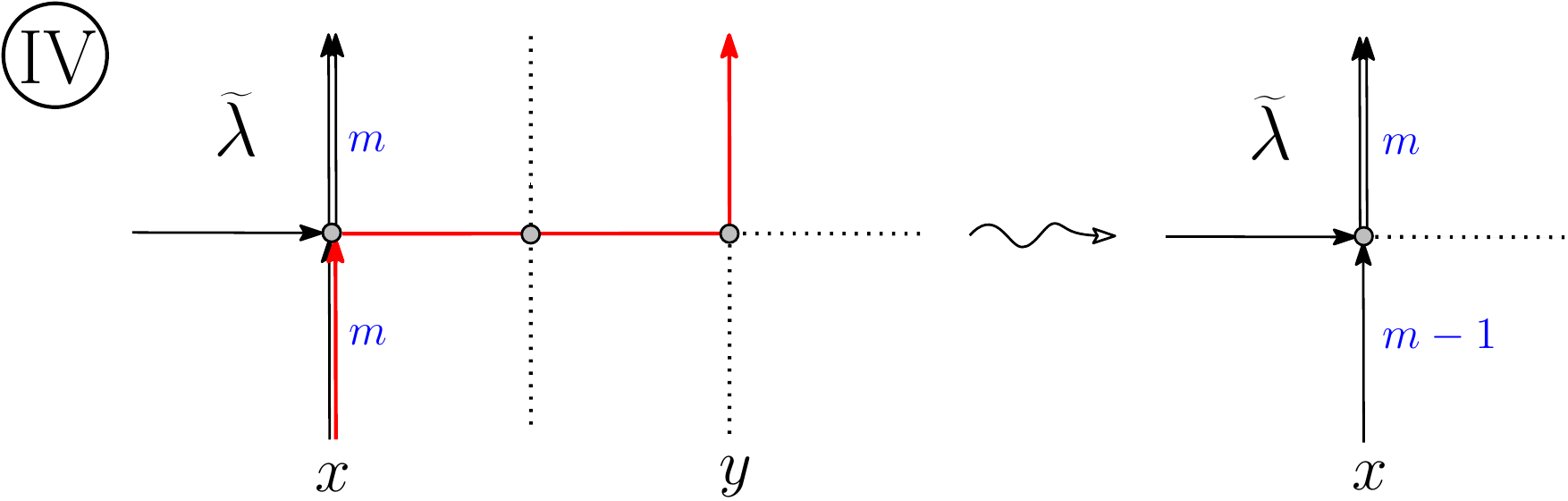}

	\caption{Four possibilities for the right-most path (in red) that is being removed.}\label{fig:right-path}
\end{figure}

\smallskip

\emph{Case I} (see Figure \ref{fig:right-path}). This is the case when the right-most path does not move to the right, and both vertical edges occupied by it are shared by the total of $m+1$ paths for some $m\ge 0$. 

The left-hand side of \eqref{eq:correction-factors} gives
\begin{equation}\label{eq:corr1}
\frac{\widehat a_{m+1}(\widetilde\la)}{\widehat a_{m}(\widetilde\la)}=\frac{f(\widetilde\la-2\eta(\La_x+1-2m))}{f(\widetilde\la-2\eta(\La_x-1-2m))}\frac{f(\widetilde\la-2\eta(\La_x-m))}{f(\widetilde\la-2\eta(\La_x+1-m))}\frac{f(\widetilde\la+2\eta m)}{f(\widetilde\la+2\eta(m+1))}\,.
\end{equation}

On the other hand, the right-hand side of \eqref{eq:correction-factors} gives, cf. \eqref{eq:squared-norm},
\begin{multline*}
\left(\prod_{j=0}^{m}\frac{f(\widetilde\la +2\eta(m+j-\La_x))f(\widetilde\la+2\eta j)}{f(2\eta(\La_x-j))}\right)
\left(\prod_{j=0}^{m}\frac{f(\widetilde\la+2\eta(m+1+j-\La_x))f(\widetilde\la+2\eta(j+1))}{f(2\eta(\La_x-j))}\right)^{-1}\\ \times
\left(\prod_{j=0}^{m-1}\frac{f(\widetilde\la +2\eta(m+j-\La_x))f(\widetilde\la+2\eta(j+1))}{f(2\eta(\La_x-j))}\right)
\left(\prod_{j=0}^{m-1}\frac{f(\widetilde\la+2\eta(m-1+j-\La_x))f(\widetilde\la+2\eta j)}{f(2\eta(\La_x-j))}\right)^{-1}
\end{multline*}
with the four expressions coming from $c_\nu(\la+2\eta),c_\ka^{-1}(\la),c_{\widetilde\ka}(\la),c_{\widetilde\nu}^{-1}(\la+2\eta)$, respectively. Canceling coinciding factors leads to \eqref{eq:corr1}. 

\smallskip

\emph{Case II} (see Figure \ref{fig:right-path}). This is the case of the right-most path starting with a vertical edge (at location $x$), moving $(y-x)\ge 1$ steps to the right, and finishing with another vertical edge (at location $y$), with the condition that no arrow enters the vertex at $x$ from the left. We denote the number of the arrows along the initial edge of the right-most path by $(m+1)\ge 1$. 

The left-hand side of \eqref{eq:correction-factors} gives
\begin{equation}\label{eq:corr2}
\frac{\widehat c_{m+1}(\widetilde \la)}{\widehat a_m(\widetilde\la)}\cdot
\widehat d_0(\widetilde\la+2\eta (2m-\La_x))\cdots \widehat d_0(\widetilde\la+2\eta( 2m-\La_{[x,y-1)}))\cdot \widehat b_0(\widetilde \la+2\eta (2m-\La_{[x,y)}))
\end{equation}
with the factors coming from columns $x,x+1,\dots,y$. The column $x$ contribution is (using \eqref{eq:hat-abcd})
\begin{equation}\label{eq:corr2a}
\frac{\widehat c_{m+1}(\widetilde \la)}{\widehat a_m(\widetilde\la)}=\frac{f(\widetilde\la)f(2\eta)}{f(\widetilde\la-2\eta(\La_x-1-2m))f(2\eta(\La_x-m))}\cdot \frac{f(\widetilde\la-2\eta(\La_x+1-2m))f(\widetilde\la+2\eta m)}{f(\widetilde\la)f(\widetilde\la-2\eta(\La_x+1-m))}\,,
\end{equation}
while the contribution of columns from $x+1$ to $y$ is (similarly to Example \ref{ex:B-stoch-one-row})
$$
\frac{f(\widetilde\la+2\eta(2m+1-\La_x))f(\widetilde\la+2\eta(2m-\La_x))f(2\eta\La_y)}{f(\widetilde\la+2\eta(2m+1-\La_{[x,y)}))f(\widetilde\la+2\eta(2m+1-\La_{[x,y]}))f(2\eta)}\,.
$$
The product of these two expressions is \eqref{eq:corr2}, and it needs to be compared to the right-hand side of \eqref{eq:correction-factors}, which is (using \eqref{eq:squared-norm})  
\begin{multline*}
\left(\prod_{j=0}^{m}\frac{f(\widetilde\la +2\eta(m+j-\La_x))f(\widetilde\la+2\eta j)}{f(2\eta(\La_x-j))}\right) \Biggl(\prod_{j=0}^{m-1}\frac{f(\widetilde\la +2\eta(m+j-\La_x))f(\widetilde\la+2\eta(j+1))}{f(2\eta(\La_x-j))}\\ \times
\frac{f(\widetilde\la +2\eta(2m+1-\La_{[x,y]}))f(\widetilde\la+2\eta(2m+1-\La_{[x,y)}))}{f(2\eta\La_y)} \Biggr)^{-1}\\ \\\times \left(\prod_{j=0}^{m-1}\frac{f(\widetilde\la +2\eta(m+j-\La_x))f(\widetilde\la+2\eta(j+1))}{f(2\eta(\La_x-j))}\right)
\left(\prod_{j=0}^{m-1}\frac{f(\widetilde\la+2\eta(m-1+j-\La_x))f(\widetilde\la+2\eta j)}{f(2\eta(\La_x-j))}\right)^{-1}
\end{multline*}
with the four expressions coming from $c_\nu(\la+2\eta),c_\ka^{-1}(\la),c_{\widetilde\ka}(\la),c_{\widetilde\nu}^{-1}(\la+2\eta)$, respectively. Direct comparison yields the desired equality. 

\smallskip

\emph{Case III} (see Figure \ref{fig:right-path}). This is the case of the right-most path not moving horizontally at all but, unlike Case I, the multiplicities of the two vertical edges it occupies are now $m$ and $m+1$ (hence, there is a horizontal arrow entering the vertex at $x$ from the left). 

The left-hand side of \eqref{eq:correction-factors} is (using \eqref{eq:hat-abcd})
\begin{equation*}
\frac{\widehat b_{m}(\widetilde\la)}{\widehat b_{m-1}(\widetilde\la)}=\frac{f(2\eta(\La_x-m))}{f(2\eta(\La_x-m+1))}\frac{f(\widetilde\la-2\eta(\La_x+1-2m))}{f(\widetilde\la-2\eta(\La_x-1-2m))}\,.
\end{equation*}

On the other hand, the right-hand side of \eqref{eq:correction-factors} is (using \eqref{eq:squared-norm})
\begin{multline*}
\left(\prod_{j=0}^{m-1}\frac{f(\widetilde\la +2\eta(m+1+j-\La_x))f(\widetilde\la+2\eta (j+2))}{f(2\eta(\La_x-j))}\right)\\ \times\left(\prod_{j=0}^{m}\frac{f(\widetilde\la+2\eta(m+1+j-\La_x))f(\widetilde\la+2\eta(j+1))}{f(2\eta(\La_x-j))}\right)^{-1}\\ \times
\left(\prod_{j=0}^{m-1}\frac{f(\widetilde\la +2\eta(m+j-\La_x))f(\widetilde\la+2\eta(j+1))}{f(2\eta(\La_x-j))}\right)
\left(\prod_{j=0}^{m-2}\frac{f(\widetilde\la+2\eta(m+j-\La_x))f(\widetilde\la+2\eta (j+2))}{f(2\eta(\La_x-j))}\right)^{-1}
\end{multline*}
with the four expressions coming from $c_\nu(\la+2\eta),c_\ka^{-1}(\la),c_{\widetilde\ka}(\la),c_{\widetilde\nu}^{-1}(\la+2\eta)$, respectively. 

Once again, the two expressions are readily seen to coincide. 

\smallskip

\emph{Case IV} (see Figure \ref{fig:right-path}). This is the case of the right-most path starting with a vertical edge (at location $x$), moving $(y-x)\ge 1$ steps to the right, and finishing with another vertical edge (at location $y$), with the condition that there is an arrow that enters the vertex at $x$ from the left. We denote the number of arrows along the initial edge of the right-most path by $m\ge 1$. 

The left-hand side of \eqref{eq:correction-factors} gives, cf. \eqref{eq:corr2},
\begin{equation}\label{eq:corr4}
\frac{\widehat d_{m}(\widetilde \la)}{\widehat b_{m-1}(\widetilde\la)}\cdot
\widehat d_0(\widetilde\la+2\eta (2m-\La_x))\cdots \widehat d_0(\widetilde\la+2\eta( 2m-\La_{[x,y-1)}))\cdot \widehat b_0(\widetilde \la+2\eta (2m-\La_{[x,y)}))
\end{equation}
with the factors coming from columns $x,x+1,\dots,y$. The column $x$ factor is (using \eqref{eq:hat-abcd})
$$
\frac{\widehat d_{m}(\widetilde \la)}{\widehat b_{m-1}(\widetilde\la)}=\frac{f(\widetilde \la)f(\widetilde\la+2\eta(m+1))}{f(\widetilde\la-2\eta(\La_x-1-2m))f(\widetilde\la-2\eta(\La_x-m))}\cdot\frac{f(\widetilde\la-2\eta(\La_x+1-2m))f(2\eta)}{f(\widetilde\la)f(2\eta(\La_x-m+1))}\,.
$$
The contribution of the vertices with coordinates $x+1,\dots,y$ is exactly the same as in Case II, which is \eqref{eq:corr2a}. The product of these two expressions needs to be compared to the right-hand side of \eqref{eq:correction-factors}, which is (using \eqref{eq:squared-norm})
\begin{multline*}
\left(\prod_{j=0}^{m-1}\frac{f(\widetilde\la +2\eta(m+1+j-\La_x))f(\widetilde\la+2\eta (j+2))}{f(2\eta(\La_x-j))}\right)\\ \times \Biggl(\prod_{j=0}^{m-1}\frac{f(\widetilde\la +2\eta(m+j-\La_x))f(\widetilde\la+2\eta(j+1))}{f(2\eta(\La_x-j))} 
\frac{f(\widetilde\la +2\eta(2m+1-\La_{[x,y]}))f(\widetilde\la+2\eta(2m+1-\La_{[x,y)}))}{f(2\eta\La_y)} \Biggr)^{-1}\\ \\\times
\left(\prod_{j=0}^{m-1}\frac{f(\widetilde\la +2\eta(m+j-\La_x))f(\widetilde\la+2\eta(j+1))}{f(2\eta(\La_x-j))}\right)\left(\prod_{j=0}^{m-2}\frac{f(\widetilde\la+2\eta(m+j-\La_x))f(\widetilde\la+2\eta (j+2))}{f(2\eta(\La_x-j))}\right)^{-1}
\end{multline*}
with the four expressions coming from $c_\nu(\la+2\eta),c_\ka^{-1}(\la),c_{\widetilde\ka}(\la),c_{\widetilde\nu}^{-1}(\la+2\eta)$, respectively. Canceling common factors shows that the two expressions are once again the same. 

The finishes the proof of our induction step, and the proof of Theorem \ref{th:B-stoch} is now complete. 
\end{proof}

\section{Degenerations}\label{sc:degen}

\subsection{The higher spin six vertex model}\label{ss:hs6v} Consider the higher spin inhomogeneous six vertex model as described in \cite[Sections 2-4]{BP}. It assigns weights to configurations of up-right paths in the square grid of the type discussed in Section \ref{sc:IRF} above, and the weight of a configuration is the product of weights of its vertices. In their turn, the vertex weights are defined by the following explicit formulas: 
\begin{align}\label{eq:hs6v-weights}
	\begin{array}{rclrcl}
		w(k,0;k,0)&=&\dfrac{1-s q^{k}\xi u}{1-s \xi u},&\qquad \qquad
		\rule{0pt}{22pt}
		w(k,1;k+1,0)&=&\dfrac{1-q^{k+1}}{1-s\xi  u},\\
		\rule{0pt}{22pt}
		w(k,0;k-1,1)&=&\dfrac{(1-s^{2}q^{k-1})\xi u}{1-s \xi u},&\qquad \qquad
		\rule{0pt}{22pt}
		w(k,1;k,1)&=&\dfrac{\xi u-s q^{k}}{1-s\xi u},
	\end{array}
\end{align}
the notation $w(i_1,j_1;i_2,j_2)$ stands for the weight of the vertex with $i_1$ arrows entering from below, $j_1$ arrows entering from the left, $i_2$ arrows exiting at the top, $j_2$ arrows exiting to the right, $u$ is the (spectral) parameter that depends on the row in which the vertex resides, and $s, \xi$ are the spin and the inhomogeneity parameters that depend on the column in which the vertex resides. The four equations in \eqref{eq:hs6v-weights} correspond to the four vertices in Figure \ref{fig:vertices-IRF} above. 

These vertex weights also have (pre-)stochastic versions given by 
\begin{align}\label{eq:hs6v-stoch-weights}
	\begin{array}{rclrcl}
		L(k,0;k,0)&=&\dfrac{1-s q^{k}\xi u}{1-s \xi u},&\qquad 
		\rule{0pt}{22pt}
		L(k,1;k+1,0)&=&\dfrac{1-s^{2}q^{k}}{1-s\xi u},\\
		\rule{0pt}{22pt}
		L(k,0;k-1,1)&=&\dfrac{-s\xi u+s q^{k} \xi u}{1-s\xi u},&\qquad 
		\rule{0pt}{22pt}
		L(k,1;k,1)&=&\dfrac{-s\xi u+s ^{2}q^{k}}{1-s\xi u},
	\end{array}
\end{align}
The stochasticity is expressed in the readily visible relation $\sum_{i_2,j_2}L(i_1,j_1;i_2,j_2)\equiv 1$. 

The above vertex weights can be obtained from our plaquette weights $a_k(\la,w),b_k(\la,w)$, $c_k(\la,w),d_k(\la,w)$ and $\astoch_k(\la,w),\bstoch_k(\la,w),\cstoch_k(\la,w),\dstoch(\la,w)$ (see Section \ref{sc:stoch-IRF} for their definitions) by the following limit transition. 

\begin{proposition} In the trigonometric case $f(x)= \sin \pi x$, with the following identification of the IRF and the higher spin six vertex models' parameters:
\begin{equation}\label{eq:parameters}
e^{2 \pi\i\eta\La}=s,\quad e^{-4\pi \i\eta}=q, \quad e^{2\pi \i z}=\xi, \quad e^{2\pi \i(\eta-w)}=u,
\end{equation}
we have
\begin{equation}\label{eq:lim-to-hs6v}
\lim_{\la\to -\i\infty}\left[\begin{matrix} a_k(\la,w)& b_k(\la,w)\\ c_k(\la,w)& d_k(\la,w)\end{matrix}\right]=
\left[\begin{matrix} q^{-k} w(k,0;k,0)&  \dfrac{q-1}{q^{k/2+1}(1-q^{k+1})}\, w(k,1;k+1,0)\\ 
\dfrac{1-q^k}{s q^{3(k-1)/2}(1-q)}\, w(k,0;k-1,1) 
&-s^{-1}q^{-k} w(k,1;k,1)\end{matrix}\right]
\end{equation}
and 
\begin{equation}\label{eq:lim-to-stoch-hs6v}
\lim_{\la\to -\i\infty}\left[\begin{matrix} \astoch_k(\la,w)& \bstoch_k(\la,w)\\ \cstoch_k(\la,w)& \dstoch_k(\la,w)\end{matrix}\right]=
\left[\begin{matrix} L(k,0;k,0)& L(k,1;k+1,0)\\ L(k,0;k-1,1) 
& L(k,1;k,1)\end{matrix}\right].
\end{equation}
\end{proposition}
\begin{proof} Straightforward computations using the definitions of the IRF weights, see \eqref{eq:abcd}, \eqref{eq:stoch-weights}, and the second line of \eqref{eq:lim-sin}. 
\end{proof}

The second group of limiting relations \eqref{eq:lim-to-stoch-hs6v} says that such a limit of our pre-stochastic IRF model of Section \ref{sc:stoch-IRF} with the row and column parameter matching as in \eqref{eq:parameters} is exactly the stochastic higher spin six vertex model from \cite{BP}. 

The relations \eqref{eq:lim-to-hs6v} also have a meaning. Observe that the four expressions in the right-hand side of \eqref{eq:lim-to-hs6v} can be written in a unified way as 
$$
\frac{[i_1]_q!}{q^{i_2}[i_2]_q!} \,s^{-\mathbf{1}_{j_2=1}} (-1)^{\mathbf{1}_{j_1=1}}\cdot w(i_1,j_1;i_2,j_2),
$$
where $(i_1,j_1;i_2,j_2)$ take all allowed values, and we use the notation $q^{\pm 1/2}=e^{\mp2\pi\i \eta}$ and 
$$
[m]_q=\frac{q^{m/2}-q^{-m/2}}{q^{1/2}-q^{-1/2}},\qquad [m]_q!=[1]_q[2]_q\cdots [m]_q,\qquad m=1,2,\dots.
$$
Via Proposition \ref{pr:B-picture} and the corresponding \cite[Definition 4.4]{BP}, this immediately implies the following limiting relation: For any nonnegative signatures $\mu=0^{m_0}1^{m_1}2^{m_2}\cdots$ and $\nu=0^{n_0}1^{n_1}2^{n_2}\cdots$ with $\ell(\nu)=N$, $\ell(\mu)=M$, $N-M=k\ge 1$, as long as the row and column parameters of the IRF model of Section \ref{sc:IRF} and the higher spin six vertex model of \cite{BP} match as in \eqref{eq:parameters}, we have
\begin{equation}\label{eq:B-to-F}
\lim_{\la\to -\i\infty} B_{\nu/\mu}(\la;w_1,\dots,w_k)=(-1)^k q^{\frac{k(N+M+1)}2}(-s)^{\sum_i\nu_i-\sum_i\mu_i}
\prod_{j\ge 0}\frac{[n_j]_q!}{[m_j]_q!}\cdot F_{\nu/\mu}(u_1,\dots,u_k),
\end{equation}
where $F_{\nu/\mu}$'s are the analogs of the $B_{\nu/\mu}$'s in the context of the higher spin six vertex model, see \cite[Section 4]{BP} for details. The correction factors in the right-hand side of \eqref{eq:B-to-F} were the reason for us not to use $F_{*/*}$ notation for our $B$-functions. 

It is also not difficult to obtain similar limits for all the related objects (like $D_{\nu/\mu}$'s or the coefficients in the left-hand side of \eqref{eq:C-symm}), but we will refrain from doing that as we have no immediate applications for such formulas. 
 
\subsection{The dynamic stochastic six vertex model}\label{ss:6v-IRF} As was mentioned in Remark \ref{rm:finite-spin}, we can set $\La_j\equiv I+m/\eta$ for all $j\ge 1$ and integral $I\ge 0$ and $m$, with the result being that all vertical edges in the pre-stochastic IRF model in the quadrant of Definition \ref{df:quadrant-IRF} will be occupied by no more than $I$ paths. Let us set $I=1$; this is the spin $\frac 12$ case with all possible plaquettes being of the form pictured in Figure \ref{fig:vertices-6v}. We will use the symbols $\nul$, $\uu$, $\ru$, $\ur$, $\rr$, $\all$ to denote the corresponding plaquette types (the order is the same as in Figure \ref{fig:vertices-6v}), and we will use $\w_{\lambda,w}(\,\cdot\,)$ to denote the corresponding weights. 

Substituting $\La=1$ into \eqref{eq:stoch-weights}, and also using \eqref{eq:parameters} to rename the parameters, leads to
\begin{equation}\label{eq:dyn-stoch-6v}
\begin{gathered}
\w_{\la,w}(\nul)=\astoch_0(\la,w)=1,\qquad
\w_{\la,w}(\all)=\dstoch_1(\la,w)=1\\
\w_{\la,w}(\uu)=\astoch_1(\la,w)=\frac{f(z-w)}{f(z-w+2\eta)}\frac{f(\la-2\eta)}{f(\la)}=\frac{1-q^{\frac 12}\xi u}{1-q^{-\frac 12}\xi u}\cdot \frac{q^{-1}-e^{2\pi \i\la}}{1-e^{2\pi \i\la}}\,,\\
\w_{\la,w}(\ru)=\bstoch_0(\la,w)=\frac{f(-\la+z-w)}{f(z-w+2\eta)}\frac{f(-2\eta)}{f(\la)}=\frac{1-q^{-1}}{1-q^{-\frac 12}\xi u}\cdot \frac{q^{\frac 12}\xi u-e^{2\pi \i\la}}{1-e^{2\pi \i\la}}\,,\\
\w_{\la,w}(\ur)=\cstoch_1(\la,w)=\frac{f(\la+z-w)}{f(z-w+2\eta)}\frac{f(2\eta)}{f(\la)}=\frac{\bigl(q^{\frac 12}-q^{-\frac 12}\bigr)\xi u}{1-q^{-\frac 12}\xi u}\cdot \frac{(q^{\frac 12}\xi u)^{-1}-e^{2\pi \i\la}}{1-e^{2\pi \i\la}}\,,\\
\w_{\la,w}(\rr)=\dstoch_0(\la,w)=\frac{f(z-w)}{f(z-w+2\eta)}\frac{f(\la+2\eta)}{f(\la)}=\frac{q^{-1}-q^{-\frac 12}\xi u}{1-q^{-\frac 12}\xi u}\cdot \frac{q-e^{2\pi \i\la}}{1-e^{2\pi \i\la}}\,,
\end{gathered}
\end{equation}
where we use the notation $q^{\pm 1/2}=e^{\mp 2\pi \i\eta}$. 

In the limit $\la\to-\i\infty$ considered in the previous subsection, the very last factors in the right-hand sides of \eqref{eq:dyn-stoch-6v} tend to 1, and one obtains the weights of the stochastic six vertex model, cf. \cite{GS}, \cite{BCG}, \cite[Section 6.5]{BP-lect}. We add the word `dynamic' to the name of the model in order to reflect the presence of the dynamic parameter $\la$ that changes between different vertices (according to the plaquettes of Figure \ref{fig:vertices-6v}). Another suitable name for the same (pre-)stochastic system would be `spin $\frac12$ stochastic IRF model'. 

\subsection{The rational pre-stochastic IRF model}\label{ss:rat-IRF} All the objects that were considered in Sections \ref{sc:prelim} through \ref{sc:stoch-IRF} have well-defined limits as the participating parameters (except for $\tau$ and $\La_i$'s) tend to 0 at the same rate. One can then use the linear approximation $f(x)\sim \mathrm{const}\cdot x$ as $x\to 0$ and rewrite all the formulas in terms of rational functions; thus, this limit is called \emph{rational}. For the stochastic plaquette weights \eqref{eq:stoch-weights}, this amounts to simply removing the letter `$f$' from the formulas.

Let us write out the weights of the corresponding rational dynamic stochastic six vertex model. Slightly abusing the notation, we replace 
$2\eta \to\epsilon, \  \la\to\epsilon \la,\  z\to\epsilon z,\ w\to \epsilon w$, multiply fillings of the unit boxes by $\epsilon^{-1}$, take $\epsilon\to 0$, and denote the limiting plaquette weights by $\rw_{\la,w}(\,\cdot\,)$, similarly to the previous subsection. The plaquettes in this limit have the same form as those in Figure \ref{fig:vertices-6v} but with $2\eta=1$, and their weights are
\begin{equation}\label{eq:rat-dyn-stoch-weights}
\begin{gathered}
\rw_{\la,w}(\nul)=1,\qquad
\rw_{\la,w}(\all)=1,\\
\rw_{\la,w}(\uu)=\frac{(\la-1)(z-w)}{\la(z-w+1)}\,,\qquad
\rw_{\la,w}(\ru)=\frac{\la-z+w}{\la(z-w+1)}\,,\\ 
\rw_{\la,w}(\ur)=\frac{\la+z-w}{\la(z-w+1)}\,,\qquad
\rw_{\la,w}(\rr)=\frac{(\la+1)(z-w)}{\la(z-w+1)}\,.
\end{gathered}
\end{equation}
Note that if all the parameters are real, $(z-w)>0$, and $|\la|$ is sufficiently large, all these weights are positive, and the model is actually stochastic. 

\subsection{The dynamic simple exclusion processes}\label{ss:exclusion} Consider the dynamic six vertex model of Section \ref{ss:6v-IRF} in the quadrant (cf. Definition \ref{df:quadrant-IRF}), choose all the inhomogeneity parameters to be equal: $\xi_i\equiv \xi$, all the spectral 
parameters to be equal: $u_i\equiv u$, and set
$
\xi u=q^{-\frac 12}(1+(1-q)\epsilon)
$
with $0<\epsilon\ll 1$. Then, introducing the notation $\al=-e^{-2\pi \i\la}$ and redenoting $\w_{\la,w}$ by $\w_{\al,\epsilon}$, from \eqref{eq:dyn-stoch-6v} we obtain 
\begin{equation}\label{eq:weights-al-ep}
\begin{gathered}
\w_{\al,\epsilon}(\nul)=\w_{\alpha,\epsilon}(\all)=1, \qquad \w_{\al,\epsilon}(\ru)=1+O(\epsilon),\quad \w_{\al,\epsilon}(\ur)=1+O(\epsilon),\\
\w_{\al,\epsilon}(\uu)=\frac{q+\al}{1+\al}\,\epsilon+O(\epsilon^2),\qquad \w_{\al,\epsilon}(\rr)=\frac{1+\al q}{1+\al}\,\epsilon+O(\epsilon^2).
\end{gathered}
\end{equation}
The new dynamic parameter is $\al$, and it changes according to the plaquettes of six types pictured in Figure \ref{fig:vertices-6v-alpha} (this is simply Figure \ref{fig:vertices-6v} redrawn with the new notations $\al=-e^{-2\pi \i\la}$ and $q=e^{-4\pi \i\eta}$).
\begin{figure}
\includegraphics[scale=1]{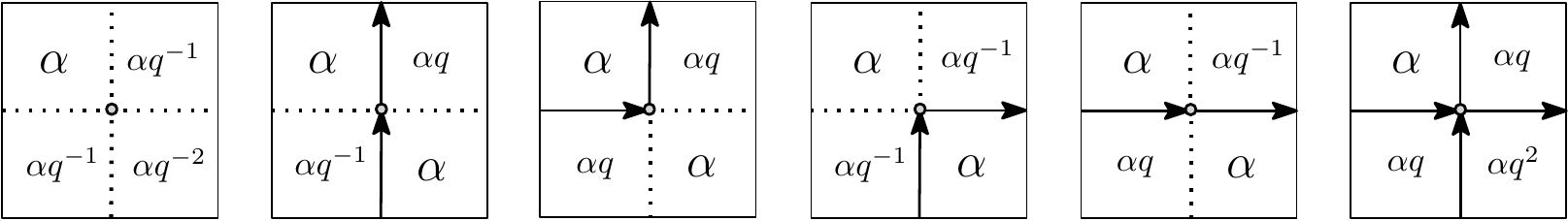}
\caption{Spin $\frac 12$ plaquettes in $(\alpha,q)$-notation.}
\label{fig:vertices-6v-alpha}
\end{figure}

Assuming that the parameters $q$ and $\al$ are such that the weights \eqref{eq:weights-al-ep} are always nonnegative (for example, $\al,q>0$), we see that for such a stochastic model in the quadrant, and for a small enough $\epsilon$, all paths will start with large deterministic pieces because of the smallness of 
$\w_{\al,\epsilon}(\uu)$ and $\w_{\al,\epsilon}(\rr)$; those will consist of strictly alternating vertices of types $\ru$ and $\ur$. An illustration of the finite neighborhood of the origin can be found in Figure \ref{fig:time-0}; there we use $\al_0$ to denote the dynamic $\al$-parameter  of the corner unit box $[0,1]\times [0,1]$. 

\begin{figure}
\includegraphics[scale=0.9]{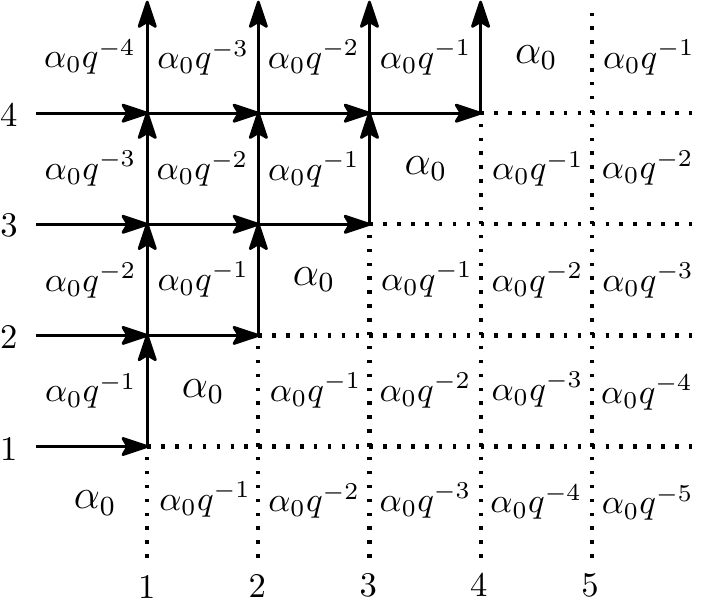}
\caption{The spin $\frac 12$ stochastic IRF configuration at $\epsilon=0$.}
\label{fig:time-0}
\end{figure}

However, as we consider the part of the quadrant of the form $[0,M]\times[0,M]$ with $M\sim \epsilon^{-1}$, finitely many vertices of types $\uu$ and $\rr$ will appear, with all of them being at finite distance from the diagonal of the quadrant. This is very similar to the limit from the stochastic six vertex model to the one-dimensional asymmetric simple exclusion process (ASEP, for short) that was considered in detail in \cite{Agg-ASEP}, where this limit transition was rigorously proven for arbitrary boundary/initial conditions (earlier mentions of this limit can be found in \cite{GS,BCG,BP,BP-lect}). It is possible to extend the proof of \cite{Agg-ASEP} to the case of the dynamic stochastic six vertex model (at least for the model in the quadrant), and we call the new limiting object \emph{the dynamic ASEP}. Let us describe it. We will only consider the initial condition of the dynamic ASEP that arises from the dynamic six vertex model in the quadrant. 

The state space of the dynamic ASEP originates from the fillings of the unit boxes along a horizontal row of the form $S_k=\R_{\ge 0}\times [k,k+1]$, $k\ge 0$, shifted by $k$ units to the left. In other words, we identify the $S_k$'s with (growing in $k$) subsets of a canonical strip $S=\R\times[0,1]$ by shifts chosen so that the diagonal square $[k,k+1]\times [k,k+1]$ of $S_k$ corresponds to the (independent of $k$) unit square $[0,1]\times[0,1]$ of $S$. 

The time $t$ state will come from $S_k$ with $k=t\epsilon^{-1}$ for any $t>0$, and the initial condition $t=0$ corresponds to $k=o(\epsilon)$. Figure \ref{fig:time-0} then shows that at time $t=0$, the box $[x,x+1]\times[0,1]$ of $S$ is filled by $\alpha_0 q^{-|x|}$. Clearly, at any time $t$ all the fillings will be from $\alpha_0q^\Z$, and we will encode a filling of $S$ by a sequence of integers $\{s_x\}_{x\in\Z}\subset \Z$ such that the filling of the box $[x,x+1]\times[0,1]$ is exactly $\alpha_0 q^{-s_x}$. In this notation, the initial condition at $t=0$ is the sequence $\{s_x=|x|\}_{x\in\Z}$. From the form of the plaquettes in Figure \ref{fig:vertices-6v-alpha} it is obvious that we always have $s_{x+1}-s_{x}\in\{-1,1\}$. 
The function $x\mapsto s_x$ can be viewed as a height function of the model, which is, however, slightly different from the height function that we use in Sections \ref{sc:observ} and \ref{sc:SSEP} below. 

The increments $\{s_{x+1}-s_{x}\}_{x\in\Z}$ can be encoded by a particle configuration in $\Z+\frac 12$, where we say that there is a particle that resides at $x+\frac 12$ if and only if 
$s_{x+1}-s_{x}=-1$. Our initial condition $\{s_x=|x|\}_{x\in\Z}$ then corresponds to particles filling up the negative semi-axis $\{-\frac 12,-\frac 32,-\frac 52, \dots\}$, which is the so-called step initial condition in the theory of exclusion processes.  

The time evolution corresponds to looking at higher and higher rows $S_k$ in the quadrant, and it is a continuous time Markov process on sequences $\{s_x\}_{x\in \Z}$. Its jumps correspond to the (rare) appearances of vertices of the types $\uu$ and $\rr$. The appearance of a $\uu$ vertex corresponds to a modification of the form $(s_x=i+1\mapsto s_x=i-1)$ or the jump of a particle from $x+\frac 12$ to $x-\frac 12$; the appearance of a $\rr$ vertex corresponds to a modification of the form $(s_x=i-1\mapsto s_x=i+1)$ or the jump of a particle from $x-\frac 12$ to $x+\frac 12$, see the second and fifth plaquettes in Figure \ref{fig:vertices-6v-alpha}. The rates of those jumps are the coefficients of $\epsilon$ in the second line of \eqref{eq:weights-al-ep} with $\al=\al_0q^{-s_x}$. 

Summarizing the above leads to the following definition. 

\begin{definition}\label{df:dynamic-ASEP} The dynamic ASEP is a continuous time Markov chain on the state space of integer-valued sequences $\{s_x\}_{x\in \Z}$ subject to the condition $s_{x+1}-s_x\in\{-1,1\}$ for any $x\in\Z$, that depends on parameters $q,\al\in\R$. Its elementary jumps are independent and have exponential waiting times with variable rates (of course, a jump is allowed only if it does not lead outside of the state space). These jumps can be of two kinds, and their form and the corresponding rates are

\smallskip

$s_x\mapsto (s_x-2)$\quad with rate\quad $\dfrac{q(1+\alpha q^{-s_x})}{1+\alpha q^{-s_x+1}}$
\qquad and \qquad
$s_x\mapsto (s_x+2)$ \quad with rate\quad $\dfrac{1+\al q^{-s_x}}{1+\al  q^{-s_x-1}}\,$,
\smallskip

\noindent where $x\in\Z$ is arbitrary. The parameters $q$ and $\al$ are assumed to be such that the above rates are always nonnegative, which is the case, for example, for $q,\al > 0$ and an arbitrary initial condition, or for $q>1$ and $\alpha>-q^{-c}$ and initial conditions satisfying $s_x(0)\ge |x|+c$ for any $x\in\Z$ and some constant $c\in\Z$. 
\end{definition}

Observe that setting $\al=0$ turns this Markov chain into the usual ASEP with the left jump rate equal to $q$ and the right jump rate equal to 1. Also, the multiplicative $q$-shift of the $\al$-parameter $\al\mapsto q\al$ is equivalent to the global shift $s_x\mapsto s_x-1$ for all $x\in \Z$.

The dynamic ASEP has a slightly simpler description in terms of the shifted sequences $\{\widetilde s_x:=s_z-\log _q|\alpha|\}$ that still satisfy $\widetilde s_{x+1}-\widetilde s_x\in\{-1,1\}$ but are no longer integer-valued. Then the elementary transitions and their rates are, taking the upper signs for $\al>0$ and the lower signs for $\al<0$,

\smallskip
$\widetilde s_x\mapsto (\widetilde s_x-2)$\quad with rate\quad $\dfrac{q(1\pm q^{-\widetilde s_x})}{1\pm q^{-\widetilde s_x+1}}$
\qquad and \qquad
$\widetilde s_x\mapsto (\widetilde s_x+2)$ \quad with rate\quad $\dfrac{1\pm q^{-\widetilde s_x}}{1\pm q^{-\widetilde s_x-1}}\,$.
\smallskip

\noindent This way the parameter $\al$ gets hidden in the choice of the initial condition and of $\pm$'s.

Let us now formally state the convergence of the dynamic stochastic six vertex model to the dynamic ASEP. 

\begin{proposition}\label{pr:6v-ASEP} Consider the dynamic stochastic six vertex model of Section \ref{ss:6v-IRF} in the quadrant, with the bottom-left unit box $[0,1]\times [0,1]$ filled by $\la$ (denote $\al=-e^{-2\pi\i\la}$), the inhomogeneity parameters $\xi_i\equiv\xi$, and the spectral parameters $u_i\equiv u$. Assume that $q,\al>0$, or $q>1$ and $\al>-1$.
Assume that $\xi u=q^{-\frac 12}(1+(1-q)\epsilon)$ for an $\epsilon>0$. Then for any $t_1,\dots,t_n\ge 0$ and $x_1,\dots,x_n\in \Z$, the random $n$-dimensional vector of fillings of the unit boxes
$$
[T_i+x_i,T_i+x_i+1]\times [T_i, T_i+1],\qquad\quad  T_i=[t_i\epsilon^{-1}],\quad 1\le i\le n, 
$$ 
converges in distribution and with all moments as $\epsilon\to +0$ to the random vector $\bigl(\al q^{-s_{x_i}(t_i)}\bigr)_{i=1}^n$, where $\{s_x(t)\}_{x\in\Z,t\ge 0}$ is the time $t$ state of the dynamic ASEP with parameters $(q,\al)$ started from $s_x(0)\equiv |x|$. 
\end{proposition}

The proof is similar to that of \cite[Theorem 3]{Agg-ASEP} and we omit it. 

Much of the same story applies to the rational stochastic IRF of Section \ref{ss:rat-IRF} converging to a dynamic version of the symmetric simple exclusion process (SSEP, for short), and we conclude this section by giving the analogs of Definition \ref{df:dynamic-ASEP} and Proposition \ref{pr:6v-ASEP}.

\begin{definition}\label{df:dynamic-SSEP} The dynamic SSEP is a continuous time Markov chain on the state space of integer-valued sequences $\{s_x\}_{x\in \Z}$ subject to the condition $s_{x+1}-s_x\in\{-1,1\}$ for any $x\in\Z$, that depends on a parameter $\la\in\R$. Its elementary jumps are independent and have exponential waiting times with variable rates. These jumps can be of two kinds, and their form and the corresponding rates are

\smallskip

$s_x\mapsto (s_x-2)$\quad with rate\quad $\dfrac{s_x-\la}{s_x-1-\la}$
\qquad and \qquad
$s_x\mapsto (s_x+2)$ \quad with rate\quad $\dfrac{s_x-\la}{s_x+1-\la}\,$,
\smallskip

\noindent where $x\in\Z$ is arbitrary. The parameter $\la$ is assumed to be such that the above rates are always nonnegative, which is the case, for example, for $\la<c$ for initial conditions satisfying $s_x(0)\ge |x|+c$ for any $x\in\Z$ and some constant $c\in\Z$. 
\end{definition} 

In terms of the shifted sequences $\{\widetilde s_x:=s_z-\lambda\}$ that still satisfy $\widetilde s_{x+1}-\widetilde s_x\in\{-1,1\}$ but are no longer integer-valued, the elementary jumps of the dynamic SSEP and their rates are

\smallskip
$\widetilde s_x\mapsto (\widetilde s_x-2)$\quad with rate\quad $\dfrac{\widetilde s_x}{\widetilde s_x-1}$
\qquad and \qquad
$\widetilde s_x\mapsto (\widetilde s_x+2)$ \quad with rate\quad $\dfrac{\widetilde s_x}{\widetilde s_x+1}\,$.
\smallskip

\begin{proposition}\label{pr:6v-SSEP} Consider the rational stochastic IRF model of Section \ref{ss:rat-IRF} in the quadrant, with the bottom-left unit box $[0,1]\times [0,1]$ filled by $\la<0$, the inhomogeneity parameters $z_i\equiv z$, and the spectral parameters $w_i\equiv w$.
Assume that $z-w=\epsilon$ for an $\epsilon>0$. Then for any $t_1,\dots,t_n\ge 0$ and $x_1,\dots,x_n\in \Z$, the random $n$-dimensional vector of fillings of the unit boxes
$$
[T_i+x_i,T_i+x_i+1]\times [T_i, T_i+1],\qquad\quad  T_i=[t_i\epsilon^{-1}],\quad 1\le i\le n, 
$$ 
converges in distribution and with all moments as $\epsilon\to +0$ to the random vector $(\la-s_{x_i}(t_i))_{i=1}^n$, where $\{s_x(t)\}_{x\in\Z,t\ge 0}$ is the time $t$ state of the dynamic SSEP with parameter $\la$ started from $s_x(0)\equiv |x|$. 
\end{proposition}

As for Proposition \ref{pr:6v-ASEP}, the proof is similar to that of \cite[Theorem 3]{Agg-ASEP} and we omit it. 

One can also view Proposition \ref{pr:6v-SSEP} as a limiting version of Proposition \ref{pr:6v-ASEP} as $\al$ approaches $-1$ and $q$ approaches $1$ at the same rate.  

\section{Observables}\label{sc:observ}

Let us now return to the (pre-)stochastic IRF model in the quadrant, as described in Definition \ref{df:quadrant-IRF}, and let us use $\la$ to denote the filling of the corner unit box $[0,1]\times [0,1]$. We aim at defining a family of observables for this model and computing their averages. 

For any $x,N\in \{1,2,\dots\}$ and an IRF configuration in the quadrant, define $\HT(x,N)$ as the number of paths that pass through or below the vertex with coordinates $(x,N)$. Clearly, the quadrant boundary conditions imply $0\le\HT(x,N)\le N$, and it suffices to know the types of the finitely many vertices in the rectangle $[1,x-1]\times[1,N]$ to know what $\HT(x,N)$ is. We call $\HT(x,N)$ the \emph{height function} for our IRF model.]\footnote{This height function is different from the function $x\mapsto s_x$ of Section \ref{ss:exclusion}. Their relationship is explained around \eqref{eq:h-and-s} below.}

Let us assume for a moment that our IRF plaquette weights \eqref{eq:stoch-weights} are actually nonnegative. Then computing the expected value of any observable that is expressed through values of the height function at finitely many locations is a finite procedure. Indeed, all one needs to do is find the outcomes of finitely many Bernoulli trials with biases given by plaquettes weights to construct the portions of the IRF paths that have a chance of passing through or below the observation locations. The resulting expression is a finite linear combination of plaquette weights, and it can certainly be analytically continued beyond the domain of their nonnegativity. In fact, the nonnegativity plays no role in such a computation, we used it just to draw the analogy with the familiar probabilistic Bernoulli trials. In what follows we will use the term `expectation' and the symbol $\E$ to denote the result of such computations, irrespectively of the nonnegativity of the weights. 

Let us introduce another observable at $(x,N)\in\Z_{\ge 1}^2$ that is closely related to the height function and is defined by 
\begin{equation}\label{eq:o-obs}
\o(x,N)=\exp\bigl({2\pi \i (\la-2\eta \HT(x,N))}\bigr)+\exp\bigl({4\pi\i\eta(\HT(x,N)-N+\La_{[1,x)})}\bigr).
\end{equation}
A known $\o(x,N)$ yields two possible values of $\exp(4\pi \i\eta\HT(x,N))$ via solving a quadratic equation. In most cases this allows to reconstruct $\HT(x,N)$ uniquely as one of the roots would not satisfy certain inequalities imposed by the model. 

Let us also rewrite $\o(x,N)$ in the higher spin six vertex model parameters, cf Section \ref{sc:degen}. Using \eqref{eq:parameters} and $\al=-\exp(-2\pi\i\la)$, we have
\begin{equation}\label{eq:o-obs-6v}
\o(x,N)=-\al^{-1} q^{\HT(x,N)}+(s_1s_2\cdots s_{x-1})^{2} q^{N-\HT(x,N)}.
\end{equation}

Observe that the following linear in $\o(x,N)$ expressions split into linear in $q^{\HT(x,N)}$ factors: 
\begin{equation}\label{eq:obs-factor}
q^{N-\La_{[1,x)}}-\al^{-1}q^{2k} -q^{k}\cdot\o(x,N)=q^{N-\La_{[1,x)}}\left(1-q^{k-\HT(x,N)}\right)\left(1+\al^{-1} q^{k+\HT(x,N)-N+\La_{[1,x)}}\right).
\end{equation}
The product of such expressions over $k=0,1,\dots,n-1$ is a degree $n$ polynomial in $\o(x,N)$ on one hand, and on the other hand it is a well-factorized expression 
$$
q^{n(N-\La_{[1,x)})}\left(q^{-\HT(x,N)};q\right)_n \left(-\al^{-1}q^{\HT(x,N)-N+\La_{[1,x)}};q\right)_n
$$
with the standard $q$-Pochhammer symbol notation $(x;q)_m=(1-x)(1-qx)\cdots (1-q^{m-1}x)$. 

Here is the main result of the present section.

\begin{theorem}\label{th:obs} Consider the (pre-)stochastic IRF model in the quadrant of Definition \ref{df:quadrant-IRF} with the corner box $[0,1]\times [0,1]$ filled by $\la$. Then, for any $n,N\ge 1$ and $x_1\ge x_2\ge\dots\ge x_n\ge 1$,
with the notation $q=\exp(-4\pi \i \eta)$, the expression
\begin{equation}\label{eq:expect}
E_N(x_1,\dots,x_n):=\frac{1}{(e^{2\pi\i\la};q)_n}\E \left[\prod_{k=0}^{n-1} 
\left(q^{N-\La_{[1,x_{k+1})}}+e^{2\pi\i\la}q^{2k} -q^{k}\cdot\o(x_{k+1},N)\right)\right]
\end{equation}
is independent of the dynamic parameter $\la$. In particular, taking $\la\to -\i\infty$ shows that, with the identification of parameters \eqref{eq:parameters},
\begin{equation}\label{eq:expect-hs6v}
E_N(x_1,\dots,x_n)=\E_{\mathrm{h.s.6\,v.m.}} \left[\prod_{k=0}^{n-1} \left(q^{\HT(x_{k+1},N)}-q^{k}\right)\right],
\end{equation}
where the expectation in the right-hand side is with respect to the stochastic higher spin six vertex model of Section \ref{ss:hs6v} above.
\end{theorem} 
\begin{remark}\label{rm:6v-expect} The expectation in the right-hand side of \eqref{eq:expect-hs6v} was computed in \cite[Lemma 9.11]{BP}; that result implies the following integral representation for $E_n$:
$$
E_N(x_1,\dots,x_n)=\frac{q^{\frac{n(n-1)}2}}{(2\pi\i)^n}
		\oint
		\dots
		\oint
		\prod_{1\le 1\le i<j\le n}\frac{y_i-y_j}{y_i-qy_j}
		\prod_{i=1}^{n}\bigg(
		\prod_{j=1}^{x_i-1}
		\frac{\xi_j-s_jy_i}{\xi_j-s_j^{-1}y_i}
		\prod_{k=1}^{N}\frac{1-qu_ky_i}{1-u_ky_i}\bigg) \frac{d y_i}{y_i}\,,
$$
where the integration contours are positively oriented loops around $\{u_k^{-1}\}_{k=1}^N$ (one can think of this integral as of the sum of the residues of the integrand taken at all possible poles of the form $w_i=u_k$). In the process of proving Theorem \ref{th:obs} we will uncover an equivalent form of this integral in the IRF notations:
\begin{multline}\label{eq:expect-int}
E_N(x_1,\dots,x_n)=\exp\left(-2\pi\i\eta\left(\frac{n(n-1)}2+nN-\sum_i\La_{[1,x_i)}\right)\right)\\ \times\oint\dots\oint \prod_{1\le i<j\le n} \frac{f(v_i-v_j)}{f(v_i-v_j+2\eta)}\prod_{i=1}^n\bigg(\prod_{j=1}^{x_i-1}\frac{f(v_i-p_j)}{f(v_i-q_j)}\prod_{k=1}^N \frac{f(v_i-w_k-2\eta)}{f(v_i-w_k)}\bigg)\,{dv_i},
\end{multline}
where the integration contours are positively oriented loops around $\{w_k\}_{k=1}^N$ (again, one can treat the integral as the sum of residues at $v_i=w_k$). 
\end{remark}

\begin{remark}\label{rm:unique} If all the plaquette weights are nonnegative, and the parameters $\al,q,\{s_j^2\}_{j\ge 1}$ are real, the averages \eqref{eq:expect} uniquely determine the joint distribution of the observables $\o(x_1,N)$, $\dots$, $\o(x_n,N)$ for any $N,x_1,\dots,x_n\ge 1$. Indeed, these are bounded real-valued random variables (cf.~\eqref{eq:o-obs-6v}), and all their joint moments can be extracted by taking linear combinations of various instances of \eqref{eq:expect}.
\end{remark}

Let us also record separately the special case of Theorem \ref{th:obs} with all $x_j$'s being equal.

\begin{corollary}\label{cr:moments} In the notations of Theorem \ref{th:obs}, for any $n,x,N\ge 1$ we have
\begin{equation}\label{eq:moment}
\frac{q^{n(N-\La_{[1,x)})}}{(-\al^{-1};q)_n}\E \left[\left(q^{-\HT(x,N)};q\right)_n \left(-\al^{-1}q^{\HT(x,N)-N+\La_{[1,x)}};q\right)_n\right]=\E_{\mathrm{h.s.6\,v.m.}} \left[\prod_{k=0}^{n-1} \left(q^{\HT(x,N)}-q^{k}\right)\right],
\end{equation}
which can be further evaluated via the integrals of Remark \ref{rm:6v-expect} with all $x_j$'s equal to $x$. 
\end{corollary}

\begin{proof}[Proof of Theorem \ref{th:obs}] As was discussed in the beginning of this section, the desired claim is a statement about a finite linear combination of plaquette weights, which is a rational function in the parameters $q^{1/2}=e^{-2\pi\i\eta}, \al=-e^{-2\pi\i\la}, \{s_i,\xi_i\}_{1\le i<x_1}, \{u_j\}_{1\le j\le N}$, where we use the notation \eqref{eq:parameters}. Hence, it is sufficient to prove for these parameters varying over any open set. 

Let us assume that the parameters satisfy the admissibility assumption of Definition \ref{df:adm}, i.e., we assume that $\eta$ is small enough, $p_i=z_i+(1-\La_i)\eta$ are sufficiently close together, $q_j=z_j+(1+\La_j)\eta$ are sufficiently close together, and these two groups of points are far apart (in the same fundamental strip of $\C/\Z$). 

The proof will proceed as follows: We will start with the integral in the right-hand side of \eqref{eq:expect-int}, divide it by $(2\pi i)^n$, and show that the result is equal to 
\begin{multline}\label{eq:expect-sum}
\frac{1}{\pi^n \prod_{i=0}^{n-1}f(\la-2\eta i)}
\sum_{\nu=(\nu_1\ge\dots\ge\nu_N\ge 1)} \Bstoch_\nu\bigl(\la-2\eta(N-\La_0);w_1,\dots,w_N\bigr) \\ \times \prod_{k=0}^{n-1}f\Bigl(2\eta\bigl (\HT_\nu(x_{k+1})-k\bigr)\Bigr)f\Bigl(-\la+2\eta\bigl(\HT_\nu(x_{k+1})+k-N+\La_{[1,x_{k+1})}\bigr)\Bigr)\\
=\frac{1}{\pi^n \prod_{i=0}^{n-1}f(\la-2\eta i)} \E\left[
 \prod_{k=0}^{n-1}f\Bigl(2\eta\bigl (\HT(x_{k+1},N)-k\bigr)\Bigr)f\Bigl(-\la+2\eta\bigl(\HT(x_{k+1},N)+k-N+\La_{[1,x_{k+1})}\bigr)\Bigr)\right],
\end{multline}
where the equality in \eqref{eq:expect-sum} is due to \eqref{eq:prob=B-stoch} (see also Theorem \ref{th:B-stoch}), and we used the notation 
\begin{equation}\label{eq:height-nu}
\HT_\nu(x)=\#\{i\ge 1: \nu_i\ge x\}.
\end{equation}
Since the initial integral from \eqref{eq:expect-int} is independent of $\la$, this would imply the claim of the theorem, modulo elementary manipulations with $f(z)= \sin\pi z$. 

If we want to see the integral from \eqref{eq:expect-int} divided by $(2\pi\i)^n$ as a linear combination of $\Bstoch_\nu(\la-2\eta(N-\La_0);w_1,\dots,w_N)$, we can alternatively seek the decomposition of
\begin{multline}\label{eq:to-decompose}
\frac{\pi^N c_{\nu}(\la-2\eta(N-\La_0))}{\prod_{i=0}^{N-1}{f(\la-2\eta i)}}
\prod_{k=1}^N \frac{f(w_k-p_0)}{f(w_k-q_0)}
\\\times \oint\dots\oint \prod_{1\le i<j\le n} \frac{f(v_i-v_j)}{f(v_i-v_j+2\eta)}\prod_{i=1}^n\bigg(\prod_{j=1}^{x_i-1}\frac{f(v_i-p_j)}{f(v_i-q_j)}\prod_{k=1}^N \frac{f(v_i-w_k-2\eta)}{f(v_i-w_k)}\bigg)\,\frac{dv_i}{2\pi\i}
\end{multline}
as a linear combination of $B_\nu(\la-2\eta(N-\La_0);w_1,\dots,w_N)$; by virtue of \eqref{eq:B-stoch} and \eqref{eq:D-rho-expl}, the resulting coefficients will be the same (given that only $\nu$'s with $\nu_N\ge 1$ enter the latter decomposition).

Let us first argue that such a decomposition exists. The dependence of the last expression on $\{w_k\}_{k=1}^N$ is through the product
$$
\prod_{k=1}^N \frac{f(w_k-p_0)} {f(w_k-q_0)}\prod_{i=1}^n \frac{f(v_i-w_k-2\eta)}{f(v_i-w_k)}\,.
$$
We proceed similarly to the proof of Proposition \ref{pr:cauchy-rho}, see also the beginning of Section \ref{sc:stoch-IRF} around \eqref{eq:simplification}, \eqref{eq:n-dep}. More exactly, we start with the Cauchy identity \eqref{eq:cauchy}, use for $\{v_j\}$ a finite sequence of indeterminates $v_1,\dots,v_n$ and a string of the form $(v,v-2\eta,\dots,v-2\eta(m-1))$, choose $v=p_0+2\eta m$, and send $m\to\infty$ assuming $\Im\eta\ne 0$.
The Cauchy summations remain uniformly convergent throughout the limit transition as long as $w_i$'s are close enough to $p_i$'s and $v_1,\dots,v_n$ do not approach those, and the limiting summation is what we want. Further, using the same uniform convergence, we can integrate the resulting summation over $v_1,\dots,v_n$ as in \eqref{eq:to-decompose}, with the conclusion that \eqref{eq:to-decompose} has an expansion in $B_\nu(\la-2\eta(N-\La_0);w_1,\dots,w_N)$ as long as the $v$-contours include both $\{p_i\}$ and $\{w_j\}$ that are close enough together. 
 
Our next goal is to actually compute the coefficient of $B_\nu(\la-2\eta(N-\La_0);w_1,\dots,w_N)$ in \eqref{eq:to-decompose} using the orthogonality relations proved in Theorem \ref{th:orth}. Relation \eqref{eq:orth} implies that this coefficient equals
\begin{multline}\label{eq:coefficient}
\frac{\pi^N}{\prod_{i=0}^{N-1}{f(\la-2\eta i)}}\oint_{\ga_1}\frac{dw_1}{2\pi\i}\cdots\oint_{\ga_N}\frac{dw_N}{2\pi\i}\oint\frac{dv_1}{2\pi\i}\cdots \oint\frac{dv_n}{2\pi\i} \prod_{1\le i<j\le N}\frac{f(w_i-w_j)}{f(w_i-w_j-2\eta)}\\ \times
\prod_{k=1}^N \Biggl(\frac{f(w_k-p_0)}{f(w_k-q_0)}\frac{1}{f(w_k-p_{\nu_k})}\prod_{j=1}^{\nu_k-1}\frac{f(w_k-q_j)}{f(w_k-p_j)}\cdot
f\bigl(\la-w_k+p_{\nu_k}+2\eta(N+1-2k-\La_{[1,\nu_k)})\bigr)\Biggr)
\\ \times\prod_{\substack{1\le k\le N\\ 1\le i\le n}} \frac{f(v_i-w_k-2\eta)}{f(v_i-w_k)} \prod_{1\le i<j\le n} \frac{f(v_i-v_j)}{f(v_i-v_j+2\eta)}\prod_{i=1}^n \prod_{j=1}^{x_i-1}\frac{f(v_i-p_j)}{f(v_i-q_j)}\,.
\end{multline}

We will compute this integral by evaluating residues, first shrinking the $v_1,\ldots,v_n$-contours and expanding the $w$-contours after that; the expansion is similar to the evaluation of $\Dnorm_\nu(\la;\rho)$ in Proposition \ref{pr:D-rho-expl} above. 

We have so far imposed no restrictions on positions of the $v$-contours with respect to each other, they are only required to surround the $w$-contours and leave the potential poles $\{q_j\}$ outside. In fact, the integral does not depend on their relative positions, one can show that taking any residue of the form $v_i=v_j+2\eta$ for $1\le i<j\le n$ leads to \eqref{eq:coefficient} vanishing. However, we do not really need this fact, and will simply assume that $v_1$-contour is well (within the $2\eta$-shift) inside the $v_2$-contour, which is well inside $v_3$-contour, etc. This way we will face no obstacles from the denominators $f(v_i-v_j+2\eta)$ when shrinking the $v$-contours and picking the residues at $v_i=w_k$. 

Let us now shrink the $v_1$-contour. Taking the residue at $v_1=w_{t_1}$, $1\le t_1\le N$, gives the $v_1$-dependent terms
$$
\frac{-f(2\eta)}{f'(0)}\prod_{k\ne t_1}\frac{f(w_{t_1}-w_k-2\eta)}{f(w_{t_1}-w_k)} \prod_{j=2}^n
\frac{f(w_{t_1}-v_j)}{f(w_{t_1}-v_j+2\eta)}\prod_{l=1}^{x_1-1} \frac{f(w_{t_1}-p_l)}{f(w_{t_1}-q_l)}\,.
$$
The middle product removes $w_{t_1}$ from $\prod_{i,k}{f(v_i-w_k-2\eta)}/{f(v_i-w_k)}$, which means that 
when we shrink the next contour (for $v_2$), we have to exclude $w_{t_1}$ from the list of possible poles. Further, the first product removes the poles at $w_{t_1}=w_k+2\eta$ for $k>t_1$ in the integrand of \eqref{eq:coefficient}.
If now $x_1>\nu_{t_1}$ then the third product removes the remaining poles at $p_1,\dots,p_{\nu_{t_1}}$ from the inside of $\ga_{t_1}$, and shrinking that contour yields zero. Hence, we can assume that $x_1\le \nu_{t_1}$. 

Let us now shrink the next, $v_2$-contour. Taking the residue at $v_2=w_{t_2}$ with $t_2\ne t_1$ gives the $v_2$-dependent terms
$$
\frac{-f(2\eta)}{f'(0)}\prod_{k\notin \{t_1,t_2\}}\frac{f(w_{t_2}-w_k-2\eta)}{f(w_{t_2}-w_k)} \prod_{j=3}^n
\frac{f(w_{t_2}-v_j)}{f(w_{t_2}-v_j+2\eta)}\prod_{l=1}^{x_2-1} \frac{f(w_{t_2}-p_l)}{f(w_{t_2}-q_l)}\,.
$$
Again we see the removal of $w_{t_2}$ from $\prod_{i,k}{f(v_i-w_k-2\eta)}/{f(v_i-w_k)}$, and we can again assume that $x_2\le \nu_{t_2}$. Indeed, otherwise if $t_2>t_1$ then the integral vanishes due to absence of singularities inside $\ga_{t_2}$, and if $t_2<t_1$ then $\nu_{t_2}\ge \nu_{t_1}\ge x_1\ge x_2$.  

Continuing in this fashion for all the $v$-contours, we conclude that for a nonzero result we have to choose pairwise distinct integers $t_1,\dots,t_n$ between $1$ and $N$ such that $\nu_{t_j}\ge x_j$ for all $1\le j\le n$, and the total contribution of the $v$-dependent terms in the integrand of \eqref{eq:coefficient} is  
\begin{equation}\label{eq:additional-factor}
\left(\frac{-f(2\eta)}{f'(0)}\right)^n\prod_{1\le i<j\le n}\frac{f(w_{t_i}-w_{t_j}-2\eta)}{f(w_{t_i}-w_{t_j})} 
\prod_{\substack{i\in\{t_1,\dots,t_n\}\\ j\notin\{t_1,\dots,t_n\}}}
\frac{f(w_{i}-w_j-2\eta)}{f(w_{i}-w_j)}\prod_{\substack{1\le i\le n\\ 1\le l\le x_i-1}} \frac{f(w_{t_i}-p_l)}{f(w_{t_i}-q_l)}\,.
\end{equation}

From this point our arguments are very reminiscent of the proof of Proposition \ref{pr:D-rho-expl}.
First, if $\nu_N=0$ then $\gamma_N$ (which is the integration contour for $u_N$) has no singularities inside and the integral vanishes. Hence, we will only have the signatures with $\nu_N\ge 1$, which was actually required to pass from $B_{\nu}$'s to $\Bstoch_\nu$'s above. 

Let us expand $\ga_1,\dots,\ga_N$ (in that order), as in the proof of Proposition \ref{pr:D-rho-expl}. There are no poles that we encounter along the way, so we just need to collect the sines arising from the asymptotics at $\pm \i\infty$. It is easier to compare with \eqref{eq:sine}; we have to take into account two discrepancies. First, $\la$ in \eqref{eq:sine} needs to be replaced by the shifted value of $\la-2\eta(N-\La_0)$. Second, we have the additional factor \eqref{eq:additional-factor} in the integrand. 

Denote $I=\{t_1,\dots,t_n\}$, $I^c=\{1,\dots,N\}\setminus \{t_1,\dots,t_n\}$.

For any $j\in I^c$, the shift of the argument in the sine of the type \eqref{eq:sine} coming from \eqref{eq:additional-factor} is $-2\eta n+2\eta \cdot\#\{i\in I\mid i<j\}$, where the second term is due to the fact that $\ga_j$ is expanded after all $\ga_i$ with $i<j$. Thus, the expansion of the contour that corresponds to the minimal element of $I^c$ gives $f(\la-2\eta n)/\pi$, the one for the second minimal element of $I^c$ gives $f(\la-2\eta (n+1))$, etc. In total, the expansion of $\ga_j$'s with $j\in I^c$ brings the factor of 
$$
\frac{f(\la-2\eta n)f(\la-2\eta(n+1))\,\cdots \,f(\la-2\eta(N-1))}{\pi^{N-n}}\,.
$$

Let us now look at the contribution of the expansion of $\ga_{t_i}$ for a $t_i\in I$. The extra shift of the sine argument due to \eqref{eq:additional-factor} is 
\begin{multline}
2\eta\, \bigl(-\#\{t_j\in I\mid t_j>t_i,j<i\}+\#\{t_j\in I\mid t_j>t_i,j>i\}+\#\{j\in I^c\mid j>t_i\}\bigr)+\sum_{l=1}^{x_i-1}(p_l-q_l)\\
=2\eta\,\bigl((N-t_i)-2\,\#\{t_j\in I\mid t_j>t_i,j<i\}-\La_{[1,x_i)}\bigr).
\end{multline}
Hence, the corresponding factor is $\pi^{-1}f\bigl(\la +2\eta\,((N-2t_i+1)-2\,\#\{t_j\in I\mid t_j>t_i,j<i\}-\La_{[1,x_i)})\bigr)$.
 
Gathering all the contributions, we conclude that (using $f'(0)=\pi$)
\begin{equation}\label{eq:coeff=sum}
(\ref{eq:coefficient}) =\frac{(-f(2\eta))^n}{\prod_{i=0}^{n-1}f(\la-2\eta i)}\sum_{\substack{1\le t_1\le T_1,\dots,1\le t_n\le T_n\\
t_i\ne t_j\ \text{for}\ i\ne j}} X_{t_1}^{(1)}X_{t_2+\inv_{\le 2}}^{(2)}\cdots X_{t_n+\inv_{\le n}}^{(n)},
\end{equation}
where for any $1\le j\le n$, $T_j$ is the largest positive integer such that $\nu_{T_j}\ge x_j$ (the sum is set to 0 if $\nu_1<x_1$, otherwise all $T_j$'s are well-defined with $T_1\le\dots\le T_n$, because $x_1\ge \dots\ge x_n$), and for all $i=1,\dots,n$,
$$
X_{t}^{(i)}=\frac{f\bigl(\la +2\eta\,((N-2t+1)-\La_{[1,x_i)})\bigr)}{\pi},\qquad \inv_{\le i}=\#\{t_j\in I\mid t_j>t_i,\, j<i\}.
$$
Note that $T_j$ is actually equal to $\HT_\nu(x_j)$, cf. \eqref{eq:height-nu}. 

The last step of the computation is the following lemma. 

\begin{lemma}\label{lm:sum}
Fix $n\ge 1$ and let $\bigl\{Y_i^{(j)}\bigr\}_{i\ge 1, 1\le j\le n}$ be indeterminates. For any $T_1,\ldots,T_n\in\Z_{\ge1}$ we have
	\begin{equation*}
		\sum_{\substack{1\le t_1 \le T_1,\ldots,1\le t_n\le T_n\\
		t_i\ne t_j\ \text{for}\ i\ne j}}
		Y_{t_1}^{(1)}Y_{t_2+\inv_{\le 2}}^{(2)}\cdots Y_{t_n+\inv_{\le n}}^{(n)}
		=\prod_{j=1}^n \bigl(Y_j^{(j)}+\ldots+Y_{T_j}^{(j)}\bigr),
	\end{equation*}
	where $\inv_{\le i}=\#\{t_j\in I\mid t_j>t_i,\, j<i\}$.
	Here the right-hand side is assumed to be 0 if one of the sums is empty.
\end{lemma}
\begin{proof} The argument is identical to the proof of \cite[Lemma 9.14]{BP}, which is a special case of the above statement when $Y_i^{(j)}$'s do not depend on the upper index $j$. 
\end{proof}

Applying Lemma \ref{lm:sum} to the right-hand side of \eqref{eq:coeff=sum} and using the readily checkable identity
$$
f(a)+f(a+2b)+\dots+f(a+2(m-1)b)=\frac{f(mb)}{f(b)}\,f(a+(m-1)b),\qquad a,b\in\C,\quad m=1,2,\dots, 
$$
we obtain 
\begin{multline*}
(\ref{eq:coefficient})=\frac{(-f(2\eta))^n}{\prod_{i=0}^{n-1}f(\la-2\eta i)}
\prod_{j=1}^n \frac{f(2\eta(T_j-j+1))}{f(2\eta)}\, \frac{f\bigl(\la +2\eta\,((N-T_j-j+1)-\La_{[1,x_j)})\bigr)}{\pi}\\
=\prod_{i=0}^{n-1} 
\frac{(-1)f(2\eta(\HT_{\nu}(x_{i+1})-i))f\bigl(\la +2\eta\,((N-\HT_\nu(x_{i+1})-i)-\La_{[1,x_{i+1})})\bigr)}
{\pi \cdot f(\la-2\eta i)}
\,.
\end{multline*}

As the last expression coincides with the coefficient of $\Bstoch_\nu(\la-2\eta(N-\La_0);w_1,\dots,w_N)$ in \eqref{eq:expect-sum}, the proof of Theorem \ref{th:obs} is complete.  
\end{proof}

Let us now investigate what Theorem \ref{th:obs} means for the limiting cases of our stochastic IRF model that were discussed in Section \ref{sc:degen}. 

The degeneration to the dynamic stochastic six vertex model is very simple --- one just needs to replace all $\La_j$ by $1$, hence $\La_{[1,x)}\equiv x-1$ and $s_j^2\equiv q$; everything else remains the same. Degenerating further to the dynamic ASEP of Definition \ref{df:dynamic-ASEP} requires a little bit of work to match the notations. 

As was mentioned in Section \ref{ss:exclusion}, for a sequence $\{s_x\}_{x\in\Z}$ with $s_{x+1}-s_x\in \{-1,1\}$ one can associate a particle configuration in $\Z+\frac 12$ by placing a particle at $x+\frac 12$ if an only if $s_{x+1}-s_x=-1$. The $\al\to 0$ limit of the dynamic ASEP is then the usual ASEP on $\Z+\frac 12$ with particles jumping left with rate $q$ and jumping right with rate 1. For ASEP configurations with finitely many particles to the right of the origin, define the height function $\HT_{ASEP}:\Z\times \R_{\ge 0}\to\Z_{\ge 0}$ as
$$
\HT_{ASEP}(x,t)=\text{the number of ASEP particles at time $t$ to the right of } x.
$$ 
We will sometimes drop `$t$' from this notation when time is not relevant. 

For the sequences $\{s_x\}_{x\in\Z}$ with $s_x\equiv x$ for $x\gg 1$, one has
\begin{equation}\label{eq:h-and-s}
\HT_{ASEP}(x)=\frac{s_x-x}{2}\qquad \text{or}\qquad s_x=2\HT_{ASEP}(x)+x. 
\end{equation}
Indeed, the two sides are obviously equal for $x\gg 1$, and it is also clear that the differences of the two sides at $x+1$ and $x$ match. 

It is not hard to see that in the limiting procedure of Proposition \ref{pr:6v-ASEP}, the height function of the six vertex model (which counts the number of paths through or below a given vertex) converges to that of the limiting ASEP:
$$
\lim_{\epsilon\to 0} \HT([t\epsilon^{-1}]+x+1,[t\epsilon^{-1}])= \HT_{ASEP}(x,t).
$$
This leads, with the help of Proposition \ref{pr:6v-ASEP}, to the following version of Theorem \ref{th:obs} and Remark \ref{rm:6v-expect}.
\begin{corollary}\label{cr:obs-ASEP} Consider the dynamic ASEP of Definition \ref{df:dynamic-ASEP} with the initial condition $s_x(0)\equiv |x|$, and introduce its observables
\begin{equation}\label{eq:obs-ASEP}
\o_{ASEP}(x,t)=-\alpha^{-1} q^{\HT_{ASEP}(x,t)}+q^{-x-\HT_{ASEP}(x,t)}=-\alpha^{-1} q^{\frac{s_x(t)-x}2}+q^{\frac{-s_x(t)-x}2}.
\end{equation}
Then, for any $t\ge 0$, $n\ge 1$ and $x_1\ge x_2\ge\dots\ge x_n$, the expression
\begin{equation}\label{eq:expect-ASEP}
E^{ASEP}_t(x_1,\dots,x_n):=\frac{1}{(-\al^{-1};q)_n}\E_{\mathrm{dynamic\, ASEP\,at\,time\,}t} \left[\prod_{k=0}^{n-1} 
\left(q^{-x_{k+1}}-\al^{-1} q^{2k} -q^{k}\cdot\o_{ASEP}(x_{k+1},t)\right)\right]
\end{equation}
is independent of the parameter $\al$. In particular, taking $\al\to 0$ shows that
\begin{equation}\label{eq:expect-ASEP-usual}
E_t^{ASEP}(x_1,\dots,x_n)=\E_{\mathrm{usual\,ASEP\,at\,time\,}t} \left[\prod_{k=0}^{n-1} \left(q^{\HT_{ASEP}(x_{k+1},t)}-q^{k}\right)\right],
\end{equation}
where the expectation in the right-hand side is with respect to the usual ASEP on $\Z+\frac 12$ with left jump rate $q$, right jump rate $1$, and particles at $t=0$ occupying all the negative locations. The latter expectation is given explicitly by 
\begin{multline}\label{eq:expect-ASEP-integral}
		E_t^{ASEP}(x_1,\dots,x_n)
		=
		\frac{q^{\frac{n(n-1)}2}}{(2\pi\i)^n}
		\oint
		\ldots
		\oint
		\prod_{1\le i<j\le n}\frac{y_i-y_j}{y_i-qy_j}
		\\\times
		\prod_{i=1}^{n}\bigg(
		\bigg(\frac{1-y_i}{1-q y_i}\bigg)^{x_i}
		\exp\bigg\{\frac{(1-q)^{2}y_i}{(1-y_i)(1-q y_i)}\,t\bigg\}
		\bigg)\frac{d y_i}{y_i},
\end{multline}
where the integration contours are small positively oriented loops around $1$.
\end{corollary}
For $x_1=\dots=x_n$ the integral representation \eqref{eq:expect-ASEP} for the (usual) ASEP goes back to \cite[Theorem 4.20]{BCS14}, and for different $x_i$'s it follows from \cite[Corollary 10.2]{BP}. 

The next degeneration we consider is the case of the rational (pre-)stochastic IRF model from Section \ref{ss:rat-IRF}. 
This involves a simple limit transition in \eqref{eq:expect} and leads to the following statement. 
\begin{corollary}\label{cr:expect-rat} Consider the rational (pre-)stochastic IRF in the quadrant from Section \ref{ss:rat-IRF}, and introduce its observables
\begin{equation}\label{eq:obs-rat-IRF}
\o_{rational}(x,N)=\HT(x,N)(\HT(x,N)-\la-N+x-1).
\end{equation}
Then, for any $n,N\ge 1$ and $x_1\ge x_2\ge\dots\ge x_n\ge 1$, the expression
\begin{equation}\label{eq:expect-rat-IRF}
E^{rational}_N(x_1,\dots,x_n):=\frac{1}{(-\la)_n}\E_{\mathrm{rational\, stoch.\, IRF}} \left[\prod_{k=0}^{n-1} 
\left(k^2-k(\la+N-x_{k+1}+1)-\o_{rational}(x_{k+1},N)\right)\right]
\end{equation}
is independent of the parameter $\la$. Its explicit form is as follows:
\begin{equation}\label{eq:expect-int-rat}
E^{rational}_N(x_1,\dots,x_n)=\oint\dots\oint \prod_{1\le i<j\le n} \frac{v_i-v_j}{v_i-v_j+1}\prod_{i=1}^n\bigg(\prod_{j=1}^{x_i-1}\frac{v_i-z_j}{v_i-z_j-1}\prod_{k=1}^N \frac{v_i-w_k-1}{v_i-w_k}\bigg)\,\frac{dv_i}{2\pi\i}\,,
\end{equation}
where the integration contours are small positively oriented loops around $\{w_k\}_{k=1}^N$. 
\end{corollary}
\begin{proof} We need to replace $(2\eta,\la,z_i,w_j)$ by $(\epsilon, \epsilon\la,\epsilon z_i, \epsilon w_j)$ and look at the $\epsilon\to 0$ asymptotics of the right-hand side of \eqref{eq:expect}. Using \eqref{eq:obs-factor} (and remembering that $\La_{[1,x)}=x-1$) gives
$$
\frac{q^{N-\La_{[1,x)}}+e^{2\pi\i\la}q^{2k} -q^{k}\cdot\o(x,N)}{1-e^{2\pi\i\la}q^k}\sim 2\pi\i\epsilon\cdot \frac{(k-\HT(x,N))(k+\HT(x,N)-N+x-1-\la)}{k-\la}\,.
$$
Taking the product over $k=0,1,\dots,n-1$ and substituting $x_{k+1}$ for $x$ leads to the right-hand side of \eqref{eq:expect-rat-IRF} times a factor of $(2\pi\i\epsilon)^n$. Comparing with the similar asymptotics of \eqref{eq:expect-int} gives the result. 
\end{proof}

The final limit transition that we consider is from the rational stochastic IRF model in the quadrant to the dynamic SSEP of Definition \ref{df:dynamic-SSEP}. It is very similar to the case of the dynamic ASEP from Corollary \ref{cr:obs-ASEP}, and we just give the result. The SSEP height function $\HT_{SSEP}$ that we use below is defined in exactly the same way as $\HT_{ASEP}$ above. 

\begin{corollary}\label{cr:obs-SSEP} Consider the dynamic SSEP of Definition \ref{df:dynamic-SSEP} with the initial condition $s_x(0)\equiv |x|$, and introduce its observables
\begin{equation}\label{eq:obs-SSEP}
\o_{SSEP}(x,t)=\HT_{SSEP}(x,t)(\HT_{SSEP}(x,t)+x-\la)={\frac{s_x(t)-x}2}\left(\frac{s_x(t)+x}2-\la\right).
\end{equation}
Then, for any $t\ge 0$, $n\ge 1$ and $x_1\ge x_2\ge\dots\ge x_n$, the expression
\begin{equation}\label{eq:expect-SSEP}
E^{SSEP}_t(x_1,\dots,x_n):=\frac{1}{(-\la)_n}\E_{\mathrm{dynamic\, SSEP\,at\,time\,}t} \left[\prod_{k=0}^{n-1} 
\left(k^2-k(\la-x_{k+1})-\o_{SSEP}(x_{k+1},t)\right)\right]
\end{equation}
is independent of the parameter $\la$. In particular, taking $\la\to \infty$ shows that
\begin{equation}\label{eq:expect-SSEP-usual}
E_t^{SSEP}(x_1,\dots,x_n)=\E_{\mathrm{usual\,SSEP\,at\,time\,}t} \left[\prod_{k=0}^{n-1} \left(k-{\HT_{SSEP}(x_{k+1},t)}\right)\right],
\end{equation}
where the expectation in the right-hand side is with respect to the usual SSEP on $\Z+\frac 12$ with both left and right jump rate equal to $1$, and particles at $t=0$ occupying all the negative locations. The latter expectation is given explicitly by 
\begin{equation}\label{eq:expect-SSEP-int}
		E_t^{SSEP}(x_1,\dots,x_n)
		=
		\oint
		\ldots
		\oint
		\prod_{1\le i<j\le n}\frac{v_i-v_j}{v_i-v_j+1}
		\prod_{i=1}^{n}\bigg(
		\bigg(\frac{v_i}{v_i-1}\bigg)^{x_i}
		\exp\bigg\{\frac{t}{v_i(v_i-1)}\bigg\}
		\bigg)\frac{d v_i}{2\pi\i}\,,
\end{equation}
where the integration contours are small positively oriented loops around $0$.
\end{corollary}

\section{One-point asymptotics of the dynamic SSEP}\label{sc:SSEP} In this section we consider the dynamic symmetric simple exclusion process described in Definition \ref{df:dynamic-SSEP} above. In the particle interpretation, this is a system of particles in $\Z+\frac 12$, no more than one particle per site, that jump left and right by one unit independently with exponential  
waiting times of certain rates.\footnote{A particle can never jump into a spot occupied by another particle though; thus `exclusion' in the name.} The rates vary in time and also between particles, and they depend on the height function $\HT:\Z_{\ge 0}\times\R_{\ge 0}\to \Z_{\ge 0}$ defined by 
$$
\HT(x,t)=\HT_{SSEP}(x,t)=\text{number of particles to the right of $x$ at time $t$}. 
$$
We only consider configurations with $\HT(x,t)\equiv 0$ for $x\gg 1$, equivalently, the number of particles to the right of the origin is always finite; thus, the height function always takes finite values. 

The rate of a left jump of a particle from $(x+\frac 12)$ to $(x-\frac 12)$ is defined as $(s_x+\bar{\la})/(s_x+\bar{\la}-1)$, where $s_x=2\HT(x)+x$ (we are muting the dependence on $t$ here), and $\bar{\la}$ is a parameter of the process\footnote{We use $\bar{\la}$ as replacement for $(-\la)$ of Section \ref{ss:exclusion} as it is more convenient to deal with a positive parameter.}, while the rate of a right jump from $(x-\frac 12)$ to $(x+\frac 12)$ is $(s_x+\bar{\la})/(s_x+\bar{\la}+1)$. The limit $\bar{\la}\to\infty$ leads to the usual SSEP with both left and right rates identically equal to 1. 

We will only consider the \emph{step} or \emph{packed} initial condition with particles occupying all negative locations at $t=0$; equivalently, $\HT(x,0)=\mathbf{1}_{x<0}\cdot |x|$ or $s_x(t=0)=|x|$. The parameter $\bar{\la}$ is assumed to be positive; in that case all the jump rates are clearly positive as well (as we will always have $\HT(x,t)\ge \HT(x,0)$ or $s_x(t)\ge s_x(0)$). 

Our goal is to study the first order behavior of $\HT(x,t)$ as time gets large. Let us first give a corresponding result for the usual SSEP that is well known. The following functions on $\R\times \R_{>0}$ will be useful:
\begin{equation}\label{eq:SSEP-limit-shape}
H(\chi,\tau)=\sqrt{\frac{\tau}{\pi}}\,\exp\left(-\frac{\chi^2}{4\tau}\right)-\frac{\chi}{2}\,\erfc\left(\frac{\chi}{2\sqrt{\tau}}\right),\qquad \chi\in\R,\quad  \tau\in\R_{>0},
\end{equation}
where $\erfc(\,\cdot\,)$ is the complementary error function. 
Observe that $H(\chi,\tau)$ solves the (1+1)d heat equation:
$$
\frac{\partial H(\chi,\tau)}{\partial\tau} =\frac {\partial^2H(\chi,\tau)}{\partial\chi^2}  
$$
with the initial condition $\lim_{\tau\to 0} H(\chi,\tau)= \mathbf{1}_{x<0}\cdot |x|$ (which also happens to be the initial condition $\HT(x,0)$ for our height function).

\begin{proposition}\label{pr:SSEP} Consider the usual SSEP with the step initial condition $\HT(x,0)=\mathbf{1}_{x<0}\cdot |x|$. Then for any $\chi\in\R$, $\tau\in\R_{>0}$ we have the following convergence in moments and in probability:
\begin{equation}\label{eq:SSEP-lim}
\lim_{L\to\infty} L^{-\frac 12}\cdot \HT(L^{\frac 12}\chi, L\tau )= H(\chi,\tau)
\end{equation}
with $H(\chi,\tau)$ given by \eqref{eq:SSEP-limit-shape}. The statement remains valid if $\chi$ in $\HT(L^{\frac 12}\chi, L\tau )$ is $L$-dependent but has a limit: $\chi=\chi(L)$ with $\lim_{L\to\infty}\chi(L)=\chi$. 
\end{proposition}
\begin{proof} This is a standard statement from the hydrodynamic theory of SSEP, cf., e.g., 
\cite[Chapter 4]{KL}, \cite[Chapter 8]{S}, where much more general statements are available.

In our concrete situation we can reach the one-point asymptotics above by taking the asymptotics of moments of $\HT(x,t)$ given by equating the right-hand sides of \eqref{eq:expect-SSEP-usual} and \eqref{eq:expect-SSEP-int}. For example, for the first moment we have 
$$
\E \HT(x,t)=-\oint_{\text{around }0} \left(\frac{v}{v-1}\right)^x \exp\left\{\frac t{v(v-1)}\right\}\frac{dv}{2\pi\i}\,,
$$
which under the change of variable $u=v/(v-1)$ turns into
$$
\E \HT(x,t)=\oint_{\text{around }0} u^x \exp\left\{t(u+u^{-1}-2)\right\}\frac{du}{2\pi\i(u-1)^2}\,.
$$ 
Turning the contour into a large straight piece parallel to the imaginary axis slightly to the left of $\Re u=1$ and a large arc closing the contour in the left half-plane, we can further change the variable as $u=1+L^{-\frac 12}z$ and obtain 
$$
\lim_{L\to\infty}  L^{-\frac 12}\cdot \HT(L^{\frac 12}\chi, L\tau )=\frac{1}{2\pi\i}\int_{\i\R-0}e^{
\chi z+\tau z^2}\,\frac{dz}{z^2}\,,
$$
which is exactly $H(\chi,\tau)$. 

The asymptotics of the right-hand side of \eqref{eq:expect-SSEP-int} for an arbitrary $n\ge 1$ is obtained in exactly the same way by repeating the above contour deformations for each of the  integration variables, with the result being the $n$th power of that for the single integral (the cross-terms from $\prod_{i<j}$ play no role). 
Comparing with \eqref{eq:expect-SSEP-usual} shows that $\lim_{L\to\infty}  L^{-\frac n2}\cdot (\HT(L^{\frac 12}\chi, L\tau ))^n=(H(\chi,\tau))^n$ for any $n\ge 1$, which implies the result. 
\end{proof}

In order to state the result for the dynamic SSEP, let us recall the \emph{gamma distribution}, which is a two parameter family $\Gamma(a,b)$ of absolutely continuous probability distributions on $\R_{> 0}$ with densities
$$
p_{a,b}(x)=\frac{b^a x^{a-1} e^{-bx}}{\Gamma(a)},\qquad x> 0,\quad a,b>0.
$$
They are uniquely determined by their moments 
\begin{equation}\label{eq:gamma-moments}
\int_0^\infty x^m p_{a,b}(x)dx = {b^m}\,\frac{\Gamma(a+m)}{\Gamma(a)}=b^m (a)_m, \qquad m=0,1,2\dots,
\end{equation}
and have the scaling property that if $X\sim\Gamma(a,b)$ then $cX\sim \Gamma(a,bc)$. 

\begin{theorem}\label{th:SSEP} Consider the dynamic SSEP with the step initial condition $\HT(x,0)=\mathbf{1}_{x<0}\cdot |x|$ and with the dynamic parameter $\bar{\la}>0$ that may depend on the large parameter $L$: $\bar{\la}=\bar{\la}(L)$. In what follows $\chi$ and $\tau$ take arbitrary fixed values in $\R$ and $\R_{>0}$, respectively. 

\smallskip

\noindent \textbf{(i)} If $\bar{\la}(L)\cdot L^{-\frac 12}\to\infty$ as $L\to\infty$, then $\lim_{L\to\infty} L^{-\frac 12}\cdot \HT(L^{\frac 12}\chi, L\tau )= H(\chi,\tau)$ in probability with $H(\chi,\tau)$ from \eqref{eq:SSEP-limit-shape}, exactly as for the usual SSEP, cf. Proposition \ref{pr:SSEP}. 

\smallskip

\noindent \textbf{(ii)} If $\lim_{L\to\infty} \bar{\la}(L)\cdot L^{-\frac 12}= l\in (0,+\infty)$, then we have the convergence in probability
$$
\lim_{L\to\infty} L^{-\frac 12}\cdot \HT(L^{\frac 12}\chi, L\tau)=\sqrt{l\cdot H(\chi,\tau)+\left(\frac{\chi+l}{2}\right)^2}-\frac{\chi+l}{2}\,.
$$

\smallskip

\noindent \textbf{(iii)} If $\bar{\la}(L)\to\infty$ as $L\to\infty$ and $\lim_{L\to\infty} \bar{\la}(L)\cdot L^{-\frac 12}=0$, then we have the convergence in probability
$$
\lim_{L\to\infty} \frac{\HT(\sqrt{\bar{\la}(L)}L^{\frac 14}\cdot\chi, L\tau)}{\sqrt{\bar{\la}(L)}L^{\frac 14}}=\sqrt{\sqrt{\frac\tau\pi}+\left(\frac\chi 2\right)^2}-\frac{\chi}{2}\,.
$$

\smallskip

\noindent \textbf{(iv)} Assume that $\bar{\la}$ does not depend on $L$. For any fixed $\chi$ and $\tau$, let $Y_{\chi,\tau}$ be a gamma distributed random variable: $Y_{\chi,\tau}\sim \Gamma(a,b)$ with $a=\bar{\la}$, $b=\sqrt{\tau/\pi}$.
Then we have the following limit in distribution:
$$
\lim_{L\to \infty} L^{-\frac 14}\cdot \HT(L^{-\frac 14}\chi,L\tau)=\sqrt{Y_{\chi,\tau}+\left(\frac\chi 2\right)^2}-\frac{\chi}{2}\,.
$$ 
\end{theorem} 

\noindent\emph{Comments} 1. In all the four cases (i)-(iv), the limiting functions tend to 0 when $\chi\to+\infty$ and behave as $|x|$ when $\chi\to-\infty$, showing that there are almost no particles or almost no holes ($=$empty sites) in the corresponding regions (with `almost' holding on the chosen scale). This shows that the above statement captures the full nontrivial limiting profiles of the height function in the corresponding regimes. 

\smallskip

\noindent 2. The formulas for the limits look slightly nicer for the $s_x(t)=2\HT(x,t)+x$. With the same scaling of $x$ and $t$, the limit of $s_x(t)$ in case (i) is $S(\chi,\tau):=2\sqrt{\tau/\pi}\cdot e^{-\chi^2/4\tau}+\chi\cdot \erf(\chi/(2\sqrt{\tau}))$, where $\erf(\,\cdot\,)$ is the error function; in case (ii) the limit is $\sqrt{2l S(\chi,\tau)+\chi^2+l^2}-l$; in case (iii) it is 
$\sqrt{4\sqrt{\tau/\pi}+\chi^2}$; and in case (iv) it is $\sqrt{4 Y_{\chi,\tau}+\chi^2}$, where $4Y_{\chi,\tau}$ is a $\Gamma(a,b)$-distributed random variable itself with parameters $a=\bar{\la}$ and $b=4\sqrt{\tau/\pi}$. 

\smallskip

\noindent 3. The four regimes of Theorem \ref{th:SSEP} can be seen as descriptions of what is happening as time progresses to the dynamic SSEP started with the step initial condition and with a small but fixed dynamic parameter $\bar{\la}$. At first one sees no difference with the usual SSEP, with height function growing as $\sqrt{t}$ and the limit shape being exactly as for the usual SSEP. When $t$ becomes comparable to $\bar{\la}^{-2}$, the height function is still of size $\sqrt{t}$ but the limit shape is no longer the same as for the usual SSEP. As $t$ become significantly larger than $\bar{\la}^{-2}$, the growth of the height function slows down, until it finally reaches $t^{1/4}$, and at that point the profile in the first approximation becomes visibly random and the limit shape phenomenon disappears. 

\smallskip

\noindent 4. The tools we use for the proof of Theorem \ref{th:SSEP}, which are the moment method and Corollary \ref{cr:obs-SSEP}, can also be used to investigate the multi-point asymptotics of the dynamic SSEP as well as its fluctuations. For example, for the part (iv) the same arguments as in the proof below show that the dependence of the random variable $Y_{\chi,\tau}$ on $\chi$ and $\tau$ may be removed, and thus the whole height function profile at a given time is governed by a single gamma-distributed random variable. 
Such investigations are no doubt very interesting, but they go beyond the goals of the present paper, and we hope to return to them in a future one. 

\begin{proof}[Proof of Theorem \ref{th:SSEP}] Our starting point is Corollary \eqref{eq:obs-SSEP}, which we will use in conjunction with Proposition \ref{pr:SSEP} to obtain the asymptotics of the moments $\E (\o(x,t))^n$ in all the limit regimes. Once the asymptotics of $\o(x,t)$ is established, we will use \eqref{eq:obs-SSEP} to obtain the asymptotics of the height function. 

Observe that relations \eqref{eq:expect-SSEP} and \eqref{eq:expect-SSEP-usual} imply (setting $x_1=\dots=x_n=x$)
\begin{equation}\label{eq:SSEP-obs-equal}
\frac{1}{(\bar{\la})_n}\E\left[\prod_{k=0}^{n-1} 
\left(\o(x,t)-k(\bar{\la}+x)-k^2\right)\right]=\E_{\mathrm{usual\ SSEP}} \left[\prod_{k=0}^{n-1} \left({\HT(x,t)}-k\right)\right].
\end{equation}
This equality can be used to express $\E [(\o(x,t))^n]$ in terms of lower moments $\E (\o(x,t))^k$, $1\le k\le (n-1)$, and the moments $\E_{\mathrm{usual\ SSEP}} [(\HT(x,t))^m]$, $1\le m\le n$. The asymptotic behavior of the latter is provided by Proposition \ref{pr:SSEP}, namely, as $L\to\infty$, for any $m\ge 1$ we have
$$
\E_{\mathrm{usual\ SSEP}} [(\HT(x,t))^m]\sim \begin{cases} L^{m/2} (H(\chi,\tau))^m,& \mathrm{cases\ (i)\ and\  (ii)},\\ L^{m/2}(H(0,\tau))^m=L^{m/2}(\tau/\pi)^{m/2},& \mathrm{cases\ (iii)\ and\  (iv).}
\end{cases}
$$ 

An easy inductive (in $n$) argument then shows that, depending on the limit regime, we have

\smallskip

\noindent (i) $\E [(\o(x,t))^n]\sim (\bar{\la}(L))^{n}L^{n/2}(H(\chi,\tau))^n$, hence, $\lim_{L\to\infty} (\bar{\la}(L))^{-1} L^{-1/2}\o(x,t)=H(\chi,\tau)$;

\smallskip

\noindent (ii) $\E [(\o(x,t))^n]\sim L^{n}(l H(\chi,\tau))^n$, hence, $\lim_{L\to\infty} L^{-1}\o(x,t)=l H(\chi,\tau)$;

\smallskip

\noindent (iii) $\E [(\o(x,t))^n]\sim  (\bar{\la}(L))^{n}L^{n/2}(\tau/\pi)^{n/2}$, hence, $\lim_{L\to\infty} (\bar{\la}(L))^{-1} L^{-1/2}\o(x,t)=\sqrt{\tau/\pi}$;

\smallskip

\noindent (iv) $\E [(\o(x,t))^n]\sim L^{n/2}(\bar{\la})_n(\tau/\pi)^{n/2}$, hence, $\lim_{L\to\infty} L^{-1/2}\o(x,t)=Y$, where $Y$ is a $\Gamma_{a,b}$-distributed random variable with $a=\bar{\la}$, $b=\sqrt{\tau/\pi}$, cf. \eqref{eq:gamma-moments}.

\smallskip

It remains to solve $\o(x,t)=\HT(x,t)(\HT(x,t)+x+\bar{\la})$, cf. \eqref{eq:obs-SSEP}, for $\HT(x,t)$. Since $\o(x,t)>0$, only one root of this quadratic equation is positive, which gives
\begin{equation}\label{eq:height-via-obs}
\HT(x,t)=\sqrt{\bar{\la}\cdot \o(x,t)+\left(\frac{x+\bar{\la}}{2}\right)^2}-\frac{x+\bar{\la}}{2}=\frac{x+\bar{\la}}2\left(\left(1+\frac{4\o(x,t)}{(x+\bar{\la})^2}\right)^{\frac 12}-1\right)\,.
\end{equation}

Again, we proceed case by case. 

\smallskip

\noindent (i) As $\bar{\la}\gg x$ and $\bar{\la}^2\gg \o(x,t)$ with high probability, Taylor expanding the square root in the last expression of \eqref{eq:height-via-obs} gives $\HT(x,t)\sim (\lambda(L))^{-1}\o(x,t)$, which is the desired result. 

\smallskip

\noindent (ii) In this regime, $\bar{\la}^2$, $x^2$, and $\o(x,t)$ are all of order $L$, and normalizing the middle expression of \eqref{eq:height-via-obs} by $L^{1/2}$ gives the result. 

\smallskip

\noindent (iii) Now $\o(x,t)$ and $x^2$ are of the same order $\bar\la(L)L^{1/2}$, which is much larger than $\la(L)$; dividing the middle expression of \eqref{eq:height-via-obs} by $(\bar\la(L))^{1/2}L^{1/4}$ leads to the desired limiting behavior. 

\smallskip

\noindent (iv) Finally, in this case $\o(x,t)$ and $x^2$ are both of order $L^{1/2}$ while $\bar{\la}$ is finite, and dividing \eqref{eq:height-via-obs} by $L^{1/4}$ gives the result. 

The proof of Theorem \ref{th:SSEP} is complete. 
\end{proof}

\end{document}